\documentclass[11pt]{article}

\makeatletter
\newcommand{\vast}{\bBigg@{4}}
\newcommand{\Vast}{\bBigg@{5}}
\makeatother
\usepackage{scalerel,stackengine}
\stackMath
\newcommand\reallywidehat[1]{%
\savestack{\tmpbox}{\stretchto{%
  \scaleto{%
    \scalerel*[\widthof{\ensuremath{#1}}]{\kern.1pt\mathchar"0362\kern.1pt}%
    {\rule{0ex}{\textheight}}
  }{\textheight}%
}{2.4ex}}%
\stackon[-6.9pt]{#1}{\tmpbox}%
}

\usepackage{setspace}
\usepackage{sn-preamble} 
\usepackage{tikz}
\usepackage{verbatim}
\usepackage{cancel}
\usepackage{enumitem}
\usepackage{bm}
\usepackage{caption}
\usepackage{changepage}   

\newcommand{\ignore}[1]{}

\newcommand{\Est}{\mathrm{Est}}
\newcommand{\boldEst}{\mathbf{Est}}

\newcommand{\err}{\mathrm{err}}

\newcommand{\corr}{\mathrm{corr}}

\newcommand{\bgamma}{\boldsymbol{\gamma}}
\newcommand{\bd}{\mathbf{d}}

\newcommand{\opt}{\textsf{opt}}

\newcommand{\nonconst}{{\textsc{Non-Const}}}
\newcommand{\learnconjunction}{\textsc{Agnostically-Learn-Conjunction}}

\newcommand{\NormInf}{{\mathbf{NInf}}} 

\newcommand{\shnoise}{{\textsc{SharpNoise}}}
\newcommand{\fave}{f_{\text{ave}}}

\newcommand{\qtester}{\textsc{Quantum-Tolerant-Junta-Tester}}
\newcommand{\tester}{\textsc{Tolerant-Junta-Tester}}
\newcommand{\estninf}{\textsc{Estimate-Ninf}}
\newcommand{\findhighcoordsrec}{\textsc{Find-High-Level-Coordinates-Recursive}}
\newcommand{\findhighcoords}{\textsc{Find-High-Level-Coordinates}}
\newcommand{\refinecoords}{\textsc{Refine-Coordinates}}

\title{A Mysterious Connection Between Tolerant Junta Testing and Agnostically Learning Conjunctions}

\author{Xi Chen\footnote{Email: \texttt{xichen@cs.columbia.edu}}
\\ \textsl{Columbia University} \and
Shyamal Patel\thanks{Email: \texttt{shyamalpatelb@gmail.com}} \\ \textsl{Columbia University} \and Rocco A. Servedio\thanks{Email: \texttt{ras2105@columbia.edu}} \\ \textsl{Columbia University}}

\date{}

\begin{document}

\pagenumbering{gobble}

\maketitle

\begin{abstract}
The main conceptual contribution of this paper is identifying a previously unnoticed connection between 
two central problems in computational learning theory and property testing: 
\emph{agnostically learning conjunctions} and \emph{tolerantly testing juntas}.  Inspired by this connection, the main technical contribution is a pair of improved algorithms for these two problems.

In more detail, 

\begin{itemize}

\item We give a distribution-free algorithm for agnostically PAC learning conjunctions over $\bits^n$ that runs in time $2^{\tilde{O}(n^{1/3})}$, for constant excess error $\eps$.  This improves on the fastest previously published algorithm, which runs in time $2^{\tilde{O}(n^{1/2})}$ \cite{KKMS}. 

\item Building on the ideas in our agnostic conjunction learner and using significant additional technical ingredients, we give an adaptive tolerant testing algorithm for $k$-juntas that makes $2^{\tilde{O}(k^{1/3})}$ queries, for constant ``gap parameter'' $\eps$ between the ``near''  and ``far''  cases.  
This improves on the best previous results, due to \cite{ITW21,nadimpalli2024optimal}, which make $2^{\tilde{O}(\sqrt{k})}$ queries.  Since there is a known $2^{\tilde{\Omega}(\sqrt{k})}$ lower bound for \emph{non-adaptive} tolerant junta testers, our result shows that adaptive tolerant junta testing algorithms provably outperform non-adaptive ones.

\end{itemize}

\end{abstract}

\newpage

\setcounter{tocdepth}{3}
 \begin{spacing}{1.0}
 {\small
      \tableofcontents}
 \end{spacing}
  
 \newpage

\pagenumbering{arabic}

\section{Introduction} \label{sec:intro}

\begin{quote}
The American mystery deepens.

\hskip 2.8in~--- Don DeLillo,  \underline{White Noise\vspace{0.06cm}}

\end{quote}

A broad theme that has emerged across modern data analysis problems in various subfields of theoretical computer science is the importance of dealing with \emph{noise}.  In both computational learning theory and property testing, the first generation of models that were considered in the early days of the field were ``noise-free'':  
the original PAC learning model of Valiant \cite{Valiant:84} did not allow for the possibility of mislabeled examples, and the original property testing model for Boolean functions that was introduced by Blum, Luby and Rubinfeld \cite{BLR93} was about testing whether a function exactly has a given property of interest.  As these fields matured, though, the motivation for extending these original models became increasingly clear: since imperfect data is ubiquitous in real settings, it is natural to consider relaxed learning and testing models that allow for noise.  In learning theory, this led to widespread interest in the model of \emph{agnostic learning} \cite{Haussler:92,KSS94}, which may be viewed as a demanding version of PAC learning in the presence of noise that can affect both examples and labels.  In property testing, the model of \emph{tolerant testing} \cite{PRR06} was introduced and intensively studied, in which the goal is to determine whether a function is $\eps_1$-close to or $\eps_2$-far from having a property of interest.  Today, a large fraction of contemporary research in learning theory deals with agnostic learning, and likewise for research in Boolean function testing and the tolerant testing model.

Beyond the clear practical motivations for tackling noisy problems, from a theoretical perspective many fundamental problems that are relatively straightforward in the noiseless setting become significantly  richer and more challenging when we make the more realistic assumption that the functions that we are working with may be corrupted by noise.  
For example, in computational learning theory, PAC learning an unknown \emph{conjunction} when there is no noise is an easy exercise, whereas agnostically learning conjunctions is much harder (as we discuss in detail below).  
As a different example, in property testing the well-studied problem of testing whether an unknown Boolean function is exactly, versus far from, a \emph{junta} was first solved more than two decades ago and is now well understood \cite{FKRSS03,ChocklerGutfreund:04,Blais08,Blaisstoc09,Saglam18,CSTWX18jacm}, whereas, as we discuss in detail below, the \emph{tolerant} junta testing problem remains largely open, with significant gaps in our current understanding.

\medskip

\noindent {\bf Our work.}
In this paper we identify and explore a previously unnoticed connection between the two problems mentioned above:  efficient algorithms for \emph{agnostically learning conjunctions} and efficient\footnote{Since ``efficient'' in the learning context refers to \emph{running time}, whereas ``efficient'' in the testing context refers to \emph{query complexity}, such a connection may seem curious indeed.  However, we remark that since the known algorithms for agnostically learning conjunctions can be cast in the Statistical Query model of Kearns \cite{Kearns:98}, we can view these algorithmic results as having information-theoretic analogues, and thus aligning with the (information-theoretic) query complexity results for tolerant junta testing.} algorithms for \emph{tolerantly testing juntas}.  (See \Cref{sec:agnostic-prelim} and \Cref{sec:tolerant-definition} for formal definitions of these well-studied problems.)
We remark that conjunctions are one of the most fundamental classes to consider in the distribution-free agnostic learning setting; indeed, already in a FOCS 2003 tutorial Avrim Blum highlighted the importance of developing agnostic learning algorithms for conjunctions \cite{Blum:03tutorial}. Additionally, it is well known that sufficiently improved algorithms for agnostically learning conjunctions would imply faster algorithms for learning DNF formulas, one of the touchstone problems in computational learning theory (we discuss this more after stating \Cref{thm:agnostic}).  Similarly, junta testing has been a touchstone problem in the field of Boolean function property testing for more than twenty years since the influential work of Fischer et al.~\cite{FKRSS03,ChocklerGutfreund:04,Blais08,Blaisstoc09,BGMW13,Saglam18,STW15,CSTWX18jacm,Bshouty19}, and the problem of tolerant junta testing has been a particular focus of research effort in recent years \cite{BCELR19,DMN19,LW19,ITW21,PRW22,chen2023new,chen2024mildly,nadimpalli2024optimal}.

The connection that we identify between agnostically learning conjunctions and tolerantly testing juntas enables us to give new algorithms for both problems, improving on the best prior results.
Before describing our new results, we begin by highlighting some similarities between the prior state of the art for these two problems and the ideas and tools used to establish those results.

In terms of upper bounds, prior to the current work the fastest published algorithm for distribution-free agnostic learning of conjunctions over $\bits^n$ was the ``$L_1$ polynomial regression algorithm'' of \cite{KKMS}, which runs in time $2^{\wt{O}(\sqrt{n})}$ for agnostically learning conjunctions to constant excess error $\eps=\Theta(1)$ (i.e.~obtaining a hypothesis with error at most $\opt+\eps$). 
We remark that $2^{\wt{O}(\sqrt{n})}$ is also the best known (again, prior to the current work) query complexity for tolerantly testing $k$-juntas over $\bits^n$ when $k=\Theta(n)$ \cite{ITW21, nadimpalli2024optimal}.  
Furthermore, in both \cite{KKMS} and \cite{nadimpalli2024optimal} the fact that the $n$-variable AND function has approximate degree $O(\sqrt{n})$ plays an important role in obtaining the square-root savings.
Turning to lower bounds, we also find that there is an alignment between the current state of the art for these two problems. 
At the level of results, in \cite{feldman2012complete}, Feldman gave a $n^{\Omega(\log(1/\eps))}$ lower bound for agnostically learning conjunctions over $\bits^n$ to excess error $\eps$ under the uniform distribution in the Statistical Query (SQ) model, while Chen and Patel \cite{chen2023new} gave a $k^{\Omega(\log(1/\eps))}$-query lower bound for tolerantly testing juntas with an $\eps$-width gap between the ``near'' and ``far'' cases.
In terms of techniques underlying these results, the proof in \cite{feldman2012complete} crucially relies on the fact that a $k$-variable conjunction disagrees with the parity function over the same set of variables on $1/2 - 2^{-k}$ fraction of all inputs, while the argument of \cite{chen2023new} uses the related notion of $k$-wise independence.

At this point, perhaps the reader is wondering whether the quantitative correspondences between results that are highlighted above, as well as the similarities in the technical ingredients, may be merely a coincidence.  To support our claim that there is a meaningful connection between these two problems, we discuss some further points of correspondence below.  This discussion is presented in the form of a brief dialogue, taking place over some wine,
in which a BELIEVER tries to increase the confidence of a SKEPTIC in the claimed connection:

\begin{itemize}

\item [SKEPTIC]~~~Agnostic learning and tolerant testing seem very different, for the following reason:  When we are doing agnostic learning, there is an arbitrary distribution ${\cal D}$ over the example space $\bits^n$.  In contrast, in property testing (tolerant or otherwise), we measure the distance between two functions $f$ and $g$ with respect to the uniform distribution over $\zo^n,$ which is very different from an arbitrary distribution.  So these frameworks do not seem to align with each other. 
\emph{``I honestly tell you, I cannot believe it.''} \cite{Melville:confidence}

\item [BELIEVER]~~~\emph{``I told you, you must have confidence''} \cite{Melville:confidence} The alignment exists because of the specific learning problem (conjunctions) and testing problem (juntas) being considered.  The relevant distribution, for the tolerant junta testing problem, is not the uniform distribution over $\bits^n$ but rather the \emph{spectral sample} distribution ${\cal P}_f$, over all subsets of $[n]$, induced by the Fourier coefficients of $f: \bits^n \to \bits$ (see \Cref{def:spectral-sample}); the distribution ${\cal P}_f$, not the uniform distribution, is the appropriate analogue of the arbitrary distribution ${\cal D}$ in the agnostic learning problem.  To see this, observe that for the tolerant junta testing problem, the goal can (roughly speaking) be reframed as finding the ``right'' subset of variables (subject to being of size $k$) that ``cover'' the largest possible amount of the spectral sample distribution (see the discussion of ``junta correlation'' in \Cref{sec:Fourier-basics}).  Similarly, for the agnostic conjunction problem, the goal can be reframed (again roughly speaking) as finding the ``right'' subset of variables that ``cover'' the maximum possible amount of the  distribution of the positive examples (subject to not covering too much of the distribution of negative examples).
\emph{``Confidence restored?$\dots$Sit down, sir, I beg, and take some of this wine.''} \cite{Melville:confidence}

\item [SKEPTIC]~~~\emph{``Ah, wine is good, and confidence is good; but can wine or confidence percolate down through all the stony strata of hard considerations, and drop warmly and ruddily into the cold cave of truth?''} \cite{Melville:confidence}
I still have my doubts.
The junta problem involves testing whether an $n$-variable function is a $k$-junta:   there are inherently two parameters, $n$ and $k$.  But in agnostic learning of conjunctions there is no analogue of the ``$k$'' parameter in junta testing; the conjunction can be of arbitrary size relative to $n$.  How can your claimed connection overcome this?  \emph{``But think of the obstacles!''} \cite{Melville:confidence}

\item [BELIEVER]~~~\emph{``I have confidence to remove obstacles, though mountains.''} \cite{Melville:confidence}
The papers of \cite{DMN19} and \cite{ITW21} have developed powerful techniques, namely highly efficient simulations of ``coordinate oracles'' (see \Cref{sec:coordinate-oracles}), which can be used in the tolerant junta testing problem to effectively reduce the number of coordinates from $n$ down to $k'=O_\eps(k)$. So to answer your question, we do not introduce a ``$k$'' parameter into the conjunction setting; rather, in the tolerant junta setting, we can effectively get rid of the ``$n$'' parameter and replace it by $k'=O_\eps(k)$.  A subset of $k \leq O_\eps(k)$ relevant variables for the junta in the tolerant testing scenario is not so different from an arbitrary subset of relevant variables for the conjunction in the agnostic learning scenario.
\emph{``But again I say, you must have confidence.''} \cite{Melville:confidence}
\end{itemize}

Leaving the SKEPTIC and the BELIEVER to their dialogue and their wine, 
we remark that the connection between these two problems will be fleshed out in significantly more detail later, when we give an overview of our agnostic learning algorithm and our tolerant junta testing algorithm and describe the common underlying ideas and similarities between the two algorithms; see \Cref{sec:overview-quantum}.  
We turn to a description of our technical results.

\subsection{Our results:  Better agnostic conjunction learning algorithms and tolerant junta testing algorithms.}

In this work we extend the connection between these two problems, by (i) giving a faster algorithm for agnostically learning conjunctions, and (ii) using the ideas in that algorithm to obtain an improved tolerant junta testing algorithm with an analogous query complexity.
Regarding (i), we prove the following:

\begin{theorem}
\label{thm:agnostic}
\Cref{alg:agnostic} agnostically learns conjunctions over $\bits^n$  to excess error $\eps$ in time $2^{n^{1/3} \cdot \polylog(n,1/\eps)}$.
\end{theorem}

We remark that Diakonikolas, Kane, and Ren have very recently achieved a similar result to \Cref{thm:agnostic} \cite{diakonikolas2025}.
 
It is well known that any algorithm for distribution-free agnostically learning conjunctions to excess error $\eps$ can be used, taking $\eps = 1/(10s)$, in tandem with a boosting algorithm (e.g.~\cite{Schapire:90} or \cite{Freund:95}) to learn $s$-term DNF, or even the richer class of total-integer-weight-$s$ linear threshold functions over conjunctions.  This is simply because thanks to the ``discriminator lemma'' \cite{HMP+:87}, given any weight-$s$-LTF-over-conjunctions target function, for any distribution some conjunction must have error at most $1/2 - 1/s$ in predicting the target function; so using an agnostic learner, we recover a weak hypothesis with error $1/2 - 1/(2s)$, which can be used for boosting.
Plugging in \Cref{thm:agnostic} as the agnostic conjunction learner, we get a DNF learning algorithm running in time $2^{n^{1/3} \cdot \polylog(n,s)}$, which essentially matches the state of the art algorithm for this well-studied problem \cite{KlivansServedio:04jcss}.

Inspired by \Cref{alg:agnostic} and its analysis, we give the following warm-up result for our main tolerant testing algorithm (in the theorem below, $\calJ_k$ is the class of all Boolean-valued $k$-juntas):

\begin{theorem}
\label{thm:quantum}
\Cref{alg:quantum} $\pm \eps$-accurately estimates $\dist(f, \calJ_k),$ for $f: \bits^{k'} \to \bits,$ where $k' = \poly(k,1/\eps)$, and makes $2^{k^{1/3} \cdot \polylog(k,1/\eps)}$ quantum queries to $f$.
\end{theorem}

\Cref{thm:quantum} is stated in the language of ``distance estimation,'' i.e.~estimating the distance $\dist(f,\calJ_k) := \min_{g \in \calJ_k}\Pr_{\bx \sim \bits^{k'}}[f(\bx) \neq g(\bx)]$ between $f$ and the closest $k$-junta.  It is well known (see e.g.~\cite{PRR06}) that the problem of $\pm \eps$-accurate distance estimation for a class of functions is, up to small factors, equivalent to the problem of $(\eps_1,\eps_2=\eps_1+\eps)$-tolerant testing for that class (see \Cref{sec:tolerant-definition}). Thus an equivalent statement of \Cref{thm:quantum} in the language of tolerant testing is the following:  for $k'=\poly(k,1/\eps)$, there is a quantum query algorithm that makes $2^{k^{1/3} \cdot \polylog(k,1/\eps)}$ quantum queries to $f: \bits^{k'} \to \bits$ and $(\eps_1,\eps_2)$-tolerantly tests $f$ for the property of being a $k$-junta, where $\eps=\eps_2-\eps_1.$

Building on the ideas and ingredients in \Cref{thm:quantum}, our main tolerant testing result is the following:

\begin{theorem}
\label{thm:classical}
For any $n$, \Cref{alg:classical} $\pm \eps$-accurately estimates $\dist(f, \calJ_k),$ for $f: \bits^{n} \to \bits,$ and makes $2^{k^{1/3} \cdot \polylog(k,1/\eps)}$ (classical) queries to $f$.
\end{theorem}

As with \Cref{thm:quantum}, an equivalent statement of \Cref{thm:classical} in the language of tolerant testing is that there is a (classical) algorithm that makes $2^{k^{1/3} \cdot \polylog(k,1/\eps)}$ black-box queries to $f: \bits^{n} \to \bits$ and $(\eps_1,\eps_2)$-tolerantly tests $f$ for the property of being a $k$-junta, where $\eps=\eps_2-\eps_1.$

Building on \cite{chen2024mildly}, in \cite{nadimpalli2024optimal} Nadimpalli and Patel showed that for constant $\eps$, any $\pm \eps$-accurate algorithm for estimating $\dist(f,\calJ_k)$ must make $2^{\tilde{\Omega}(\sqrt{k})}$ many queries to $f$. In contrast, our adaptive algorithm requires only $2^{\tilde{O}(k^{1/3})}$ queries, and so we can conclude that adaptive algorithms provably outperform non-adaptive ones for tolerant junta testing.  To the best of our knowledge, this gives the first super-polynomial separation between adaptive and non-adaptive algorithms for a natural tolerant Boolean function property testing problem.

\subsection{Discussion}

We provide some discussion of our results below.

\begin{remark} [On the quantum nature of \Cref{alg:quantum}]
\label{rem:quantum}
The only reason that \Cref{alg:quantum} requires a quantum black-box oracle to $f$ is because it makes draws from the ``spectral sample'' of $f$ (see \Cref{def:spectral-sample}) in lines~2 and 3(b); it is easily verified that given these draws, the rest of the algorithm is entirely classical.  It is well known (see e.g.~Chapter~5 of \cite{NielsenChuang16}) that a quantum algorithm with black-box quantum oracle access to a $\bits$-valued function $f$ can make a draw from the spectral sample of $f$ by performing a single quantum oracle call to $f$; this is why we refer to \Cref{alg:quantum} as a quantum query algorithm.
(In contrast, an argument similar to the classical lower bound for Simon's problem \cite{Simon97} shows that no classical algorithm can even approximately make draws from the spectral sample of an unknown and arbitrary Boolean function $f.$)
   \Cref{thm:quantum} could alternatively be formulated as a statement about purely classical algorithms which, in addition to having black-box oracle access to $f$, are equipped with the ability to make draws from the spectral sample.
\end{remark}

\begin{remark} [Why \Cref{thm:quantum} given \Cref{thm:classical}?]
\label{rem:comparing}
\Cref{thm:classical} is strictly stronger than \Cref{thm:quantum} in two senses: first, it only requires classical oracle access to $f$ rather than quantum query access, and second, it holds for general $n$-variable functions $f$ for any $n$ whereas \Cref{thm:quantum} assumes that $k'=\poly(k,1/\eps)$. 
Thus the reader may be wondering why we have mentioned  \Cref{thm:quantum} at all; this is done for several reasons.  First, the proof of \Cref{thm:classical} builds on a number of ideas and ingredients that are introduced in the simpler context of \Cref{thm:quantum}; \Cref{thm:quantum} should be viewed as an intermediate result on the path to \Cref{thm:classical}.  Second, the connection between agnostic conjunction learning and tolerant junta testing is clearest and most evident in the technically simpler setting of \Cref{thm:quantum}, cf.~the discussion and comparison of \Cref{alg:agnostic} versus \Cref{alg:quantum} in \Cref{sec:overview-quantum}. (Intuitively, this is because, as suggested earlier by the BELIEVER, making a draw from the spectral sample distribution of $f: \bits^{k'} \to \bits$ --- which is possible in the quantum but not the classical setting --- corresponds to drawing an $n$-bit string from the unknown and arbitrary distribution over $\bits^n$ in the agnostic learning scenario. In the classical tolerant junta testing setting, we have to do additional work to compensate for not being able to make draws from the spectral sample.)
\end{remark}

\begin{remark} [On the assumption that the number of variables in \Cref{thm:quantum} is $\poly(k,1/\eps)$] \label{rem:kprime-assumption} The assumption that the number of variables for $f$ is $k' = \poly(k,1/\eps)$ rather than $n$ is made for the sake of simplicity.  As the BELIEVER mentioned, dealing with the general $n$-variable scenario requires the use of ``coordinate oracles'' to effectively reduce the number of variables down to $\poly(k,1/\eps)$ (see \Cref{sec:coordinate-oracles}), and integrating these coordinate oracles with quantum queries to efficiently make draws from the spectral sample seems to introduce non-trivial complications. We bypass these complications by simply assuming that $k'=\poly(k,1/\eps)$. Again, we stress that our main (classical) tolerant property testing result, \Cref{thm:classical}, is not limited in any way by this assumption, which holds only for the warm-up quantum result given in \Cref{thm:quantum}.
\end{remark}


\section{Technical Overview} \label{sec:technical}

\subsection{Overview for agnostically learning conjunctions} \label{sec:overview-agnostic}

Before sketching our approach, we remark that prior work has identified limitations of the \cite{KKMS} approach to agnostic learning.
In fact, more than a decade before \cite{KKMS}, Paturi \cite{Paturi:92} showed that real polynomials require degree $\Omega(\sqrt{n})$ to pointwise approximate the conjunction $x_1 \wedge \cdots \wedge x_n$ over $\bn$, 
which implies that a direct application of \Cref{thm:KKMS} to agnostically learn conjunctions must use degree $d=\Omega(\sqrt{n})$ and have running time $n^{\Omega(\sqrt{n})}$.  A more substantial barrier was identified by Klivans and Sherstov \cite{KS:07}, who showed that for \emph{any fixed set} of real-valued functions (not just low-degree monomials), if all $n$-variable conjunctions can be pointwise approximated as a linear combination of those features then the set must have $2^{\Omega(\sqrt{n})}$ elements.  This result is interpreted in \cite{KS:07} as ``show[ing] the fundamental limitations of this approximation-based paradigm''.  
More recently, Gollakota et al.~\cite{GKK20} showed that any \emph{Correlational Statistical Query (CSQ)} algorithm for agnostically learning conjunctions must have complexity $2^{\Omega(\sqrt{n})}.$

We circumvent the limitations defined by \cite{KS:07} by using an approximation that is not a pointwise good estimator of the conjunction.\footnote{We further remark that our algorithm evades the \cite{GKK20} lower bound because it is not a CSQ algorithm. Recall that a CSQ algorithm can only obtain estimates of $\Pr_{(\bx,\by) \sim {\cal D}}[\bh(\bx) \neq \by]$; the way that our algorithm uses the event $E_{\vec{\ba}},$ which is defined using a fixed sample of examples and does not involve the labels at all, does not fit into this framework.} Indeed as we'll see shortly, we will define a ``ball distribution'' on which we want a pointwise approximation of the target conjunction, but outside of the support of this distribution, our approximation can (and provably must) make errors.

At a high level, our approach proceeds as follows.\footnote{The tags (i), (ii), etc.~in what follows are for later reference in \Cref{sec:overview-quantum}, when we are explaining points of correspondence between our quantum tolerant junta testing algorithm and our agnostic conjunction learner. } We (i) draw a small number $m = n^{1/3}$ of positively labeled examples $\ba^1,\dots,\ba^m$ from the distribution $\calD$ over $\bits^n \to \bits$, and use them to adaptively design a new distribution $\calD'$, which we call a ``ball distribution'' and is
supported on a carefully chosen subset of $\bits^n$.  The distribution $\calD'$ is  crafted in such a way that there exists a low-degree polynomial which is a small-squared-error approximator of $c^*$ under $\calD'$, where (ii) $c^*$ is the conjunction which is optimal under $\calD$.  \Cref{thm:KKMS} can then be used to find a high-accuracy hypothesis $h'$ with respect to $\calD'$, and this $h'$ is then used to give a high-accuracy hypothesis $h$ under the original distribution $\calD.$

\begin{remark} \label{remark:oversimplification}
The above sketch is an oversimplification: in particular, the draw of a small number of examples in the first step alluded to above (iii) only succeeds in obtaining a ``useful'' set of examples with some non-negligible probability (which is roughly $\eps^{n^{1/3}}$, see \Cref{commonlemma1}).  The actual algorithm (iv) performs many independent repetitions of the approach sketched in the previous paragraph to obtain multiple hypotheses, and does a final stage of hypothesis testing to select one final hypothesis.  In the rest of this overview we ignore this ``outer loop,'' the analysis of which is standard and straightforward, and give a more detailed overview of the approach sketched in the previous paragraph.
\end{remark}

We describe our approach in a bit more detail. Let  $c^*$ be the (unknown) ``target conjunction'' which (v) satisfies $\err_\calD(c^*)=\opt_\calD.$  (vi) The distribution $\calD'$ alluded to above is $\calD' := \calD|E$ where $E$ is an event over $X=\bits^n$.  The event $E=E_{\vec{\ba}}$ is based on the initial draw of examples $\vec{\ba}=(\ba^1,\dots,\ba^m)$ from $\calD$ alluded to above; we remark that $E$ is ``easily recognizable'' in the sense that given an input $x \in X$ it is easy to determine whether or not $x$ satisfies $E$.
The distribution $\calD'=\calD|E$ is referred to as a ``ball distribution'' because (vii) it is supported on points which (in a certain set of coordinates based on $\ba^1,\dots,\ba^m$) lie in a certain Hamming ball of radius $n^{2/3}$.
Assuming that the initial draw of examples $\vec{\ba}$ is such that the event $E=E_{\vec{\ba}}$ is ``useful'' (see \Cref{def:useful}; ensuring that this happens is the purpose of the outer loop), we show that the distribution $\calD'$ has the following crucial property:
\begin{itemize}

\item The target conjunction $c^*$ has approximate degree at most $O(n^{1/3} \cdot \log(1/\eps))$ over all inputs in the support of $\calD'$ (see \Cref{lem:apx-deg}).

\end{itemize}

(We emphasize that the approximate degree bound given above is crucial to our algorithm's achieving an overall $2^{n^{1/3} \cdot \polylog(n,1/\eps)}$ running time.)
As alluded to earlier, we use \Cref{thm:KKMS} on examples drawn from $\calD'$ to obtain a hypothesis $h'$ which has $\err_{\calD'}(h) \leq \err_{\calD'}(c^*)+\eps$.  We take the hypothesis $h$ to be
\begin{equation} \label{def:h'}
h(x) = \begin{cases}
		h'(x) & \text{if $E(x)$ holds} \\ -1 \text{~(i.e.~false)} & \text{if $\neg E(x)$ holds}
	\end{cases}
\end{equation}
(note that $h$ is easy to evaluate since $E$ is ``easily recognizable'').
It is not too difficult to show that $\opt_{\calD}(h) \leq \opt_{\calD} + 2\eps$, as desired; this is done in \Cref{lem:finalconj}.

\subsection{Overview of quantum tolerant junta testing} \label{sec:overview-quantum}

We turn to give an overview of our quantum tolerant $k$-junta testing algorithm, \Cref{alg:quantum}, highlighting the points of correspondence with \Cref{alg:agnostic}.
As mentioned earlier, we consider the quantum $k$-junta $\pm \eps$-distance-estimation problem over $\bits^{k'}$ where $k'=\poly(k,1/\eps)$. 
For simplicity, we mostly ignore the $\eps$-dependence in the following discussion.

We begin by highlighting some correspondences between various mathematical objects in \Cref{alg:agnostic} and its analysis, and their counterparts in \Cref{alg:quantum} and its analysis.
First, we remind the reader that $k'$ (the number of variables in the testing setting) aligns with $n$, the dimension of the agnostic learning problem.  
Next, we recall the analogy between the arbitrary distribution ${\cal D}$ over labeled examples in $\bits^n \times \bits$ in the agnostic learning setting (or more precisely its marginal over $\bits^n$, the domain of the examples) and the ``spectral sample'' distribution ${\cal P}_f$ over $\bits^{k'}$ that is induced by the Fourier spectrum of $f: \bits^{k'} \to \bits$. 
The set of coordinates corresponding to the optimal ``target conjunction'' $c^*$ in the agnostic setting (recall (ii) in the previous subsection) corresponds to the set of $k$ relevant variables, which we call 
$U^\star \subset [k']$, for the $k$-junta that is closest to $f$. Recalling (v) in the previous subsection, the value of $\opt_{\cal D}$ in the agnostic setting, roughly speaking, corresponds to the \emph{junta correlation} of $f$, which is closely tied to the amount of Fourier weight that $f$ puts on sets containing only coordinates in $U^*$ (see \Cref{eq:corr}).

Corresponding to (iv) in the previous section, our quantum junta distance estimation algorithm performs many  independent repetitions of a main loop, because each iteration of the loop only has a small but non-negligible success probability of $\eps^{k^{1/3}}$, see \Cref{finalquantumclaim} (this is analogous to the $\eps^{n^{1/3}}$ success probability of each repetition described in (iv) in the previous subsection). In the discussion below, similar to the previous subsection, we focus on the case of a successful iteration of the loop.

Within each repetition of the loop, the quantum tester makes $k^{1/3}$ draws (the sets $\bA_1,\dots,\bA_m$ in line~2(a) of \Cref{alg:quantum}) from the spectral sample distribution ${\cal P}_f$ over $\bits^k$, analogous to the draw of $m=n^{1/3}$ examples $\ba^1,\dots,\ba^m$ from the distribution ${\cal D}$ alluded to in (i) of the previous subsection. (Recall that as mentioned in \Cref{rem:quantum}, drawing from the spectral sample distribution is the only quantum aspect of our tester.)
A successful iteration of the loop corresponds to having the set $\bC := \bA_1 \cup \cdots \cup \bA_{k^{1/3}}$ be such that the spectral sample distribution, restricted to subsets of $U^\star$, only puts a small amount of Fourier weight on sets that contain many (more than $k^{2/3}$) elements outside of $\bC$.  This notion of a successful iteration is closely analogous to (iii) the agnostic learning algorithm obtaining a ``useful'' set of examples, as is evident from the similarity in both the quantitative bounds and proofs of \Cref{commonlemma1} and \Cref{finalquantumclaim}. 
The spectral sample distribution ${\cal P}$ restricted to subsets $S$ of $[k']$ that have $|S \setminus \bC| \leq k^{2/3}$ is analogous to (vii), the ``ball distribution'' ${\cal D}'={\cal D}|E_{\vec{\ba}}$ of the agnostic conjunction learner.

\begin{remark}
We quickly remark that in the context of tolerant junta testing, the above approach
was essentially considered in \cite{ITW21}, albeit with different parameters. In particular, they perform the same procedure using $\sqrt{k}$ draws from the spectral sample to get a $2^{\wt{O}(\sqrt{k})}$ time quantum tolerant $k$-junta testing algorithm in this setting. (They also show that there is a classical tester with the same query complexity, which makes draws from a distribution that is related to the spectral sample, cf.~the discussion of the ``normalized influence distribution'' in the next subsection.) However, as we'll see shortly we will need to bring in several new tools in order to improve the query complexity 
to $2^{\wt{O}(k^{1/3})}$.
\end{remark}

\begin{remark}
\label{rem:cookie}
We comment that while we first developed our agnostic learning algorithm and then ported over the underlying ideas to the tolerant junta testing problem, it seems equally natural, given the  results and techniques of \cite{ITW21}, to go in the opposite direction and to give a $2^{\tilde{O}(\sqrt{n})}$-time agnostic conjunction learning algorithm without \cite{KKMS} inspired by their approach.  Alternatively, such a ``\cite{KKMS}-free'' agnostic learning algorithm could be constructed by changing the parameters of our ball distribution and using a brute-force approach over the $2^{\tilde{O}(\sqrt{n})}$-size support of that distribution. Our $2^{\tilde{O}(n^{1/3})}$-time agnostic learning algorithm can be viewed as the result of speeding up this brute-force search using \cite{KKMS}.
\end{remark}

We return to a discussion of our quantum tolerant tester.  While we have established many points of analogy between the agnostic conjunction learning setting and the tolerant junta testing setting in the discussion so far, from a technical perspective the two proofs now diverge. In particular, we would like to modify $f$ so as to zero out all of its Fourier coefficients $S$ with $|S \setminus \bC| > k^{2/3}$ (in analogy (cf.~(vi) and (vii)) with restricting our attention to $\calD'$ in the agnostic learning setting). While the standard approach to do this would be to suitably apply noise to the function, this will not decrease the weight enough. Indeed, we will require a very sharp attenuation of the Fourier weight for our algorithmic tools, described later, to work; in particular, there must be at most $2^{-\wt{O}(k^{1/3})}$ Fourier weight on the coefficients $S$ with $|S \setminus \bC| \geq k^{2/3}$. To achieve this, we design a novel $\shnoise^{\overline{\bC}}$ operator, which we introduce and analyze in \Cref{sec:sharpnoise}, to define a function $f^{\overline{\bC}}$ (see line~2(b) of \Cref{alg:quantum}). This operator should be thought of as ``zeroing out'' all of the Fourier coefficients of $f$ that have more than $k^{2/3}$ elements outside of $\bC$ and leaving all other coefficients roughly untouched.  (In reality $\shnoise$ only approximately achieves this up to some error which we show is manageable and which we gloss over in the rest of this high-level discussion.)
Thus the Fourier coefficients of $f^{\overline{\bC}}$ are closely analogous to the ``ball distribution'' $\calD'=\calD|E_{\vec{\ba}}$ that was mentioned in item (vi) of the previous subsection, and the Hamming ball of radius $n^{2/3}$ described in (vii) of the previous subsection is analogous to the condition that Fourier coefficients that survive into $f^{\overline{\bC}}$ have no more than $k^{2/3}$ elements outside of $\overline{\bT}$.

Given our attenuated function $f^{\overline{C}}$, we now wish to estimate its junta correlation. This again marks a point of (now algorithmic) divergence from the agnostic learning problem (although conceptually the high-level mathematical structures remain in lockstep).
As was already mentioned, for agnostic conjunction learning the main algorithmic workhorse (see line~2(c) of \Cref{alg:agnostic}) is the $L_1$ regression algorithm of \cite{KKMS}.
In contrast, for tolerant junta testing the chief algorithmic ingredient that we employ is the ``local estimator'' notion that was introduced in the recent work of Nadimpalli and Patel \cite{nadimpalli2024optimal}.  Our quantum algorithm relies on a novel analysis of the \cite{nadimpalli2024optimal} local estimator for the mean of a function that has small Fourier weight at high levels (see \Cref{lem:smooth-estimator}). We use this local mean estimator to efficiently estimate the correlation between $f$ and the best junta over a subset $U$ of variables, for various different sets $U$ (see \Cref{lem:every-set-accurate}); crucially, using local estimators allows us to make relatively few (only $\approx \exp(k^{1/3})$) queries even though the number of sets for which we do this estimation may be as large as $\approx \exp(k)$ (see the discussion immediately preceding \Cref{lem:juntaestlemma}). 

Finally, we remark  there is some delicate and careful technical work required to combine the technology of local estimators with the $\shnoise$ operator; see \Cref{sec:smooth} through \Cref{sec:local-est-junta-corr}.

\subsection{Overview of classical tolerant junta testing} \label{sec:overview-classical}

Finally, we discuss our classical tolerant tester. For simplicity, we will assume that we are working in a similar setting to the quantum algorithm. Namely, we will think of $\eps$ as a small constant, say $0.01$, and mostly ignore the $\eps$ parameter in our discussions. Moreover we will assume that the function $f$ is over $k' = 2k$ variables rather than $n$ to avoid discussion of coordinate oracles. We also again denote the best set of $k$ junta coordinates by $U^\star$.

At a very high level, we would like to follow the same path as our quantum testing algorithm. Fortunately, many of the tools we developed there, such as
\shnoise{} and our local estimators, are already immediately implementable by a classical algorithm. That said, the central piece that will not carry over is our ability to generate a set $\bC \subseteq U^\star$
such that the spectral sample restricted to subsets of $U^\star$ only puts a small amount of Fourier weight on sets that contain more than $k^{2/3}$ many elements outside of $\bC$. This brings us to our central technical challenge in this section: Assuming that
	\[\sum_{S \subseteq U^\star: |S| \geq k^{2/3}} \wh{f}^2(S) \text{~is large (say, at least $\Omega(1)$),}
	\]
	we are tasked with finding a set of coordinates $C \subseteq U^\star$ such that
	\[\sum_{S \subseteq U^\star: |S \setminus C| \geq k^{2/3}} \wh{f}^2(S) \text{~is small (say, at most some $o(1)$).}
	\]
We start by noting that there is no hope of mimicking the quantum algorithm directly. Indeed, as alluded to earlier, any classical algorithm must make $2^{\Omega(k)}$ queries to $f$ to make a draw from a distribution with small total variation distance from the the spectral sample.
Instead, we design an entirely new approach, which we describe below.

At a very high level, we give a procedure, called \refinecoords{} (\Cref{alg:refine-coordinates}), that takes as input a set of coordinates $C \subseteq [k']$ and outputs a collection of $2^{k^{1/3} \polylog(k)}$ sets $C' \subseteq C$ such that for at least one set $C^\star$ that we output,
 \begin{itemize}
 
 \item [(a)] $C^\star$ has half as many irrelevant variables as $C$, i.e.~$|C^\star \cap \overline{U^\star}| \leq \frac{1}{2} |C \cap \overline{U^\star}|$; and moreover,
 \item [(b)] 
Only a very small amount of the Fourier mass of $f$ is on terms $S \subseteq U^\star$ with $|S \setminus C'| \geq k^{2/3}$, i.e.
\begin{equation}
\label{eq:blop}
	\sum_{S \subseteq U^\star: |S \setminus C^\star| \geq k^{2/3}} \wh{f}^2(S) \leq \frac{1}{k}.
\end{equation}
\end{itemize}

\refinecoords{} is useful for us for the following reason:  Given such a procedure, we can can run it on $C := [k']$ to construct a family of sets $C'$. For each set $C'$ in this family, we can then run \refinecoords$(f,C')$ to get a new family. After doing this $\log(k')$ times, we will eventually end up with a family of $2^{k^{1/3} \polylog(k)}$ sets, and one of the sets (call it $\bC^{\star\star}$) will satisfy 
		\[\sum_{S \subseteq U^\star: |S \setminus \bC^{\star\star}| \geq k^{2/3}} \wh{f}^2(S) \leq o(1) \]
and $\bC^{\star\star} \subseteq U^\star$ as desired. Roughly speaking, we can then conclude by applying the \shnoise{} operator and our local estimators as in the quantum tester.

With this performance goal for \refinecoords{} in mind, we turn to a discussion of how to achieve it. We do this by using \emph{normalized influences}, akin to \cite{ITW21}. These are defined in detail in \Cref{sec:normalized-influences}. For us, the crucial property, as we will see later, is that the  normalized influence of a set $U$ of size $\kappa$, denoted $\NormInf_U[f]$, is defined in such a way that drawing a set $\bT$ of size $\kappa$ at random so that $\Pr[\bT=U]$ is proportional to $\NormInf_U[f],$ corresponds to the following process:

\begin{quote}
$(\star)$ First draw a set $\bS$ from the the spectral sample of $f$, conditioned on $|\bS| \geq \kappa$, and then choose $\bT$ as a uniform random subset of $\bS$ of size $\kappa$.
\end{quote}

To actually get access to these normalized influences, we give a new algorithm, \Cref{alg:estimate-ninf}, that approximates $\NormInf_U[f]$ to arbitrary accuracy. Using this algorithm, we can sample from a distribution with total variation distance at most $\eta$ to the true distribution of level-$\kappa$ normalized influences in time $\poly( \exp(\wt{O}(\kappa),\eta^{-1})$.\footnote{We believe that this is near optimal. In particular, we believe that drawing from the level-$\kappa$ normalized influence distribution requires $2^{\Omega(\kappa)}$ black-box queries to $f$, as draws from this distribution seem useful for solving Simon's problem for a random function on $2\kappa$ variables.} Notably, this improves on a prior result of \cite{ITW21}, which could only compute $\NormInf_U[f]$ up to a constant multiplicative factor.
Given the above algorithmic result, we will assume throughout the remainder of our discussion that we can in fact draw samples  from the level-$k^{1/3} \polylog(k)$ normalized influence distribution.  

With this tool in hand, let us discuss how we can give a procedure \refinecoords{} with the desired properties. As a starting point for the algorithm, we guess a random set $\bI \subseteq C$ of $k^{1/3} \polylog(k)$ coordinates. The point of doing this is to collect a set of $k^{1/3} \polylog(k)$ coordinates which are hopefully irrelevant, i.e.~lie outside of $U^\star$, such that the function resulting from averaging out these coordinates, which we denote $f_{\ave}^\bI$, satisfies
	\begin{equation}
	\label{eq:beebop}
		\sum_{S \subseteq [k']: |S \setminus U^\star| \geq k^{2/3}/\polylog(k)} \wh{f_{\ave}^{\bI}}^2(S) \leq \frac{1}{k^{10}}.
	\end{equation}
	In particular, note that the above Fourier attenuation is precisely what we would expect to get if $\bI$ were drawn uniformly at random from $C \cap \overline{U^\star}$. Because $\bI \subseteq C \cap \overline{U^\star}$ with probability roughly $2^{-k^{1/3} \polylog(k)}$, we expect to sample a ``good'' $\bI$ after randomly sampling $2^{k^{1/3} \polylog(k)}$ sets $\bI$ from $C$ uniformly at random.

Assuming that we sample an $\bI$ that satisfies \Cref{eq:beebop}, we now turn to describe how to use this to construct a set $C'$. (Note that we will construct one set $C'$ per set $\bI$, resulting in a total of $2^{k^{1/3} \polylog(k)}$ sets $C'$.) Given a candidate $\bI$, we build $\bC'$ iteratively, starting from $\bC'=\emptyset$, roughly as follows: (1) If $f_{\ave}^{\bI}$ has at most $\frac{1}{\poly(k)}$ Fourier mass on terms $S$ with $|S \setminus \bC'| \geq k^{2/3}$, output $\bC'$. Otherwise, (2) draw a random set $\bT$ from the normalized influences of $f_{\ave}^\bI$ conditioned on the set $\bS$ from $(\star)$ satisfying $|\bS \setminus \bC'| \geq k^{2/3}$. (These two steps can roughly be performed by measuring the variance and sampling from the normalized influences of a suitable application of \shnoise{} to $f_{\ave}^\bI$.) Afterwards, update $\bC' \gets \bC' \cup (\bT \cap C)$ and repeat.

	We note that condition (1) of the algorithm sketched above will roughly ensure that \Cref{eq:blop} holds, giving (b). Moreover, it turns out that through a careful martingale-style analysis of the above process, we can show that we will not add too many irrelevant variables to $\bC'$, as required by (a). To give a flavor of why one should expect this, we'll give some intution for why $\bC'$ ought to contain more elements from $C \cap U^\star$ than from $C \cap \overline{U^\star}$. (This is a significantly weaker statement than what we will truly wish to show, but it illustrates the intution well.) Observe that because $\bI$ satisfies \Cref{eq:beebop}, we expect $\bS$, the spectral sample draw 
 from which $\bT$ is drawn in the $(\star)$ view of the normalized-influence distribution, to contain at most $k^{2/3}/\polylog(k)$ irrelevant coordinates. On the other hand, by construction we have that the set $\bS$ corresponding to the spectral sample has at least $k^{2/3}$ new elements outside of $\bC'$. So we can conclude $|(\bS \setminus \bC') \cap U^\star|$ is much larger than $|(\bS \setminus \bC') \cap \overline{U^\star}|$. Since $\bT$ is a uniform random subset of $\bS$, we then expect that in each iteration, $\bT\setminus \bC'$ also contains more variables from $U^\star$ than from $\overline{U^\star}$. By a concentration argument, it should then follow that in every round we add more relevant variables than irrelevant ones. (We briefly remark that this final concentration step can be shown to break when we need to actually prove the stronger statement that $|\bC' \cap U^\star| \leq \frac{1}{2} |C \cap U^\star|$, requiring us to give a more careful analysis.)
	
	This concludes our sketch of the ideas behind our \refinecoords{} procedure, and thus our classical tolerant junta tester.


\def\bx{\mathbf{x}}

\section{Preliminaries} \label{sec:preliminaries}

\subsection{Preliminaries for agnostically learning conjunctions} \label{sec:agnostic-prelim}

Recall that in the problem of agnostic learning of conjunctions, there is an unknown and arbitrary distribution $\calD$ over pairs $\{\pm 1\}^n \times \{\pm 1\}$.   For a function $g: \{\pm 1\}^n \to \{\pm 1\}$ we define
\[
\err_{\calD}(g) := \Prx_{(\bx,\by) \sim \calD}[g(\bx) \neq \by]\quad\text{and}\quad
\opt_{\cal D} := \min_{c \in \calC} \err_{\calD}(c),
\] where $\calC$ denotes the class of all conjunctions (where $+1$ corresponds to True and $-1$ corresponds to False). An agnostic learning algorithm is an algorithm which, with 
  probability at least $2/3$, constructs a hypothesis $h: \{\pm 1\}^n \to \{\pm 1\}$ that has $\err_{\calD}(h) \leq \opt_{\cal D} + \eps.$

\medskip
\noindent {\bf Notation.}
Throughout {the conjunction learning part} (\Cref{sec:agnostic}) we write $X$ to denote the domain $\{\pm 1\}^n$ for convenience and ${\cal C}$ to denote the class of all conjunctions over $\{\pm 1\}^n$.  We write $\smash{\overrightarrow{(a,b)}}$ to denote a tuple of labeled examples $((a^1,b^1),\dots,(a^m,b^m)) \in (X \times \{\pm 1\})^m$, and write $\vec{a}$ to denote the corresponding tuple $(a^1,\dots,a^m) \in X^m$.

\medskip
\noindent {\bf A generic algorithm for agnostic learning.} We make essential use of the following result from \cite{KKMS}. It says that if every $c\in \calC$ can be approximated by a low-degree polynomial with small squared error, then the $L_1$ polynomial regression algorithm is an efficient agnostic learner for $\calC$:

\begin{theorem}[Theorem~5 of \cite{KKMS}] \label{thm:KKMS}
Fix a distribution $\calD$ over $X \times \{\pm 1\}$.  Suppose that for each $c \in {\cal C}$, there exists a degree-at-most-$d$ polynomial $p$ such that 
\[
\Ex_{(\bx,\by)\sim \calD} \big[(p(\bx) - c(\bx))^2\big] \leq \eps^2.
\]
Then the ``degree-$d$ $L_1$ polynomial regression'' algorithm
runs in time $\poly(n^{d},1/\eps,\log(1/\delta))$ using examples drawn i.i.d.~from $\calD$, and with probability at least $1-\delta$ outputs a hypothesis $h$ that has
$\err_{\calD}(h) \leq \opt_{\cal D} + \eps.$
\end{theorem}

\subsection{Fourier basics} \label{sec:Fourier-basics}

Given two functions $f,g:\{\pm 1\}^n\rightarrow \mathbb{R}$, we write
$$
\langle f,g\rangle =\frac{1}{2^n}\sum_{x\in \bits^n} f(x)g(x)=\Ex_{\bx\sim \bits^n}
  [f(\bx)g(\bx)]\quad\text{and}\quad
\|f\|_2=\sqrt{\langle f,f\rangle}.
$$
Given $S\subseteq [n]$, let $\chi_S(x): \bits^n\rightarrow \{\pm 1\}$ be the function $\chi_S(x)=\prod_{i\in S} x_i$, and let $\hat{f}(S):=\langle f,\chi_S\rangle$ be the 
  Fourier coefficients of $f$.
Then we have $$f=\sum_{S\subseteq [n]} \hat{f}(S)\chi_S\quad\text{and}\quad
\|f\|_2^2=\sum_{S\subseteq [n]} \hat{f}(S)^2 \text{~~~~(by Parseval's identity).}$$

For $f: \bits^n \to \R$ and an integer $L \geq 0$, we write $\bW^{\geq L}[f]$ to denote $\sum_{|S| \geq L} \widehat{f}(S)^2$, the Fourier weight of $f$ at levels $L$ and above.

\subsection{Preliminaries for tolerantly testing juntas} \label{sec:tolerant-prelims}

We write $\calJ_k$ to denote the class of all Boolean-valued $k$-juntas, i.e. all functions $f: \bits^n \to \bits$ such that there exists a function $g: \bits^k \to \bits$ and indices $i_1 < \cdots < i_k$ such that $f(x_1,\dots,x_n) = g(x_{i_1},\dots,x_{i_k})$ for all $x \in \bits^n$.

\subsection{Junta Correlation} 

We will primarily be working with junta correlation in this note, as our functions will not be boolean valued in general. 
Recall that the \emph{$k$-junta correlation of $f$} is defined as 

\[
\corr(f,{\calJ_k}) := \max_{T \in {[n] \choose k}} \corr(f,\calJ_T),
\]
where for any $T \subseteq [n]$, we have
\[
\corr(f,\calJ_T) := 
 \max_{g \in \calJ_T} \E[f(\bx)g(\bx)].\]
We recall from \cite{ITW21} (Section~2.2, specifically Claim~2.7) that
\begin{equation} \label{eq:corr}
\corr(f,\calJ_T) = \Ex[\sign(f^{\overline{T}}_{\ave}(\bx)) f(\bx)], \ \text{where} \
f^{\overline{T}}_{\ave} (x) = \Ex_{\by \sim \bits^n} \left [ f(\by)\hspace{0.06cm} \big |\hspace{0.06cm} \by_T = x_T 
\right] = \sum_{S: S \subseteq T} \widehat{f}(S) \chi_S(x).
\end{equation}
We further recall from \cite{ITW21} that for any subset $T \subseteq [n]$, we have
\begin{equation} \label{eq:corfJC}
\corr(f,\calJ_T)=
\Ex_{\bx \sim \bits^{n}} \left[ \left| f^{\overline{T}}_{\ave} (x) \right| \right] = 
\Ex_{\bx \sim \bits^{n}}\left[ \left| \sum_{A \subseteq T} \widehat{f}(A) \chi_A(\bx) \right|\right]. 
\end{equation}

Recall that for $\pm 1$-valued functions $f,g$ we have that 
	\[ \E[f(\bx) g(\bx)] = 1 - 2 \Pr[f(\bx) \not = g(\bx)]. \]
So it follows that estimating junta correlation to accuracy $\eps$ will also give us an estimate for $\dist(f, \calJ_k)$ if $f$ is boolean valued, where 
\[
\dist(f,\calJ_k) := \min_{g \in \calJ_k} \Pr[f(\bx) \neq g(\bx)]
\]
is the distance from $f$ to the nearest $k$-junta.

We also record the fact that changing a function by a small amount (in terms of expected squared error, i.e.~changing it in the Fourier basis) does not significantly change its junta correlation:

\begin{lemma}
\label{lem:c-s}
Suppose that $f,h: \bits^n \rightarrow \mathbb{R}$ are such that $\|f - h\|_2 \leq \eps.$ Then 
	\[\left| \max_{g \in \calJ_k} \E[ f(\bx) g(\bx)] - \max_{g \in \calJ_k} \E[ h(\bx) g(\bx)] \right| \leq \eps. \]
\end{lemma}

\begin{proof}
	For any $g \in \calJ_k$ we have that
	
		\[
		\left| \E \left [ \left(  f(\bx) - h(\bx) \right)  g(\bx) \right] \right| 
		\leq \| f - h \|_2 \cdot \|g\|_2 \leq \eps \]
	by Cauchy-Schwartz.
\end{proof}

\paragraph{Tolerant testing} \label{sec:tolerant-definition}
For $0<\eps_1<\eps_2$, an \emph{$(\eps_1,\eps_2)$-tolerant $k$-junta tester} is an algorithm which makes black-box queries to an unknown and arbitrary $f: \bits^n \to \bits$ and distinguishes between the two cases that $\dist(f,\calJ_k) \leq \eps_1$ versus $\dist(f,\calJ_k) \geq \eps_2$.
As mentioned earlier, the problem of estimating $\dist(f,\calJ_k)$ up to additive error $\pm \eps$ equivalent to the problem of $(\eps_1,\eps_2=\eps_1+\eps)$-tolerant testing the class of $k$-juntas.

\subsection{Bernoulli noise}

Given a set $V \subseteq [k']$,
 we recall that the \emph{Bernoulli noise operator at noise rate $\rho$ over coordinates in $V$}, denoted $\T^V_\rho$, acts on $f: \bits^{k'} \to \R$ via 
 \begin{equation} 
 \label{eq:Trhodef}
 \T^V_\rho f (x) = \Ex_{\by \sim N^V_\rho(x)}\left[f(\by)\right].
 \end{equation}
Here $N^V_\rho(x)$ is the distribution over $\by \sim \bits^{k'}$ such that each coordinate $\by_i$ of $\by$ is independently distributed as follows:
\begin{flushleft}\begin{itemize}
\item If $i \notin V$ then $\by_i=x_i$;
\item If $i \in V$ then $\by_i$ is set to $x_i$ with probability $\rho$ and is set to be uniform random over $\bits$ with the remaining $1-\rho$ probability.  
\end{itemize}\end{flushleft}
Equivalently, we have $\T^V_\rho \chi_S(x) = \rho^{|S \cap V|} \chi_S(x).$
When $V$ is the entire set $[k']$ of all variables, we omit $V$ from the notation and simply write $T_\rho$ and $N_\rho$.

\subsection{Randomized algorithms and estimating function values, Fourier coefficients, sums of squares of Fourier coefficients, etc.}

We will frequently work with functions that have real-valued outputs. Towards that end, we will use the following definition of a randomized algorithm computing a bounded function:

\begin{definition} 
[$M$-Bounded Randomized Algorithm for an $[-M,M\text{]}$-Bounded Function]
\label{def:randomized-algorithm}
Let $g: \bits^{n} \to [-M,M]$.  We say that algorithm $\calA$ is a \emph{$M$-bounded randomized algorithm for $g$} if on any fixed input $x$, algorithm $\calA$ outputs a random value $\by \in [-M,M]$ with $\E[\by]=g(x).$
\end{definition}

We can easily obtain a high-accuracy, high-confidence estimate of the evaluation of a $[-M,M]$-bounded function or of its Fourier coefficient:

\begin{claim}
[High-Accuracy, High-Confidence Estimation of an $[-M,M\text{]}$-Bounded Function]\label{claim:estimation-of-bounded-function}
Let $g: \bits^{n} \to [-M,M]$ and let $\calA$ be an $M$-bounded randomized algorithm for $g$. Then there is an algorithm $\calA_{\eps,\delta}$ which, on input $x \in \bits^{n}$, makes $O((M/\eps)^2 \cdot \log(1/\delta))$ calls to $\calA$ on input $x$, estimates $g(x)$ up to additive error $\pm \eps$ with probability at least $1-\delta$, and always outputs a value in $[-M,M].$ 
\end{claim}
\begin{proof} The algorithm simply calls $\calA$ repeatedly on $x$ and takes the empirical average of those calls. The analysis  follows from a standard Hoeffding bound. \end{proof}

We will also need the following result which says that we can use a $M$-bounded randomized algorithm for $g$ to estimate $\E_{\bx}[g^2(\bx)]$:

\begin{lemma}[Estimating the $L_2$ Norm]
\label{lem:est-l2-bounded-function}
Let $\eps, \delta \in [0,1/4]$ and suppose that there exists a $M$-bounded randomized algorithm $\calA$ for computing $g: \bits^n \rightarrow [-M,M]$. It then follows that there is an algorithm that estimates the value of $\E[g(\bx)^2]$ to additive error $\pm \eps$ with probability at least $1-\delta$ and makes $O\left ( \frac{M^8}{\eps^4} \log^2(1/\delta) \right)$ queries to $\calA$. Moreover, the value of the estimate always lies in $[0,M^2]$.
\end{lemma}

\begin{proof}
	Let $N := 1000 \frac{M^4}{\eps^2} \log(1/\delta)$ and fix some $x \in \bits^n$. Let $\bz_1(x), \bz_2(x), \dots, \bz_N(x)$ denote the output of $N$ independent runs of $\calA(x)$. Let $\bd_i(x) = \bz_i(x) - g(x)$ and observe $\E_{\calA}[\bd_i(x)] = 0$. Now consider the random variables
		\[ \bZ(x) = \frac{1}{N} \sum_{i=1}^N \bz_i(x)\]
	and
		\[ \bD(x) = \frac{1}{N} \sum_{i = 1}^N \bd_i(x).\]
Since the $\bd_i$'s are independent, have mean $0$, and lie in $[-2M,2M]$, we have that
	\[\Ex_{\calA} \left [\left( \bD(x) \right)^2 \right] = \frac{1}{N^2} \sum_{i=1}^N \E \left[\left(\bd_i(x) \right)^2 \right] \leq \frac{4M^2}{N} \leq \frac{\eps}{2}. \]
	On the other hand, we have
		\[\Ex_{\calA} \left [ \left(\bZ(x) \right)^2 \right] = \Ex_{\calA} \left [ \left(g(x) + \bD(x) \right)^2 \right] =g(x)^2 + \Ex_{\calA} \left[ \left ( \bD(x) \right)^2 \right], \]
	which implies
		\[g(x)^2 \leq \Ex_{\calA} \left [ \left(\bZ(x) \right)^2 \right] \leq g(x)^2 + \frac{\eps}{2}.\]

Given the above, our estimator works as follows: it samples $N$ independent uniform values $\bx^1, \dots, \bx^N \sim \bits^n$ and outputs
		\[\frac{1}{N} \sum_{i=1}^N \left( \bZ(\bx^i) \right)^2\]
	as the estimate for $\E[g^2(\bx)]$. Note that estimate always lies in $[0,M^2]$. 
	By the previous computation, we have that
		\[\left| \Ex_{\bx^i, \calA} \left[ \frac{1}{N} \sum_{i=1}^N \left( \bZ(\bx^i) \right)^2 \right] - \Ex_{\bx} \left[\left( g(\bx) \right)^2 \right] \right| \leq   \frac{\eps}{2}. \]
Since $\bZ(x)^2$ always lies in $[0,M^2]$, by a Hoeffding bound we then have that 
		\[\Pr \left[ \left| \frac{1}{N} \sum_{i=1}^N \left( \bZ_i(\bx^i) \right)^2 - \Ex_{\bx} \left[ \left( g(\bx) \right)^2 \right] \right| \geq \eps \right] \leq \exp \left( \frac{-N\eps^2}{2 \cdot M^4} \right) \leq \delta \]
	as desired. We can easily see that computing $\frac{1}{N} \sum_{i=1}^N \left( \bZ(\bx^i) \right)^2$ requires at most $N^2$ queries to $\calA,$ yielding the claimed query bound.
\end{proof}

We will also frequently need to apply noise operators, so it will be convenient to have the following lemma.

\begin{lemma}
\label{lem:alg-for-noise}
	Suppose that $\calA$ is an $M$-bounded randomized algorithm for $g: \bits^n \rightarrow [-M,M]$. Then for any $\rho \in [0,1]$ and $V \subseteq [n]$, there exists a  $M$-bounded randomized algorithm $\calA'$ for $\T_\rho^V g$ making a single query to $\calA$.
\end{lemma}

\begin{proof}
Given $x \in \bits^n$, the algorithm $\calA'$ simply samples a string $\by \sim N_\rho^V(x)$ and then outputs $\calA(\by)$. This is clearly $M$-bounded and satisfies
\[
\Ex_{\calA'}[\calA'(x)] = \Ex_{\by, \calA} [\calA(\by)] = \Ex_{\by} [g(\by)] = \T_{\rho}^V g(x). \qedhere \]
\end{proof}

It will also be convenient to have a lemma about randomized algorithms for functions obtained by averaging out over a subset of coordinates:
\begin{lemma}
\label{lem:alg-for-average}
	Suppose that $\calA$ is a $M$-bounded randomized algorithm for $g: \bits^n \rightarrow [-M,M]$. Then for any $V \subseteq [n]$, there exists a  $M$-bounded randomized algorithm, making a single call to $\calA$, for the function $h(x) = \Ex_{\by \sim \bits^n} \left [ g(\by) \bigg| y|_{\overline{V}} = x \right]$.
\end{lemma}

\begin{proof}
Consider the algorithm $\calA'$ that given $x$ samples a random $\bz \in \bits^{V}$ and computes $\calA(x|_{\overline{V}} \sqcup \bz)$. We can then observe
\[
\Ex_{\calA'}[\calA'(x)] = 
\Ex_{\bz, \calA} [\calA(x|_{\overline{V}} \sqcup \bz)] = \Ex_{\bz} [g(x|_{\overline{V}} \sqcup \bz) ] = \Ex_{\by \sim \bits^n} \left [ g(\by) \bigg| y|_{\overline{V}} = x \right] \]
as desired. Clearly, $\calA'$ is $M$ bounded and makes a single query to $\calA$.
\end{proof}

\subsection{Flat Polynomials}
\label{subsec:flat-polys}

Underlying our results are constructions of ``flat'' polynomials, originally developed by Linial and Nisan~\cite{LinialNisan:90} and Kahn, Linial, and Samorodnitsky~\cite{kahn1996inclusion} to prove \emph{approximate inclusion-exclusion} bounds. 
We first start with the following construction due to Kahn et al.~\cite{kahn1996inclusion}:

\begin{lemma}[Theorem~2.1 of \cite{kahn1996inclusion}]
\label{lem:flat-polynomials-fine}
Fix integers $r, N$ with $2 \sqrt{N} \leq r \leq N$. Then there exists a polynomial $p:\R\to\R$ of degree at most $r$ with the following properties:
\[
\text{(i)}~~p(0)=0 
\qquad\text{and}\qquad
\text{(ii)}~~\max_{i\in[N]} |p(i) - 1| \leq 2 \exp\pbra{-\Omega \left( \frac{r^2}{N \log N} \right) }.
\]
\end{lemma}

We will also need the construction of Linial and Nisan~\cite{LinialNisan:90} to obtain improved bounds when $\eps_2 \gg \eps_1$.

\begin{lemma}[Theorem~1 of \cite{LinialNisan:90}]
\label{lem:flat-polynomials-coarse}
Fix integers $r,N$ with $r\leq 2 \sqrt{N}$. Then there exists a polynomial $p:\R\to\R$ of degree at most $r$ with the following properties:
\[
\text{(i)}~~p(0)=0 
\qquad\text{and}\qquad
\text{(ii)}~~\max_{i\in[N]} |p(i) - 1| \leq \pbra{1 + \Theta\pbra{\frac{r^2}{N}}}^{-1}.
\]
\end{lemma}

It will be important for us that these polynomials do not blow up too much for values greater than $N,$ and also  that they do not have very large coefficients. Towards this, we will use the following results:

\begin{lemma}[Lemma~11 of \cite{nadimpalli2024optimal}]
\label{lem:flat-polynomial-oos-bound}
Let $r,N$ be integers and suppose that $p: \mathbb{R} \rightarrow \mathbb{R}$ is a polynomial of degree $r$ such that $|p(i)| \leq 2$ for all $i = 1, \ldots, N$ with $p(0) = 0$. For any $\ell \geq N$, $p(\ell) \leq 4\ell^r$. 
\end{lemma}

\begin{lemma}[Lemma~12 of \cite{nadimpalli2024optimal}]
\label{lem:flat-polynomial-coefficient-bound}
Let $r \leq N$ be integers and suppose that $p: \mathbb{R} \rightarrow \mathbb{R}$ is a polynomial of degree $r$ such that $|p(i)| \leq 2$ for all $i = 1, \ldots, N$ with $p(0) = 0$. Moreover, set $\alpha^{r,N}_i$ such that 
\[p(x) = \sum_{i=1}^r \alpha_i^{r,N} \binom{x}{i}\qquad\text{where}\qquad\binom{x}{i} := \frac{x(x-1)...(x-i+1)}{i!}.\] 
Then we have $|\alpha_i^{r,N}| \leq 2r^r$.
\end{lemma}

\subsection{Normalized Influences} \label{sec:normalized-influences}

\begin{definition}
\label{definition:norm-influences-general}
For a given subset $U \subseteq [n]$, we define its \emph{normalized influence} as follows:

\[\NormInf_U[f] := \sum_{S:\, U \subseteq S} \frac{\hat{f}(S)^2}{\binom{|S|}{|U|}}.\]

\end{definition}

Note that the sum $\sum_{|U|=m} \NormInf_U[f]$ of all normalized influences of size $m$ is precisely $W^{\geq m} [f]$.
We remark that normalized influences correspond to drawing a set $\bS$ from the spectral sample and sampling $|U|$ elements. (In more detail, given a particular value of $m$, if we were to draw a random set $\bU$ of size $m$ by choosing each $|U|=m$ with probability ${\frac {\NormInf_U[f]}{\sum_{|U|=m} \NormInf_U[f]}}$, we could alternately describe this distribution of $\bU$ as follows: Repeatedly sample a set $\bS$ from the spectral sample until a set of size $\geq m$ is obtained, and then sample a uniform random subset of $m$ of its elements and take that subset to be $\bU$.)

\subsection{Local estimators}

We recall the notion of a \emph{local estimator} from \cite{nadimpalli2024optimal}.

\begin{definition} [Statistic]
\label{def:statistic}
A \emph{statistic} is a function $S$ that maps a function $f: \bn \to \R$ to a real number in $\R$.
\end{definition}

We recall that the Hamming ball $B(x,r)$ of radius $r$ around a point $x \in \bits^n$ is the set of all $y \in \bits^n$ such that $\dist(x,y) \leq r$, where $\dist(x,y)$ is the Hamming distance between $x$ and $y$.

\begin{definition} [$r$-local estimator]
\label{def:local-estimator}
Given $f: \bn \to \R$, $x \in \bn$, and a positive integer $r \leq n$, an \emph{$r$-local estimator ${\cal E}$} takes as input the values of $f$ restricted to the Hamming ball $B(x,r)$ (we denote this by $f|_{B(x,r)}$) and outputs a real number.

For $\tau\geq 0$, an $r$-local estimator ${\cal E}$ is said to \emph{$\tau$-approximate} a statistic $S$ if
\[
\left|
\Ex_{\bx \sim \bn}\sbra{{\cal E}(f|_{B(\bx,r)})} - S(f)
\right| \leq \tau.
\]
Finally, we say that the estimator ${\cal E}$ is \emph{$\kappa$-bounded} if its range is $[-\kappa,\kappa]$.
\end{definition}

\subsection{Coordinate oracles and approximate versus exact computation} \label{sec:coordinate-oracles}

A central tool in the tolerant testing of $k$-juntas is the notion of approximate coordinate oracles developed by \cite{DMN19}. Essentially, these let us reduce the number of coordinate from $n$ down to $\poly(k,1/\eps)$.

In particular, we will need the following Corollary 4.7 from \cite{ITW21}, which builds on \cite{DMN19}.

\begin{theorem}[Corollary 4.7 of \cite{ITW21}]
\label{thm:coordinate-oracles-exist}
	With $\poly(k, \eps^{-1}, \log(\delta^{-1}))$ queries to $f$, we can gain access to a set of approximate oracles $\calO = \{\calO_1, ... \calO_{k'}\}$ for a set $\calS$ of $k'$ coordinates from $[n]$.
	 Moreover, these coordinates satisfy the following properties:
	\begin{enumerate}
		\item For every coordinate, $i \in \calS$, there exists a $g \in \calO$ such that $g$ is $0.1$ close to $\Dict_i(x)=x_i$ with probability at least $1-\delta.$
		(Hence we refer to $\calO$ as a \emph{set of approximate coordinate oracles}.)
		\item $\min_{S \subseteq [n]: |S| \leq k} \dist(f, \calJ_S) - \min_{S \subseteq \calS: |S| \leq k} \dist(f, \calJ_S) \leq \eps$. 
		\item $k' = |\calS| \leq \poly(k, \eps^{-1}, \log(1/\delta))$.
		\item For an algorithm $A$ that uses at most $q$ queries to $\calD$, we can assume that we have perfect oracle access to the dictator correspond to each coordinate oracle up to an additive loss of $\delta$ in the confidence and a multiplicative overhead of $\poly(\log(q/\delta))$ in query complexity.
		(We refer to this set $\calO'$ of oracles for the exact dictator functions corresponding to the elements of $\calO$ as a \emph{set of coordinate oracles}.)
	\end{enumerate}
\end{theorem}

Given point $(4)$, we will assume throughout the paper that we indeed have exact access to the coordinate oracles i.e. each coordinate oracle $\calO_i$ exactly corresponds to the dictator function of its corresponding coordinate. 
Now it's important to note that while each $\calO_i$ corresponds to a single coordinate, we do not know which of the $n$ coordinates this is. (Notably, that would information-theoretically require $\Omega(\log(n))$ queries, which is too large when $k$ is small.) As such we only have \emph{implicit} access to these coordinates (cf. \cite{servedio2010testing}). Nonetheless, \cite{ITW21} are also able to show that we can query these coordinates to get a randomized algorithm for $f_{\ave}^{\overline{\calS}}$:

\begin{theorem}[Theorem 4.10 of \cite{ITW21}	]
\label{thm:alg-for-coordinate-oracle-avg}
	Let $f \isafunc$ and $\calO'=\{\calO'_1,\dots,\calO'_{k'}\}$ be a set of $k'$ coordinate oracles corresponding to a set $\calS \subset [n], |\calS| = k'.$ Let $g$ be a function from $\bits^{k'} \rightarrow \bits$ defined by 
	\[g(y) := \E_{\bx} \left [f(\bx) \bigg | \calO'_1(\bx) = y_1, \calO'_2(\bx) = y_2, \dots, \calO'_{k'}(\bx) = y_{k'} \right ] = f_{\ave}^{\overline{\calS}}.\]
Then there exists a $1$-bounded randomized algorithm for $g$ that makes $\poly(k')$ queries to $f$ in expectation. 
\end{theorem}


\def\const{\textsc{Const}}

\section{Agnostically learning conjunctions}
\label{sec:agnostic}

In this section we present our agnostic learning algorithm for conjunctions
  and prove \Cref{thm:agnostic}.

\subsection{Setup for the learning algorithm} \label{sec:setup}

Before explaining the algorithm we establish some helpful notation and terminology. (We view $+1$ as corresponding to True and $-1$ as corresponding to False, and recall that $X=\{\pm 1\}^n$.)

\begin{definition} \label{def:nonconst}
Given a tuple of points $\vec{a}=(a^1,\dots,a^m) \in X^m$, we write $\nonconst(\vec{a}) \subseteq [n]$ to denote the set of coordinates that are not constant across the strings in $\vec{a}$. Formally, $$\nonconst(\vec{a}) := \big\{i \in [n]: a^1_i \not = a^j_i\ \text{for some $j \in [m]$} \big\} $$
and $\const(\vec{a})=[n]\setminus \nonconst(\vec{a})$. So all strings in $\vec{a}$ agree on coordinates
  in $\const(\vec{a})$.

Given any $x\in \{\pm 1\}^n$ and $\vec{a}=(a^1,\ldots,a^m)$, we use $\dist(x,\vec{a})$ to denote
  the Hamming distance between $x$ and $a^1$ over coordinates in $\const(\vec{a})$, i.e., the number of 
  $i\in \const(\vec{a})$ with $x_i\ne a^1_i$.
\end{definition}

\begin{definition} \label{def:ball-dist}
Given a tuple of points $\vec{a}=(a^1,\dots,a^m) \in X^m$, the event
$E_{\vec{a}} 
\subseteq X$ 
is defined to be the set of all $x \in X$ such that 
$\dist(x,\vec{a})\le n^{2/3}$.  Formally, 
\[
E_{\vec{a}}
:= 
\cbra{x \in X: \text{number of $i\in \const(\vec{a})$ with $x_i\ne a^1_i \leq n^{2/3}$}}.
\]
Given a distribution  $\calD$ over $X\times \{\pm 1\}$,
we write $\calD|E_{\vec{a}}$ to denote the distribution of $(\bx,\by)\sim \calD$ conditioned on 
 $\bx$ belonging to $E_{\vec{a}}$, and we refer to $\calD|E_{\vec{a}}$ as a ``ball distribution'' since its marginal over the coordinates in $ \const(\vec{a})$ is supported on strings contained in a Hamming ball.
\end{definition}

\subsection{The algorithm} \label{sec:algorithm}

Our algorithm for agnostically learning conjunctions is presented in 
  \Cref{alg:agnostic}. 
Let $\calD$ be the input distribution over $X\times \{\pm 1\}$, and let
  $c^*$ be its (unknown) closest conjunction with $\err_\calD(c^*)=\opt_\calD$.
Our algorithm will draw an $m$-tuple of independent labeled examples 
$$
\overrightarrow{(\ba,\bb)}=\big((\ba^1,\bb^1),\dots,(\ba^m,\bb^m)\big)$$ from $\calD$.
Intuitively, we would like to work with the ball distribution
  $\calD|E_{\vec{\ba}}$ obtained from a tuple $\vec{\ba}=(\ba^1,\ldots,\ba^m)$ such that $c^*(\ba^i)=1$ for each $i$, but we do not have access to the value of $c^*(\ba^i)$; rather, we only have access to the label bit $\bb^i$.  For the moment, though, let us hope or pretend that if $\bb^1=\cdots=\bb^m=1$ then additionally indeed $c^*(\ba^i)=1$ for all $i$ (this is roughly what it means to obtain a ``useful'' set of examples as alluded to in \Cref{sec:overview-agnostic}).\footnote{We remark that this approach of ``hoping'' that all examples satisfy $c^*$ is similar in spirit to an idea used in the DNF learning context in \cite{de2014learning}.  In that setting the algorithm ``hopes'' that all members of some set of positive examples satisfy the same term of the DNF, and takes a bitwise-AND of examples to try to identify that term.  Both \cite{de2014learning} and the current paper use multiple repeated trials, since in both settings the desired ``hope'' may only hold with fairly small probability.} 
  
Now if indeed $\vec{\ba}$ satisfies that $c^*(\ba^i)=1$ for each $i$, 
  the relevant variables of the conjunction $c^*$ must all belong to $\const(\vec{\ba})$, the set of ``previously unanimous'' coordinates that always took the same value across all $m$ examples $\ba^1,\dots,\ba^m$.  Intuitively, the event $E_{\vec{\ba}}$ is the event that a string $x\in X$ disagrees with $\ba^1$ on ``not too many'' of these ``previously unanimous'' coordinates in $\const(\vec{\ba})$ (this is the event $E$ that was alluded to in \Cref{sec:overview-agnostic}).

We briefly describe  each step of \Cref{alg:agnostic}.
In each iteration $i$ of the main loop of Step~2, the algorithm draws $\overrightarrow{(\ba,\bb)}$ from $\calD^m$ (Step 2(a)) and checks that all examples $(\ba^1,\bb^1),\dots,(\ba^m,\bb^m)$ have $\bb^i=1$ (Step~2(b)), and continues only when this happens.
As discussed above, the algorithm ``hopes'' that $c^*(\ba^i)=1$ for all $i$, and further checks that $\Pr_{(\bx,\by)\sim\calD}[E_{\vec{\ba}}(\bx)]$ is not too small (Step~2(c)). If this check passes, it runs the $L_1$ regression algorithm on examples from $\calD|E_{\vec{\ba}}$ (this is the $\calD'$ alluded to in \Cref{sec:overview-agnostic}) to obtain a hypothesis $h_i$. 
At the end of the main loop, it draws more examples from $\calD$ and output the 
  best hypothesis $h_i$ with the minimum disagreement.

{\begin{algorithm}[t!]
\addtolength\linewidth{-2em}

\vspace{1em}

\textbf{Input:} Samples drawn from an arbitrary and unknown distribution $\calD$ over $X \times \bits$. \\[0.25em]
\textbf{Output:} A hypothesis $h: X \to \bits$.

\vspace{0.5em}

\learnconjunction:

\vspace{0.5em}

\begin{enumerate}

\item Set $m :=  n^{1/3}$,
$L := (1/\eps)^{O(m)}$, 
and $h_0: X \to \bits$ to be $h_0 \equiv -1.$
\item Repeat for $i=1,\dots,L$:
	\begin{enumerate}
	\item Sample $\overrightarrow{(\ba,\bb)} = ((\ba^1, \bb^1), \dots (\ba^m, \bb^m)) \sim \calD^m$.\vspace{0.08cm}
	
	\item  If any $\bb^i=-1$, set $h_i \equiv -1$ and skip to the next iteration of the loop.\vspace{0.08cm} 	
	
	\item 
	Run the $L_1$ regression algorithm (\Cref{thm:KKMS}) (with its failure probability parameter $\delta$ set to 0.01, its error parameter set to $\eps$, and its degree parameter set to $d=O(n^{1/3}\log(1/\eps))$), using examples drawn from $\calD|E_{\vec{\ba}}$, to obtain a hypothesis $h'_i$. 
Access to $\calD|E_{\vec{\ba}}$ is simulated by drawing from $\calD$ and checking if the sample satisfies $E_{\vec{\ba}}$. Let $N$ be the number of samples needed by the $L_1$ regression algorithm. Draw $O(1/\eps) \cdot N$ samples from $\calD$ and feed the first $N$ of them that belong to $E_{\vec{\ba}}$ to the $L_1$ regression algorithm; set $h_i\equiv -1$ and skip to the next iteration if there are 
less than $N$ samples in $E_{\vec{\ba}}$. 

\item Finally, set $h_i$ to be the following hypothesis: For each $x\in X$:
			\begin{equation} \label{eq:hi}
			h_i(x) = \begin{cases}
h'_i(x) & \text{if $E_{\vec{\ba}}(x)$ holds} \\ -1 \text{~(i.e. False)}& \text{if $E_{\vec{\ba}}(x)$ does not hold}
		\end{cases}
			\end{equation}

	\end{enumerate}

\item Draw a set of $\poly(n/\eps)$ samples from $\calD$ and choose the hypothesis from $h_0,h_1, \dots, h_L$ that has the minimum disagreement with the samples drawn.

\end{enumerate}
\caption{An algorithm to agnostically learn conjunctions.}
\label{alg:agnostic}
\end{algorithm}}

\subsection{Analysis of the algorithm}

Let us write $\calD_{c^*}$ to denote the distribution of $(\bx,\by) \sim \calD$ conditioned on having ${c^*}(\bx)=\by=1.$
(Note that this is well defined if there are such pairs in the support of $\calD$.
We will always assume this is the case in this subsection. 
The case when the support of $\calD$ contains no such pairs is trivial and 
  will be handled by the default $h_0$ in the algorithm. See \Cref{sec:conjproof}.)

To analyze a single iteration of the loop, we will use the following definition with $m=n^{1/3}$:

\begin{definition} \label{def:useful}
We say a tuple of pairs $\overrightarrow{(a,b)}=((a^1,b^1),\dots,(a^m,b^m)) \in (X \times \bits)^m$ is \emph{useful} if $c^*(a^t)=b^t=1$ for all $t \in [m]$ and $\vec{a}=(a^1,\dots,a^m)$ satisfies the following condition:
\begin{equation} \label{eq:high-prob}
\Prx_{(\bx,\by) \sim \calD_{c^*}} \sbra{
E_{\vec{a}}(\bx)
} \geq 1-\eps.
\end{equation}
\end{definition}

The following lemma shows that the probability of $\overrightarrow{(\ba,\bb)} \sim \calD_{c^*}^m$
  being useful is not too small:

\begin{lemma}\label{commonlemma1}
Let $\overrightarrow{(\ba,\bb)} \sim (\calD_{c^*})^m$ with $m=n^{1/3}$.
The probability that $\overrightarrow{(\ba,\bb)}$ is useful is at least $\eps^{m}$.
\end{lemma}
\begin{proof}
We view the samples $\ba^1,\ldots,\ba^m$ as being drawn one by one. We always use $\vec{\ba}$ to denote 
  the current tuple of $\ba^1,\ldots,\ba^t$ collected after $t$ rounds,
  and we let $E_{\vec{\ba}}$ be the event using the current $\vec{\ba}$:
  \[
E_{\vec{\ba}}
:= 
\cbra{x \in X: \text{number of $i\in \const(\vec{\ba})$ with $x_i\ne \ba^1_i \leq n^{2/3}$}}.
\]
Note that since our draws are from ${\cal D}_{c^*}$ we get $c^*(\ba^t)=\bb^t=1$ for all $t$ ``for free,'' and we need only worry about the condition given in \Cref{eq:high-prob}.

After drawing $\ba^1$, we have that $\const(\vec{\ba})=[n]$, and 
  $x\in X$ satisfies $E_{\vec{\ba}}$ if and only if
  the number of $i\in [n]$ with $x_i\ne \ba^1_i$ (i.e.~the Hamming distance $\dist(x,\ba^1)$) is at most $n^{2/3}$.
If the probability of $E_{\vec{\ba}}(\bx)$ over $(\bx,\by)\sim \calD_{c^*}$ at this point 
  is already at least $1-\eps$ then we are done
  (this is because as we draw more samples, $\const(\vec{\ba})$ can only shrink 
  and the probability of $E_{\vec{\ba}}(\bx)$ can only increase).
So assume that after the first round, the probability of $E_{\vec{\ba}}(\bx)$ 
  over $(\bx,\by)\sim \calD_{c^*}$ is at most $1-\eps$; this
  this means that with probability at least $\eps$ over $(\ba^2,\bb^2)\sim \calD_{c^*}$, the size of 
  $\const(\vec{\ba})$ goes down by $n^{2/3}$.
We pay this factor of $\eps$ in probability, and ask for such a second sample $(\ba^2,\bb^2)$
  and add it to $\vec{\ba}$; after this, the size of $\const(\vec{\ba})$ has shrunk by at least $n^{2/3}$.
  
We repeat the above argument.  This can repeat no more than $n^{1/3}$ times since $\const(\ba)$
  can shrink by at least $n^{2/3}$ for no more than $n^{1/3}$ rounds.
This means that within the $m=n^{1/3}$ rounds, there must be a round after which
  the probability of $E_{\vec{\ba}}(\bx)$ over $(\bx,\by)\sim \calD_{c^*}$ becomes at least $1-\eps$.
The lemma follows given that the total probability we paid for this to happen is at most $\eps^m$.
\end{proof}

Before giving the next lemma we recall some basics about Chebyshev polynomials.
We write $T_k$ to denote the $k$th Chebychev polynomial of the first kind. We will use the following well-known facts about Chebyshev polynomials:

\begin{fact} \label{fact:chebyshev}
(I) $T_k(1 + \eps) \geq(1/2) e^{k \sqrt{\eps}}$ for any $0 \leq \eps \leq 0.4$.
(II) $|T_k(x)| \leq 1$ for all $x \in [-1,1]$.
\end{fact}

\noindent Item (I) follows immediately from the definition $$T_k(x) = {\frac 1 2} \sbra{ 
\pbra{x + \sqrt{x^2 - 1}}^k + \pbra{x - \sqrt{x^2 - 1}}^k }$$ and that $1+\sqrt{2\eps} \geq e^{\sqrt{\eps}}$ for $0 \leq \eps \leq 0.4$.  Item (II) follows from the fact that $T_k(\cos \theta) = \cos(k \theta).$

\begin{lemma} \label{lem:apx-deg}
Suppose that $\overrightarrow{(a,b)}\in (X\times \{\pm 1\})^m$ is useful. Then there exists a polynomial $p$ of degree $d={O}(n^{1/3}\log(1/\eps))$ such that 
$|p(x)-c^*(x)| \leq \eps$ for every $x$ in $E_{\vec{a}}$ and consequently, $$\Ex_{(\bx,\by) \sim \calD|E_{\vec{a}}}\sbra{\pbra{p(\bx) - c^*(\bx)}^2} \leq \eps^2.$$
\end{lemma}

\begin{proof}
For simplicity, we assume that $c^*$ is a monotone conjunction of the form $x_1 \wedge \cdots \wedge x_s$; the general case follows similarly. We essentially use the usual Chebychev construction of approximate polynomials for \textsf{AND}, but we take advantage of the fact that by the nature of $E_{\vec{a}}$, we are promised that at least $s - n^{2/3}$ literals in the conjunction are always satisfied. 
To see this is the case, note that $c^*(a^i)=1$ for all $i$ so $[s]\subseteq \const(\vec{a})$ and 
  $a^1_1=\cdots =a^1_s=1$.
Given that $x\in E_{\vec{a}}$ can only disagree with $a^1$ at $n^{2/3}$ coordinates in $\const(\vec{a})$, it can only falsify $n^{2/3}$ of $x_1,\ldots,x_s$.

Inspired by this observation, let $$\Delta: = \min\big(n^{2/3},s\big),\quad d:=\left\lceil 3 \sqrt{\Delta} \log(1/\eps)\right\rceil\quad\text{and}\quad
	q(t) = \frac{T_{d}\left(\frac{t}{\Delta} \right)}{T_{d}\left(\frac{\Delta+1}{\Delta}\right)},$$
where $T_d$ is the $d$th Chebychev polynomial of the first kind. It's clear that $q(\Delta+1) = 1$. Moreover, by \Cref{fact:chebyshev}, we have that for any $t \in [-\Delta, \Delta]$, 
	\begin{equation} \label{eq:q-prop}
	|q(t)| 
	\leq 2e^{-d\sqrt{1/\Delta}}\le  2e^{-3 \log(1/\eps) \sqrt{\Delta } \cdot \sqrt{1/\Delta}} \leq \eps,
	\end{equation}
	when $\eps$ is sufficiently small.
Thus it suffices to set
	\[p(x) = q \left( \pbra{\sum_{i \in [s]} x_i} - (s - \Delta) +1\right).
	\]
Using \Cref{eq:q-prop}, we have that $|p(x)-c^*(x)| \leq \eps$ for any string $x$ at distance at most $\Delta$ from satisfying $c^*$. Since $E_{\vec{a}}$ contains such points only, the lemma is proved.
\end{proof}

\subsection{Proof of \Cref{thm:agnostic}}\label{sec:conjproof}

\noindent
{\bf Efficiency.}
It is easy to verify that the overall running time 
   of \Cref{alg:agnostic} is dominated by the $L$ calls to the $L_1$ regression algorithm in Step~2(d). Using \Cref{thm:KKMS}, the overall runtime  is $$L \cdot \poly\left(n^d,1/\eps\right)=\left(\frac{1}{\eps}\right)^{\tilde{O}(n^{1/3})},$$
   where $d={O}(n^{1/3}\log(1/\eps))$, which is as claimed in \Cref{thm:agnostic}.

\medskip
\noindent
{\bf Correctness.}
Our goal is to prove the following lemma:

\begin{lemma}\label{lem:finalconj}
With probability at least $0.9$, 
  at least one of the functions $h_0,h_1,\ldots,h_L$ at the end of the main loop satisfies
  $\err_{\calD}(h_i) \leq \opt_{\calD}+2\eps$.
\end{lemma}

On the other hand, with probability at least $0.9$, 
  Step 3 can approximate $\err_{\calD}(h_i)$ for every $i$ 
  up to error $\eps$.
So by a union bound, with probability at least $0.8$, the hypothesis $h$ returned satisfies 
  $ \err_{\calD}(h)\le \opt_{\calD}+4\eps.$ 
We prove \Cref{lem:finalconj} in the rest of this section.

We start by dealing with a trivial case: 
\begin{equation}		\label{eq:atleastepsprob}
		\Prx_{(\bx, \by) \sim \calD}\big[c^* (\bx) = \by = 1\big]< \eps.
		\end{equation}
(Note that this includes the case mentioned earlier about $\calD_{c^*}$ being not well defined.)
We show that in this case $h_0\equiv -1$ has a small $\err_\calD(h_0)$. To see that this is the case, we have 
		\[1 - \opt_{\calD} = \Prx_{(\bx, \by) \sim \calD}\big[c^* (\bx) = \by = 1\big] + \Prx_{(\bx, \by) \sim \calD}\big[c^* (\bx) = \by = -1\big]<
		\eps+ \Prx_{(\bx, \by) \sim \calD}\big[c^* (\bx) = \by = -1\big]\] 
and thus, 
$$
1-\err_\calD(h_0)\ge \Prx_{(\bx, \by) \sim \calD}\big[c^* (\bx) = \by = -1\big]>1-\opt_\calD-\eps 
$$
so $\err_\calD(h_0)\le \opt_\calD+\eps$.
In the rest of the proof we assume that \Cref{eq:atleastepsprob} does not hold.

We will show that the following are sufficient conditions for $h_i$ from the $i$th loop
  to satisfy $\err_{\calD}(h_i) \leq \opt_{\calD}+2\eps$.
\begin{enumerate}
\item $\overrightarrow{(a,b)}$ drawn in this loop is useful.
\item The $L_1$ regression algorithm receives all of the samples it needs.
\item The $L_1$ regression algorithm returns a function $h_i'$ that satisfies
\begin{equation}\label{eq:hehe}
\err_{\calD|E_{\vec{ a}}}(h_i')\le \opt_{\calD|E_{\vec{ a}}}+\eps.
\end{equation}
\end{enumerate}
Before that we show that these conditions hold for some $i$ with probability at least $0.9$.

Using \Cref{commonlemma1} and the assumption that \Cref{eq:atleastepsprob} does not hold,
  we have that $\overrightarrow{(\ba,\bb)}\sim \calD^m$ is useful 
  with probability at least 
$
\eps^m\cdot \eps^m=\eps^{O(m)},
$
where we pay the first $\eps^m$ to draw $m$ samples from $\calD_{c^*}$ (using \Cref{eq:atleastepsprob}) and 
  the second $\eps^m$ is from \Cref{commonlemma1} for the tuple to be useful.
Given our choice of $L=(1/\eps)^{O(m)}$, 
we have that with probability at least $1-o_n(1)$,
  $\overrightarrow{(\ba,\bb)}$ satisfies the first item for at least one loop $i$.
Let's fix such a useful pair $\overrightarrow{(a,b)}$.

Next, given that $\overrightarrow{(a,b)}$ is useful, we have
$$
\Prx_{(\bx, \by) \sim \calD_{c^*}}\big[E_{\vec{a}} (\bx)\big] \geq 1-\eps. 
$$ 
Combining this with the assumption that \Cref{eq:atleastepsprob} does not hold, we have 
\[\Prx_{(\bx, \by) \sim \calD}\big[E_{\vec{a}} (\bx)\big] \geq \eps(1-\eps) \geq \eps/2, \]
from which the second item occurs during that loop with probability at least $0.99$.
Assuming both the first and second items, it follows from \Cref{lem:apx-deg} and \Cref{thm:KKMS}
  that the third item holds with probability at least $0.99$.
As a result, the probability of having at least one loop satisfying all three
  items is at least $0.9$.
It suffices to show that when all three items hold, $h_i$ (obtained from $h_i'$) satisfies  $\err_{\calD}(h_i) \leq \opt_{\calD}+2\eps$. 
  
First, from \Cref{eq:hehe} we have (using $c^*\in \calC$)
		\begin{equation} \label{eq:a}
		\Prx_{(\bx,\by) \sim {\calD}|E_{\vec{ a}}}\big[h_i'(\bx) \neq \by\big] \leq 
		\Prx_{(\bx,\by) \sim {\calD}|E_{\vec{a}}} \big[c^*(\bx) \not = \by\big] + \eps.
		\end{equation}
To relate the error of $h_i'$ under $\calD|E_{\vec{a}}$ to the error of 
  $h_i$ under $\calD$, we note that
\[
\err_{\calD}(h_i) =
\Prx_{(\bx,\by) \sim \calD}\big[h_i(\bx) \neq \by\big] =
\Prx_{(\bx,\by) \sim \calD}\big[\neg E_{\vec{ a}}(\bx) \wedge (h_i(\bx) \neq \by)\big]
+
\Prx_{(\bx,\by) \sim \calD}\big[E_{\vec{ a}}(\bx) \wedge (h_i(\bx) \neq \by)\big].
\]	
For the first term, recalling \Cref{eq:hi} we have that if $\neg E_{\vec{ a}}$ holds then $h_i=- 1$, so
\begin{align*}
&\Prx_{(\bx,\by) \sim \calD}\big[\neg E_{\vec{ a}}(\bx) \wedge (h_i(\bx) \neq \by)\big]\\
&=
\Prx_{(\bx,\by) \sim \calD}\big[\neg E_{\vec{ a}}(\bx) \wedge (h_i(\bx)=-1, \by=1)\big]\\
&\leq 
\Prx_{(\bx,\by) \sim \calD}\big[\neg E_{\vec{ a}}(\bx) \wedge (\by=1)\big]\\
&=
\Prx_{(\bx,\by) \sim \calD}\big[\neg E_{\vec{ a}}(\bx) \wedge (\by=1,c^*(\bx)=-1)\big]
+
\Prx_{(\bx,\by) \sim \calD}\big[\neg E_{\vec{ a}}(\bx) \wedge (\by=c^*(\bx)=1)\big]\\
&\leq
\overbrace{
\Prx_{(\bx,\by) \sim \calD}\big[\neg E_{\vec{ a}}(\bx) \wedge (c^*(\bx)\neq \by)\big]}^{=A}
+
\overbrace{\Prx_{(\bx,\by) \sim \calD}\big[\neg E_{\vec{ a}}( x) \wedge (\by=c^*(\bx)=1)\big]}^{=B}.
\end{align*}
For the second term, we have
\begin{align*}
\Prx_{(\bx,\by) \sim \calD}\big[E_{\vec{ a}}(\bx) \wedge (h_i(\bx) \neq \by)\big]
&= \Prx_{(\bx,\by) \sim \calD|E_{\vec{ a}}}\big[h_i(\bx) \neq \by\big] \cdot \Prx_{(\bx,\by) \sim \calD}\big[E_{\vec{ a}}(\bx)\big]\\
&\leq
\overbrace{
\pbra{\Prx_{(\bx,\by) \sim \calD|E_{\vec{ a}}}\big[c^*(\bx) \neq \by\big] + \eps}
\cdot \Prx_{(\bx,\by) \sim \calD}\big[E_{\vec{ a}}(\bx)\big]}^{=C},
\end{align*}
where the inequality is by \Cref{eq:a}.  It remains to argue that $A+B+C \leq \opt_\calD + 2 \eps.$
Since $$C \leq \Prx_{(\bx,\by) \sim \calD}\big[E_{\vec{ a}} \wedge (c^*(\bx) \neq \by)\big] + \eps,$$ by inspection we have that $A+C \leq \opt_\calD + \eps.$
Finally, we have $$B \leq \Prx_{(\bx,\by) \sim \calD}\big[\neg E_{\vec{ a}}(\bx) \hspace{0,06cm}|\hspace{0.06cm} c^*(\bx) = \by = 1\big] \le \eps $$ 
using the assumption that $\overrightarrow{(a,b)}$ is useful.
This finishes the proof of \Cref{lem:finalconj}. 

\def\JuntaCorr{\textsc{Junta-Correlation}}
\def\overr{\textsc{over}}

\section{Warm-Up: A $2^{\wt{O}(k^{1/3}))}$-Query Quantum Tolerant Junta Tester}
\label{sec:quantum}

In this section we prove \Cref{thm:quantum}. 
While this result will be subsumed by \Cref{thm:classical}, which gives a classical tester that has a matching query complexity, 
many of the tools which we develop in a simpler setting here will be needed later for our classical tester. Throughout the section, we take $k'$ to be a generic parameter and show 
that the algorithm for testing functions $f: \bits^{k'} \to \bits$ makes $\exp(k^{1/3}\cdot\polylog(k'/\eps))$ queries. 
\Cref{thm:quantum} follows by plugging in $k'= \poly(k,\eps^{-1})$. 

\subsection{Drawing from the spectral sample} \label{sec:spectral-sample}
We recall the definition of the spectral sample of a Boolean function:

\begin{definition} [Spectral sample, Definition~1.18 of \cite{odonnell-book}] \label{def:spectral-sample}
Given a function $f: \bits^{k'} \to \bits$, the \emph{spectral sample} of $f$, denoted ${\cal P}_f$, is the probability distribution on subsets of $[k']$ (equivalently, on elements of $\bits^{k'}$) in which the set $S$ has probability $\smash{\widehat{f}(S)^2}$. 
(Recall from \Cref{sec:Fourier-basics} that $\sum_S \widehat{f}(S)^2=1$, so this is indeed a valid probability distribution.)
\end{definition}

As discussed in \Cref{rem:quantum}, the only quantum aspect of \Cref{alg:quantum} occurs in lines~2 and~3(b) where the algorithm makes draws from the spectral sample ${\cal P}_f$ of $f$. 
Thus the rest of our discussion in this section will not involve any quantum considerations.

\subsection{Smooth functions and local mean estimation} \label{sec:smooth}

We will \emph{informally} refer to a function $f:\{\pm 1\}^{k'}\rightarrow \R$ 
  as an $L$-smooth function for some positive integer $L\in [k']$ if 
  $\bW^{\ge L}[f]$ is \emph{tiny}.
As we will see, $L$-smooth functions are useful because, when $\bW^{\ge L}[f]$ is sufficiently small, using the techniques of \cite{nadimpalli2024optimal}, the magnitude of their means can be estimated using evaluations of $f$ that come from a random ball of radius roughly $\sqrt{L}$, 
i.e.~they have \emph{local estimators}. 
To establish this, we start with a  simple lemma from \cite{nadimpalli2024optimal}.

\def\trunc{\textsc{trunc}}

\begin{lemma}
\label{lem:low-var-expectation}
Let $\bX$ be a random variable with $\Var[\bX] = \sigma^2$. Then 
\[0 \leq  \E \left [|\bX| \right] - \left |\E[\bX] \right|  \leq \sigma.\]
\end{lemma}
\begin{proof}
The first inequality is trivial since $\E[|\bX|] \geq |\E[\bX]|$ for every random variable $\bX$.

For the second one, we have
\[
\E \left [|\bX| \right] - \left |\E[\bX] \right|
=
\E\left[\left|\bX\right| - \left|\E[\bX]\right| \right]
\leq 
\E\left[ \left| \bX - \E\left[\bX\right]\right|\right]
\leq
\sqrt{\E\left[\left(\bX - \E\left[\bX\right]\right)^2\right]} = \sigma,
\]
which completes the proof.
\end{proof}

Now we prove the existence of local estimators for $L$-smooth functions.  The error of these local estimators depends on the variance of the $L$-smooth function as well as its Fourier weight on levels at least $L$. 

\begin{lemma} [Local Estimators for Smooth Functions]
\label{lem:smooth-estimator}
Let $\tau$ and $L$ be two parameters such that  $\tau\in (0,1/2]$ and $L\le k'$ is a positive integer.
There exists an $r$-local estimator $\calE$ with
$$
r=\Theta\left(\sqrt{L \log(L) \log(1/\tau)}\right)
$$ that approximates $|\E[f]|$ for all functions 
   $f: \bits^{k'} \to \R$ with error at most
$$
\tau\sqrt{\Var[f]}+5(k')^r \sqrt{\bW^{\ge L}[f]}).
$$
\end{lemma}

\begin{proof}
We closely follow the proof of Lemma~17 of \cite{nadimpalli2024optimal}.
Let $r := \Theta(\sqrt{L \log(L) \log(1/\tau)})$ and let $p^{L}_r(x)$ be the polynomial from \Cref{lem:flat-polynomials-fine} (with $L$ playing the role of $N$).
We set the hidden constant in $r$  large enough so that 
\begin{equation} \label{eq:deviation-bound}
\max_{i\in [L]} \left|p^L_r(i)-1\right|\le 2\exp\left(-\Omega\left(\frac{r^2}{L\log L}\right)\right)\le {\tau}.
\end{equation}
As in \Cref{lem:flat-polynomial-coefficient-bound}, let $\alpha_i^{r,L}$ be the coefficients 
  with $|\alpha_i^{r,L}|\le 2r^r$ such that
\[p_r^L(x) = \sum_{i\in [r]} \alpha_i^{r,L} \binom{x}{i}.\]
We take
\begin{equation}\label{eq:neededlater1} g\left(f\big|_{B(x,r)}\right) := f(x) - \sum_{i\in [r]} \alpha_i^{r,L} \sum_{S \subseteq [k']: |S| = i} \frac{\partial f}{\partial x_S} (x) \chi_S(x).
\end{equation}
Note that this can be computed by only querying $f$ on $B(x,r)$, as
\begin{equation}\label{eq:neededlater2}
\frac{\partial f}{\partial x_S}(x)\chi_S(x) = \frac{1}{2^{|S|}} \cdot \left( \sum_{T \subseteq S}   (-1)^{|T|} f(x^{\oplus T})\right).
\end{equation}

Writing $f(x^{\oplus T}) = \sum_{U \subseteq [k']} \wh{f}(U)\chi_U(x^{\oplus T})$, it is easy to check that
\begin{align}
	g\left(f\big|_{B(x,r)}\right) &= f(x) - \sum_{i\in [r]} \alpha_i^{r,L} \sum_{S\sse[k']: |S| = i} \sum_{U\supseteq S}\wh{f}(U)\chi_U(x) \nonumber\\
	&= f(x) - \sum_{U \neq \emptyset} p_r^L(|U|)\wh{f}(U)\chi_U(x) \nonumber\\
	&= \E[f] + \sum_{U\neq\emptyset} \pbra{1 - p_r^{L}(|U|)} \wh{f}(U)\chi_U(x), \label{eq:united}
\end{align}
where we used the fact that $\wh{f}(\emptyset) = \E[f]$. It is immediate from the above that 
\begin{equation}
\label{eq:Efgmatch}
\Ex_{\bx\sim\bits^{k'}} \left [ g \left( f\big|_{B(\bx,r)} \right) \right] = \E[f],
\end{equation}
where we used $\bE[\chi_S(\bx)]=0$ for $S\ne\emptyset$.
Furthermore, we have 
\begin{align}
\Varx_{\bx\sim\bits^{k'}}\sbra{g \left( f\big|_{B(\bx,r)} \right)} 
= \Ex_{\bx\sim\bits^{k'}} \left [ \left( g \left( f\big|_{B(x,r)} \right) - \E[f] \right)^2 \right]
= \sum_{i\in [k']} \left(1 - p_r^{L}(i)\right)^2 \bW^{=i}[f], \label{eq:cowbell}
\end{align}
where \Cref{eq:cowbell} follows from \Cref{eq:united} via Parseval's formula. 

Splitting the sum into two parts, the RHS of \Cref{eq:cowbell} becomes
\begin{align}
\sum_{i \leq L} \left(1 - p_r^{L}(i)\right)^2 \bW^{=i}[f] + \sum_{i > L} \left(1 - p_r^{L}(i)\right)^2 \bW^{=i}[f] \nonumber  \leq \tau^2\cdot \Var[f] + 
25(k')^{2r} \cdot \bW^{\ge L}[f]  
\label{eq:potato2}
\end{align}
where we used Item~(ii) of~\Cref{lem:flat-polynomials-fine} and \Cref{eq:deviation-bound} to bound the first term,
and we bounded the second term using \Cref{lem:flat-polynomial-oos-bound}.
Taking the estimator $\calE$ to be
\begin{equation}
\label{eq:iamcale}
\calE\left( f\big|_{B(\bx,r)} \right) := \left|g \left( f\big|_{B(\bx,r)} \right) \right|,
\end{equation}
the lemma follows from \Cref{eq:Efgmatch} and \Cref{lem:low-var-expectation}.
\end{proof}

\subsection{The Sharp Noise Operator} \label{sec:sharpnoise}

We now turn to designing a ``sharp noise'' operator on a set of variables $V$, which we denote by $\smash{\shnoise_{\ell,\kappa,\Delta}^V}$ with three parameters $\ell,\kappa$ and $\Delta$. 
Intuitively, we would like this operator, when applied to a function $f:\{\pm 1\}^{k'}\rightarrow \mathbb{R}$, to zero out all Fourier coefficients $\smash{\hat{f}(S)}$ for which~$|S \cap V|>\ell$, and to leave unchanged all other Fourier coefficients.  
We do not achieve this exactly; instead, all Fourier coefficients for which $|S \cap V|\ge \kappa \ell$ are made ``very small'' (as captured by the $\Delta$ parameter), and all Fourier coefficients for which $|S \cap V|\le \ell$ are ``approximately preserved.''  
More precisely, the $\shnoise$ operator has the following properties: 

\begin{lemma}
\label{lem:sharp-noise}
Given any $V \subseteq [k']$ and three positive integers $\ell, \kappa,\Delta$ such that $\ell\in [k']$ and $\kappa$ $\ge 5$ (so that $e^{-\kappa/2}\kappa\le 0.5$), there exists an operator $\shnoise_{\ell,\kappa,\Delta}^V$ on $f:\{\pm 1\}^{k'}\rightarrow \mathbb{R}$ such that
	\[\shnoise_{\ell,\kappa,\Delta}^V \hspace{0.05cm}f = \sum_{S \subseteq [k']} \lambda(S) \wh{f}(S) \chi_S(x), \]
where each $\lambda(S) \in [0,1]$ and satisfies 
	\[\lambda(S) = \begin{cases} \geq 1 - \Delta 2^{-\kappa} & \text{if\ }|S \cap V| \leq \ell \\ \leq 2^{-\Delta} & \text{if\ }|S \cap V| \geq \kappa \ell \\ \end{cases}.
\]
Moreover, this operator can be written as
\begin{equation}
\label{eq:shnoise-explicit}
\shnoise_{\ell,\kappa,\Delta}^V = \sum_{i = 0}^{\kappa\Delta} \alpha_i \T^V_{\rho^i},\quad
\text{where $
\rho := 1 - {\frac 1 {2\ell}}\quad 
\text{and} \quad
\sum_{i=0}^{\kappa\Delta} |\alpha_i| \le 2^{2\kappa\Delta}.$}
\end{equation} 
For functions $f:\{\pm 1\}^{k'}\rightarrow [-1,1]$,
  we always have $$\Var\left[\shnoise^V_{\ell,\kappa,\Delta}\hspace{0.05cm} f\right]\le 1.$$
\end{lemma}

\begin{proof}
We design $\shnoise^V_{\ell,\kappa,\Delta}$ in stages using a few intermediate operators. 

First we consider the operator $U$, which is defined as

		\[Uf := f - \T_{\rho}^V f = \sum_{S \subseteq [k']} \left(1 - \rho^{|S \cap V|}\right) \wh{f}(S) \chi_S(x),
		\]
		with $\rho=1-1/(2\ell)$.
	Note that by our choice of $\rho$, we have that 
			\[\left(1 - \rho^{|S \cap V|}\right) \text{~is~} \begin{cases} \leq {1}/{2} & \text{if\ }|S \cap V| \leq \ell \\ \geq 1 - e^{-\kappa/2} & \text{if\ }|S \cap V| \geq \kappa \ell \\ 
			\in [0,1] & \text{otherwise} \end{cases}.\]
Next we consider the operator $U'$, which is defined as
		\[U'f := f - U^{\kappa} f = \sum_{S \subseteq [k']} \left(1 - \left(1 - \rho^{|S \cap V|}\right)^{\kappa}\right) \wh{f}(S) \chi_S(x), \]
	which has brought us closer to our goal as
		\[\left(1 - \left(1 - \rho^{|S \cap V|}\right)^{\kappa}\right) \text{~is~} \begin{cases} \geq 1 - 2^{-\kappa} & \text{if\ }|S \cap V| \leq \ell \\ \leq 1/2 & \text{if\ }|S \cap V| \geq \kappa \ell \\ 
		\in [0,1] & \text{otherwise} \end{cases},\]
		where we used $e^{-\kappa/2}\kappa\le 1/2$.
Finally we  define $\shnoise$ by applying $U'$ repeatedly $\Delta$ times:
		\[\shnoise_{\ell,\kappa,\Delta}^V \hspace{0.05cm}f := (U')^{\Delta} f = \sum_{S \subseteq [k']}\left(1 - \left(1 - \rho^{|S \cap V|}\right)^{\kappa}\right)^{\Delta} \wh{f}(S) \chi_S(x).\]
		This has the desired properties, since
		\[ 
		\left(1 - \left(1 - \rho^{|S \cap V|}\right)^{\kappa}\right)^{\Delta}
		\text{~is~} \begin{cases} \geq 1 - \Delta 2^{-\kappa} & \text{if\ }|S \cap V| \leq \ell \\ \leq 2^{-\Delta} & \text{if\ }|S \cap V| \geq \kappa\ell \\ 
		\in [0,1] & \text{otherwise} \end{cases}. \]
	
To obtain \Cref{eq:shnoise-explicit} about writing $\shnoise_{\ell,\kappa,\Delta}^V$ as a linear combination of Bernoulli noise operators over $V$ with different noise rates, observe that
\[
\left(1 - \left(1 - \rho^{|S \cap V|}\right)^{\kappa}\right)^{\Delta}
\]
can be written as $p(\rho^{|S \cap V|})$ where $$p(x)=\left(1-\left(1-x\right)^\kappa\right)^\Delta=\sum_{i=0}^{\kappa\Delta} \alpha_i x^i$$ is a univariate polynomial of degree at most $\kappa\Delta$
with  $\sum_i |\alpha_i| \le 2^{2\kappa\Delta}.$
This gives \Cref{eq:shnoise-explicit}.  
\end{proof}

\subsection{Parameters in the main algorithm}

Before presenting the main algorithm, let's list parameters that will be used for $\shnoise$:
$$
\ell= {k}^{2/3},\quad
\rho=1-\frac{1}{2\ell},\quad 
\kappa=10\log\left(\frac{k'}{\eps}\right)\quad\text{and}\quad\Delta=10r\log\left(\frac{k'}{\eps}\right)$$
and parameters that will be used for the local estimator:
$$
L=\kappa\ell,\quad  
\tau=\frac{\eps}{10}\quad\text{and}\quad\text{and}\quad r=\Theta\left(\sqrt{L\log L\log (1/\tau)}\right).
$$
In the rest of the section, $\calE$ denotes the $r$-local
  estimator given in \Cref{lem:smooth-estimator} with parameters $\tau$ and $L$.
Another parameter needed is
\begin{equation} \label{eq:N}
N:={k'}^{O(r)}\cdot 2^{O(\kappa\Delta)}\cdot \frac{k'^2}{\eps^2}\le \exp\left(k^{1/3} \cdot \polylog\left(\frac{k'}{\eps}\right)\right).
\end{equation}

We note that while we take $\ell$ to be $k^{2/3}$ throughout our quantum testing algorithm, we will think of it as a parameter throughout \Cref{sec:local-est-junta-corr}, as we will need to set $\ell = k^{2/3} \polylog(k'/\eps)$ in our classical tester.

\subsection{Local estimation of junta correlations} \label{sec:local-est-junta-corr}

For our quantum junta tester we only need the results of this section for functions $f:\{\pm 1\}\rightarrow \bits,$ but in the next section we will need these results for functions $f:\{\pm 1\}\rightarrow [-1,1]$; hence we handle this more general setting in this subsection.

Given a $k$-subset $U\subseteq [k']$ and $C\subseteq U$, a subroutine is needed 
  in the main algorithm to estimate
$$
\corr\left(f^{\overline{C}},\calJ_U\right),\quad\text{where}\quad
f^{\overline{C}}:=\shnoise^{\overline{C}}_{\ell,\kappa,\Delta}\hspace{0.05cm} f.
$$
Recalling \Cref{eq:corr,eq:corfJC}, we can express the correlation as 
\begin{equation} \label{eq:corr-expr}
\corr\left(f^{\overline{C}},\calJ_U\right)=\Ex_{\by\sim \{\pm 1\}^U}\left[\left| \Ex_{\bz\sim \{\pm 1\}^{\overline{U}}}\left[ 
f^{\overline{C}}_{U\rightarrow \by}(\bz)\right]\right|\right].
\end{equation}
We first show that the following is a good estimation of the correlation:
\begin{equation} \label{eq:iamest}
\Est:=\Ex_{\bx\sim \bits^{k'}} \left[\calE\left(f^{\overline{C}}_{U\rightarrow \bx_U}\big|_{B(\bx_{\overline{U}},r)}\right)\right].
\end{equation}

\begin{lemma}\label{lem:closeness}
Let $f:\{\pm 1\}^{k'}\rightarrow [-1,1]$. For any $k$-subset $U\subseteq [k']$ and $C\subseteq U$, we have 
$$
\left|\Est-\corr\left(f^{\overline{C}},\calJ_U\right)\right|\le 2\tau.
$$
\end{lemma}
\begin{proof}
To align with \Cref{eq:corr-expr}, for clarity we can rewrite $\Est$ as 
$$
\Ex_{\by\sim \bits^U }\left[\Ex_{\bz\sim \bits^{\overline{U}}} \left[\calE\left(f^{\overline{C}}_{U\rightarrow \by}\big|_{B(\bz,r)}\right)\right]\right].
$$
By \Cref{lem:smooth-estimator}, we have that for every $y\in \{\pm 1\}^U$,
$$
\Ex_{\bz} \left[\calE\left(f^{\overline{C}}_{U\rightarrow y}\big|_{B(\bz,r)}\right)\right]= \left| \Ex_{\bz}\left[ 
f^{\overline{C}}_{U\rightarrow y}(\bz)\right]\right|\pm \left(\tau\sqrt{\Var\left[f^{\overline{C}}_{U\rightarrow y}\right]}+ 5(k')^r \sqrt{\bW^{\ge L}\left[f^{\overline{C}}_{U\rightarrow y}\right]}\right),
$$
so taking expectation over $\by$ and recalling \Cref{eq:corr-expr}, we have that
\[
\Est = \corr(f^{\overline{C}},\calJ_U) \pm \Ex_{\by \sim \bits^U} \left[
\left(\tau\sqrt{\Var\left[f^{\overline{C}}_{U\rightarrow \by}\right]}+ 5(k')^r \sqrt{\bW^{\ge L}\left[f^{\overline{C}}_{U\rightarrow \by}\right]}\right)
\right].
\]
We proceed to bound the two contributions to the error term above.  For the first, we have
$$
\Ex_{\by\sim \{\pm 1\}^U}\left[\sqrt{\Var\left[f^{\overline{C}}_{U\rightarrow \by}\right]}\right]
\le \sqrt{\Ex_{\by}\left[\Var\left[f^{\overline{C}}_{U\rightarrow \by}\right]\right]
} \leq \sqrt{\Var\left[f^{\overline{C}}\right]
}\le 1
$$
where the first inequality is Cauchy-Schwarz, the second is the law of total variance, and the third is because $f$'s outputs lie in $[-1,1]$.
For the second, we have
$$
\Ex_{\by} \left[\sqrt{\bW^{\ge L}\left[f^{\overline{C}}_{U\rightarrow \by}\right]}\right]
\le \sqrt{\Ex_{\by} \left[
\bW^{\ge L} \left[f^{\overline{C}}_{U\rightarrow \by}\right]\right]
}= 
\sqrt{\sum_{R: |R \setminus U| \geq L}
 \widehat{f^{\overline{C}}}(R)^2 
}
$$
where the first inequality again is Cauchy-Schwarz and the second is by an application of Proposition~3.22 of \cite{odonnell-book}.
 
Since $C \subseteq U$, every $R$ such that $|R \setminus U| \geq L$ also satisfies $|R\setminus C| \geq L$. 
As a result, we have 
$$
\sum_{R: |R \setminus U| \geq L} \widehat{f^{\overline{C}}}(R)^2 
\le
\sum_{R: |R \setminus C| \geq L} \widehat{f^{\overline{C}}}(R)^2 
\le 2^{-2\Delta}\cdot \sum_{R}\widehat{f}(R)^2\le 2^{-2\Delta}
$$
where the second inequality is by \Cref{lem:sharp-noise} and the third is because
 $f:\{\pm 1\}^{k'}\rightarrow [-1,1]$.
Combining everything, we have 
$$
\left|\Est-\corr\left(f^{\overline{C}},\calJ_U\right)\right|\le {\tau + 5(k')^r\cdot 2^{-\Delta}\le 2\tau},
$$
by the choice of $\Delta$. 
\end{proof}

Given the above lemma, it suffices to estimate $\Est$.
For this we introduce a notion of a ``sample bundle'' for a 
given point $x$:

\begin{definition} \label{def:sample-bundle}
Given 
$x\in \{\pm 1\}^{k'}$, a \emph{sample bundle $\calB$ for $x$} is a multiset of ${k' \choose \leq r} \cdot (\kappa \Delta + 1)$ many points  
  $\calB(y,i)\in \{\pm 1\}^{k'}$, where $y$ ranges over $B(x,r)$ and $i$ ranges over $[0:\kappa\Delta]$.
 \end{definition}
We write $\calD_{x,C}$ to denote the following distribution of sample bundles
  $\calB$ for $x$:
\begin{quote}
For each $y\in B(x,r)$ and $i\in [0:\kappa\Delta]$, draw independently a sample
  $\calB(y,i)\sim N_{\rho^i}^{\overline{C}}(y)$.
\end{quote}

Now we give a subroutine which, given $C \subseteq U$ and sample bundles for
  $N$ points, 
  computes a high-confidence, high-accuracy estimate of $\corr(f^{\overline{C}},\calJ_U)$ 
  when all the points and their sample bundles are drawn correctly.
In addition, the subroutine queries $f$ on
  points only in the sample bundles.
Looking ahead, it is crucial that the way these samples are drawn is independent
  of $U$ and thus they can be reused to estimate $\corr(f^{\overline{C}},\calJ_U)$ for all $k$-subsets $U$ of $[k']$ that contain $C$.

\begin{lemma}\label{lem:juntaestlemma}
Let $\calA$ be a 1-bounded randomized algorithm for $h:  \bits^{k'} \to [-1,1].$
There exists an
algorithm which, 
  given any $k$-subset $U$, $C\subseteq U$, $\smash{x^{(1)},\ldots,x^{(N)}\in \bits^{k'}}$ and 
  a sample bundle $\calB^{(i)}$ for $x^{(i)}$ for each $i\in [N]$,
only calls $\calA$ on points in $\smash{\calB^{(1)},\ldots,\calB^{(N)}}$ and returns a number ${\boldEst'}$.
When $\x{1},\ldots,\x{N}\sim \bits^{k'}$ and $\smash{\calB^{(i)}\sim \calD_{\bx^{(i)},C}}$ independently, we have $$\left|{\boldEst'}-\corr\left(h^{\overline{C}},\calJ_U\right)\right|\le \eps/2$$ with probability at least 
  $1-2^{-(k'/\eps)^2}$.
  Moreover, the number of calls to $\calA$ that are made by the algorithm is at most $N \cdot \poly((k')^{r},\kappa \Delta,1/\eps) \leq \exp\left(k^{1/3} \cdot \polylog\left(\frac{k'}{\eps}\right)\right).$
  \end{lemma}
\begin{proof}
Given \Cref{lem:closeness} (and the choice of $\tau=\eps/10$), it suffices to give a $(\eps/3)$-estimation of $\Est$.

By \Cref{eq:iamest} and \Cref{eq:iamcale}, we have that
\[
\Est = 
\Ex_{\bx\sim \bits^{k'}} \left[\calE\left(h^{\overline{C}}_{U\rightarrow \bx_U}\big|_{B(\bx_{\overline{U}},r)}\right)\right] 
=
\Ex_{\bx\sim \bits^{k'}} \left[ 
\left|g \left(h^{\overline{C}}_{U \to \bx_{U}}\big|_{B(\bx_{\overline{U}},r)} \right) \right|
\right].
\]
Recalling \Cref{eq:neededlater1,eq:neededlater2,eq:Trhodef,eq:shnoise-explicit}, we get that
\[
\Est = 
\Ex_{\bx\sim \{\pm 1\}^{k'}}\left[\left| \sum_{y\in B(\bx_{\overline{U}},r)}\sum_{i=0}^{\kappa \Delta} \alpha_{y,i}\cdot \E[h(\bz_{y,i,\bx})]\right|\right],
\]
where each $\alpha_{y,i}$ is a fixed coefficient that satisfies $|\alpha_{y,i}|\le (k')^{O(r)}\cdot 2^{2\kappa\Delta}$ and each $\bz_{y,i,\bx}$ is an independent draw from $\smash{N_{\rho^i}^{\overline{C}}(\bx_U\circ y)}$.  Given that $y\in B(\bx_{\overline{U}},r)$, we have $\bx_U\circ y\in B(\bx,r)$.  Thus, we define the output of our algorithm (the estimator $\boldEst'$ for $\Est$) to be 
\begin{equation} \label{eq:estprime}
\boldEst' = {\frac 1 N} \sum_{j=1}^N \boldEst'_j(x^{(j)}), \text{~where~}
\boldEst'_j (x^{(j)})=
\left| \sum_{y\in B(x^{(j)}_{\overline{U}},r)}\sum_{i=0}^{\kappa \Delta} \alpha_{y,i}\cdot \calA_{\eps',\delta'}\left(\calB^{(j)}(x^{(j)}_U \circ y, i)\right)\right|,
\end{equation}
where $\calA_{\eps',\delta'}$ is the $\eps'$-accuracy, $\delta'$-confidence algorithm for estimating $h$ that is obtained from the randomized algorithm  $\calA$ for $h$ as described in \Cref{claim:estimation-of-bounded-function}.
(Note that this algorithm indeed only calls $\calA$ on points in the sample bundles $\calB^{(1)},\dots,\calB^{(N)}$ as claimed.) Here $\eps',\delta'$ are parameters that we will set later (see \Cref{eq:epsprimedeltaprime}).

It remains to argue that when $\x{1},\ldots,\x{N}\sim \bits^{k'}$ and $\smash{\calB^{(i)}\sim \calD_{\bx^{(i)},C}}$ independently, we have 
\begin{equation}
\label{eq:boldestprimegood}
\left|{\boldEst'}-\Est\right| \leq \eps/3 \text{ except with failure probability at most~}
2^{-(k'/\eps)^2}.
\end{equation}
Towards this end, we have the following claim:

\begin{claim} \label{claim:a} 
For each $j \in [N]$, with probability 1 we have that 
\[
|\boldEst'_j(\bx^{(j)})| \leq E_{\max} := \poly((k')^r,2^{\kappa \Delta}) 
\]
and
\[
\left|
\E[\boldEst'_j(\bx^{(j)})] - \Est
\right| 
\leq {\frac \eps 6}.
\]
\end{claim}
\begin{proof}
The bound on the magnitude of $\boldEst'_j(\bx^{(j)})$ follows directly from the definition of $\boldEst'_j$ given in \Cref{eq:estprime}, the coefficient bounds $|\alpha_{y,i}| \leq (k')^{O(r)} \cdot 2^{2 \kappa \Delta}$, and the absolute bound on the output of $\calA_{\eps',\delta'}$ given by \Cref{claim:estimation-of-bounded-function}.

We turn to bounding $\left| \E[\boldEst'_j(\bx^{(j)})] - \Est \right|.$ 
Let us define
\[
V(\bx^{(j)}) := 
\left| 
\sum_{y\in B(\bx^{(j)}_{\overline{U}},r)}\sum_{i=0}^{\kappa \Delta} \alpha_{y,i}\cdot \Ex_{\bz_{y,i,\bx^{(j)}} \sim N_{\rho^i}^{\overline{C}}(\bx^{(j)}_U\circ y)}[h(\bz_{y,i,\bx^{(j)}})]
\right|,
\quad \text{so~}\Est = \Ex_{\bx \sim \bits^{k'}}[V(\bx)].
\]
We have
\[
\left|
\E[\boldEst'_j(\bx^{(j)})] - \Est
\right|
=
\left|
\E[\boldEst'_j(\bx^{(j)})] - \E[V(\bx^{(j)})]
\right|
\]
so our goal is to prove that 
\[
\left|
\E[\boldEst'_j(\bx^{(j)})] - \E[V(\bx^{(j)})]
\right|
\leq {\frac \eps 6}.
\]
We have
\begin{align*}
&\left|
\E[\boldEst'_j(\bx^{(j)})] - \E[V(\bx^{(j)})]
\right|\\
&=
\left|
\E \left[
\left| \sum_{y\in B(\bx^{(j)}_{\overline{U}},r)}\sum_{i=0}^{\kappa \Delta} \alpha_{y,i}\cdot \calA_{\eps',\delta'}\left(\calB^{(j)}(\bx^{(j)}_U \circ y, i)\right)\right|
-
\left| 
\sum_{y\in B(\bx^{(j)}_{\overline{U}},r)}\sum_{i=0}^{\kappa \Delta} \alpha_{y,i}\cdot \Ex_{\bz_{y,i,\bx^{(j)}} \sim N_{\rho^i}^{\overline{C}}(\bx^{(j)}_U\circ y)}[h(\bz_{y,i,\bx^{(j)}})]
\right|
\right]
\right|\\
&\leq
\E\left[
\sum_{y\in B(\bx^{(j)}_{\overline{U}},r)}\sum_{i=0}^{\kappa \Delta} |\alpha_{y,i}| \cdot 
\left|
 \calA_{\eps',\delta'}\left(\calB^{(j)}(\bx^{(j)}_U \circ y, i)\right)
 - \Ex_{\bz_{y,i,\bx^{(j)}} \sim N_{\rho^i}^{\overline{C}}(\bx^{(j)}_U\circ y)}[h(\bz_{y,i,\bx^{(j)}})]
 \right|
\right].
\end{align*}
Recalling the distribution of our sample bundles, we have that
for each $y \in B(\bx^{(j)}_{\overline{U}},r), i \in [0:\kappa \Delta]$ the distribution of $\calB^{(j)}(\bx^{(j)}_U \circ y,i)$ is identical to the distribution of $\bz_{y,i,\bx^{(j)}} \sim N^{\overline{C}}_{\rho^i}(\bx^{(j)}_U \circ y)$ (in fact, this is true on an outcome-by-outcome basis for $\bx^{(j)}$). 
Hence we may rewrite the above as
\begin{align*}
&
\E\left[
\sum_{y\in B(\bx^{(j)}_{\overline{U}},r)}\sum_{i=0}^{\kappa \Delta} |\alpha_{y,i}| \cdot 
\left|
 \calA_{\eps',\delta'}\left(\calB^{(j)}(\bx^{(j)}_U \circ y, i)\right)
 - \Ex_{\calB^{(j)}}[h(\calB^{(j)}(\bx^{(j)}_U \circ y, i))]
 \right|
\right]\\
&\leq
\Ex_{\bx^{(j)}}\left[
\sum_{y\in B(\bx^{(j)}_{\overline{U}},r)}\sum_{i=0}^{\kappa \Delta} |\alpha_{y,i}| \cdot 
\Ex_{\calB^{(j)}}
\left[
\left|
 \calA_{\eps',\delta'}\left(\calB^{(j)}(\bx^{(j)}_U \circ y, i)\right)
 - h(\calB^{(j)}(\bx^{(j)}_U \circ y, i))
 \right| 
 \right]
 \right].
\end{align*}
Now, since on every input $\calA_{\eps',\delta'}$ is a $\pm \eps'$-accurate estimate of $h$ except with failure probability $\delta'$, and $\calA_{\eps',\delta'}$ and $h$ both always output values in $[-1,1],$ by \Cref{claim:estimation-of-bounded-function}, the above is at most
\begin{equation} \label{eq:la}
{k' \choose \leq r} \cdot 2^{2\kappa \Delta} \cdot \left(2\delta' + \eps' \right).
\end{equation}
Choosing
\begin{equation} \label{eq:epsprimedeltaprime}
\eps' = \eps \cdot \poly\left({\frac 1 {(k')^r}}, {\frac 1 {2^{\kappa \Delta}}} \right),
\delta' = \eps \cdot \poly\left({\frac 1 {(k')^r}}, {\frac 1 {2^{\kappa \Delta}}} \right),
\end{equation}
we get that \Cref{eq:la} is at most $\eps/6$ as desired.
\end{proof}

With \Cref{claim:a} in hand we can prove \Cref{eq:boldestprimegood}.
By a Hoeffding bound, using the fact (\Cref{claim:a}) that each  $|\boldEst'_j(\bx^{(j)})| \leq E_{\max}$, we have that
\begin{align*}
\mathbf{Pr} \left[ \left| {\frac 1 N} \sum_{j=1}^N \boldEst'_j(\bx^{(j)}) - \mathbf{E}[\boldEst'_j(\bx^{(j)})] \right| \ge {\frac \eps 6} \right]
&=
\mathbf{Pr} \left[ \left| \boldEst' - \mathbf{E}[\boldEst'_j{(\bx^{(j)})}] \right| \ge {\frac \eps 6} \right]\\
& \le 2 \exp\left( \frac{-2N ({\frac \eps 6})^2}{(E_{\max})^2} \right)\\
& \le 2^{-(k'/\eps)^2},
\end{align*}
where the last inequality holds by our choice of $N$ (recall \Cref{eq:N}).
On the other hand, note that when $ \left| \boldEst' - \mathbf{E}[\boldEst'_j{(\bx^{(j)})}] \right| \leq {\frac \eps 6}$, recalling that \Cref{claim:a}  gives us
$\left|
\E[\boldEst'_j(\bx^{(j)})] - \Est
\right| 
\leq {\frac \eps 6},$ by the triangle inequality we have that
\[
\left| \boldEst' - \Est \right| \leq \eps/3,
\]
which completes the proof of \Cref{eq:boldestprimegood}.

It remains only to bound the number of calls to $\calA$. 
Recalling \Cref{eq:estprime} and \Cref{claim:estimation-of-bounded-function}, this is easily seen to be at most $N \cdot \poly((k')^{r},\kappa \Delta,1/\eps)$, and the lemma is proved.
\end{proof}

\begin{remark} \label{rem:rand-alg-oracle}
In the next section, when we use \Cref{lem:juntaestlemma} on a function $h$ with real-valued outputs, in order to bound the query complexity we will need to account for the number of oracle calls that the randomized algorithm $A$ for $h$ makes to the underlying function $f: \bits^n \to \bits$ (the function that is actually being tested and for which we have black-box oracle access) in the course of its execution on a given input.  In the context of this section, though, we will only need \Cref{lem:juntaestlemma} for $h=f$, the original $\bits$-valued function that is being tested. In this case there is no need for the randomized algorithm $A$ and the number of calls to $f$ that are required is simply the number of points in the sample bundles $\calB^{(1)},\dots,\calB^{(N)}$, i.e.~at most $N \cdot {k' \choose \leq r}$.
\end{remark}

\subsection{Putting It Together: A Quantum Tester and the Proof of \Cref{thm:quantum}}

Our quantum tolerant junta tester is given in \Cref{alg:quantum}.  

\begin{algorithm}[t!]
\addtolength\linewidth{-2em}

\vspace{0.5em}

\textbf{Input:} A Boolean function $f: \bits^{k'} \rightarrow \bits$ and $\eps \in [0,1/2]$ \\[0.25em]
\textbf{Output:} An estimate $\gamma$ of $\corr(f, \calJ_k)$

\vspace{0.5em}

\qtester:

\vspace{0.5em}

\begin{enumerate}
	\item Set $\gamma =0$ and   
	draw $\smash{\bS_1, \dots, \bS_{(k'/\eps)^4}}$ independently from the spectral sample of $f$.  
	\item Repeat $\smash{(1/\eps)^{O(k^{1/3})}}$ times:
		\begin{enumerate}
			\item Sample $\bA_1, \dots, \bA_{k^{1/3}} \subseteq [k']$ independently from the spectral sample of $f$.
			\item Let $\bC$ be the union of $\bA_i$'s and define
			\begin{equation*}
			f^{\overline{\bC}} := \shnoise^{\overline{\bC}}_{\ell,\kappa,\Delta}\hspace{0.05cm}f.
			\end{equation*}
			\item Sample $\x{1}, \dots, \x{N} \sim \bits^{k'}$ uniformly and independently at random.
			\item For each $\x{i}$, draw a sample bundle $\calB^{(i)}\sim \calD_{\x{i},\bC}$.
			\item For each set $U \subseteq [k']$ of size $k$ such that $\bC \subseteq U$ and
\begin{equation} \label{eq:few}
				{\frac {\big|i: \bS_i \subseteq U\ \text{and}\ |\bS_i\setminus \bC|\ge \ell\big|}{(k'/\eps)^4}}\le \frac{\eps^2}{10}, 
\end{equation}
			 run the algorithm of \Cref{lem:juntaestlemma} using $\x{1},\ldots,\x{N},\calB^{(1)},\ldots,\calB^{(N)}$ to get an\\
			  $(\eps/2)$-estimation $\boldEst_{\bC,U}$ 
			 of $\corr(f^{\overline{\bC}},\calJ_U)$; 
   update $\gamma = \max(\gamma, \boldEst_{\bC,U})$. 
	
	\end{enumerate}
	\item Return $\gamma$.

\end{enumerate}
\caption{A Quantum Tolerant Junta Testing Algorithm}
\label{alg:quantum}
\end{algorithm}

\subsubsection{Efficiency} \label{sec:quantum-efficiency}

\begin{lemma} [Query complexity]
\label{lem:quantum-efficiency}
 \Cref{alg:quantum} makes $\exp(k^{1/3} \polylog(k'/\eps))$ 
 (quantum) queries. 
\end{lemma}
\begin{proof}
We first account for the queries to $f$ that are ``actually quantum.''  As discussed in \Cref{sec:spectral-sample} and \Cref{rem:quantum}, these correspond to the draws from the spectral sample ${\cal P}_f$ of $f$ in lines~1 and~2(a); since each draw from ${\cal P}_f$ requires one quantum query, by inspection of \Cref{alg:quantum}, it makes at most 
$\smash{(1/\eps)^{O(k^{1/3})}}$ oracle calls that are ``actually quantum''.

We proceed to bound the number of classical queries to $f$ made by \Cref{alg:quantum}.  The only step of \Cref{alg:quantum} that make classical oracle calls to $f$ is 2(e), 
  where points in the bundles $\calB^{(1)},\ldots,\calB^{(N)}$ are queried.
  As remarked before \Cref{lem:juntaestlemma}, for each outcome of $\bC$ we reuse the 
  points for all the sets $U \subseteq [k']$ of size $k$ that contain $\bC$ that are considered in line~2(e).  So the total number of classical queries to $f$ that are made, recalling \Cref{rem:rand-alg-oracle}, is
$$
\left(\frac{1}{\eps}\right)^{O(k^{1/3})}\cdot N\cdot {k'\choose \leq r}\cdot O(\kappa\Delta)
=\exp\left(k^{1/3}\cdot\polylog(k'/\eps)\right).\qedhere
$$
\end{proof}

\subsubsection{Correctness} \label{sec:quantum-correctness}

We will refer to \Cref{eq:few} in the algorithm as a \emph{test} on the pair of $\bC$ and $U$.
We say $\bC$ and $U$ \emph{pass the test} if \Cref{eq:few} holds.

The two main results of this section are \Cref{lem:every-set-accurate} and \Cref{lem:jackpot}. 
We start with \Cref{lem:every-set-accurate}, which states that with probability at least $0.99$,
  every pair of $\bC$ and $U$ that passes the test during the execution of \Cref{alg:quantum} satisfies 
  that $\corr(f,\calJ_U)$ and $\corr(f^{\overline{C}},\calJ_U)$ are close.

\begin{lemma} \label{lem:every-set-accurate}
With probability at least $99/100$, we have the following: For every pair $\bC$ and $U$ in step 2(e) that passes the \Cref{eq:few} test, it is the case that
\begin{equation} 
\label{eq:nevertoolarge}
\left|\boldEst_{\bC,U}-
  \corr\left(f, \calJ_U\right)\right|\le \eps /2. 
\end{equation}	
\end{lemma}
\begin{proof}
We start with a sufficient condition for the test to fail with high probability:
\begin{claim} \label{claim:wellmeetagain}
Let $U$ be a  $k$-subset of $[k']$ and $C\subseteq U$ be such that 
$$
\sum_{R\subseteq U: |R\setminus C|\ge \ell} \widehat{f}(R)^2\ge \frac{\eps^2}{5}.
$$
Then the test fails with probability at least $1-e^{-\Omega(k'^4/\eps^2)}$. 
\end{claim}
\begin{proof}
This just follows from a standard multiplicative Chernoff bound and the fact that there are $(k'/\eps)^4$ many $\bS_i$'s.
\end{proof}

Using the above claim, by a union bound 
we may assume that every pair of $\bC$ and $U$ that passes the test satisfies 
\begin{equation}\label{eq:assump}
\sum_{R\subseteq U: |R\setminus C|\ge \ell} \widehat{f}(R)^2\le \frac{\eps^2}{5}
\end{equation}
(at the cost of a $o(1)$ failure probability).
Fix any such pair $C$ and $U$.
Recall from \Cref{eq:corfJC} that
\begin{equation} \label{eq:bokchoi}
\corr\left(f,\calJ_U\right)=
\Ex_{\bx \sim \bits^{k'}}\left[ \left| \sum_{R \subseteq U} \widehat{f}(R) \chi_R(\bx) \right|\right]  \text{~and~} 
\corr\left(f^{\overline{C}},\calJ_U\right)=
\Ex_{\bx \sim \bits^{k'}}\left[ \left| \sum_{R \subseteq U} \widehat{f^{\overline{C}}}(R) \chi_R(\bx) \right|\right].
\end{equation}
So we have
\begin{align}
\left|\corr\left(f,\calJ_U\right)-\corr\left(f^{\overline{C}},\calJ_U\right)\right|
&\le \Ex_{\bx \sim \bits^{k'}}\left[ \left| \sum_{R \subseteq U} \left(\widehat{f}(R)-\widehat{f^{\overline{C}}}(R)\right) \chi_R(\bx) \right|\right] \nonumber \\
&\le \sqrt{\Ex_{\bx \sim \bits^{k'}}\left[ \left( \sum_{R \subseteq U} \left(\widehat{f}(R)-\widehat{f^{\overline{C}}}(R)\right) \chi_R(\bx) \right)^2\right]} \nonumber \\
&=\sqrt{\sum_{R\subseteq U} \left(\widehat{f}(R)-\widehat{f^{\overline{C}}}(R)\right)^2}. \label{eq:gum}
\end{align}
We divide the sum inside the square root into two parts:
\begin{flushleft}\begin{enumerate}
\item $R\subseteq U$ with $|R\setminus C|=|R\cap \overline{C}|\le \ell$:
For this part we have from \Cref{lem:sharp-noise} that
$$
\sum_{R. \subseteq U: |R \setminus C| \leq \ell} \left(\widehat{f}(R)-\widehat{f^{\overline{C}}}(R)\right)^2\le \sum_{R \subseteq U: |R \setminus C| \leq \ell} \widehat{f}(R)^2\cdot (\Delta 2^{-\kappa})^2
\le  {(\Delta 2^{-\kappa})^2 \le \frac{\eps^2}{20}}
$$
by our choices of $\kappa$ and $\Delta$.
\item $R\subseteq U$ with $|R\setminus C|>\ell$:
For this part we have from \Cref{lem:sharp-noise} and \Cref{eq:assump} that 
$$
\sum_{R \subseteq U: |R \setminus C| > \ell}  \left(\widehat{f}(R)-\widehat{f^{\overline{C}}}(R)\right)^2\le \sum_{R \subseteq U: |R \setminus C| > \ell}  \widehat{f}(R)^2 
\le \frac{\eps^2}{5}.
$$
\end{enumerate}\end{flushleft}
As a result, we have $|\corr(f,\calJ_U)-\corr(f^{\overline{C}},\calJ_U)|\le \eps/2$.

By another union bound, using \Cref{lem:juntaestlemma} we may assume that (except with failure probability $o(1)$) for every $\bC$ and $U$ that pass the test, 
  we have $$\left|\boldEst_{\bC,U}-\corr\left(f^{\overline{\bC}},\calJ_U\right)\right|\le \eps/2.$$
As a result, we have 
$|\boldEst_{\bC,U}-\corr(f,\calJ_U)|\le \eps$ by combining the two inequalities.
\end{proof}

\Cref{lem:every-set-accurate} implies that with very high probability the algorithm never overestimates $\corr(f,\calJ_k)$:

\begin{corollary}\label{coro:1}
With probability at least $99/100$, the algorithm returns $\gamma \le \corr(f,\calJ_k)+\eps$.
\end{corollary}
\begin{proof}
For every pair that passes the test, 
we have $\boldEst_{\bC,U}\le \corr(f,\calJ_U)+\eps
\le \corr(f,\calJ_k)+\eps.$
\end{proof}

Next we show that with high probability, the algorithm returns $\gamma$ satisfying $\gamma\ge \corr(f,\calJ_k)-\eps$.
The correctness of the algorithm follows by combining \Cref{coro:1} and \Cref{lem:jackpot}.

\begin{lemma} \label{lem:jackpot}
With probability at least $98/100$, the algorithm returns 
  $\gamma \ge \corr(f, \calJ_k)-\eps$.
\end{lemma}
\begin{proof}
Let $U^*$ be a $k$-subset  that achieves the optimal correlation:
  $\corr(f,\calJ_{U^*})=\corr(f,\calJ_k)$.

We will assume in the rest of the proof of the lemma that   
  \begin{equation} \label{eq:R-not-super-light}
\sum_{R \subseteq U^*} \widehat{f}(R)^2 \geq \eps^2.
\end{equation}
If this is not the case, then it is easy to see that $\corr(f,\calJ_k)=\corr(f,\calJ_{U^*})\le \eps$, since by \Cref{eq:corfJC} and Cauchy-Schwarz we have
\[
\corr(f,\calJ_{U^*})=
\Ex_{\bx \sim \bits^{k'}}\left[ \left| \sum_{R\subseteq U^*} \widehat{f}(R) \chi_R(\bx) \right|\right]\le \sqrt{\sum_{R\subseteq U^*} \widehat{f}(R)^2}\le \eps.
\]
In this case the conclusion trivially holds given that the algorithm always 
  returns a nonnegative number.

It suffices to show that with probability at least $99/100$, 
  one of the $\bC$'s sampled satisfies $\bC\subseteq U^*$ and passes the test together with $U^*$.
If this is the case, then by a union bound with \Cref{lem:every-set-accurate},
  with probability at least $98/100$, this pair of $\bC$ and $U^*$ not only passes the test 
  but also leads to 
  $$\boldEst_{\bC,U}\ge \corr(f,\calJ_{U^*})-\eps\ge \corr(f,\calJ_k)-\eps.$$
For this purpose we consider the distribution of the spectral sample $\calP_f$ conditioned on the set that is drawn being a subset of $U^*$; we denote this conditioned distribution by $\calP_f^*$.
(Note that $\calP_f^*$ is well defined given \Cref{eq:R-not-super-light}.)

We prove the following claim:

\begin{claim}\label{finalquantumclaim}
Let $\bA_1,\ldots,\bA_{k^{1/3}}$ be independent samples from $\calP_f^*$, and $\bC$
  be their union.
Then 
\begin{equation}\label{eq:hehe10}
\Pr_{\bR\sim \calP_f^*}\left[|\bR\setminus \bC|\ge k^{2/3}\right]
\le \frac{\eps^2}{20}.
\end{equation}
with probability at least $\eps^{O(k^{1/3})}$.
\end{claim}

We prove \Cref{finalquantumclaim} at the end. 
Combining it with \Cref{eq:R-not-super-light},
we have that $\bC$ satisfies
  the condition given in \Cref{eq:hehe10} when $\bA_1,\ldots,\bA_{k^{1/3}}\sim \calP_f$ with probability 
  at least $$(\eps^2)^{k^{1/3}}\cdot \eps^{O(k^{1/3})}=\eps^{O(k^{1/3})}.$$
For such a $\bC$, the pair $\bC,U^*$ would pass the 
  \Cref{eq:few} test with high probability. Since the number of repetitions of the outer loop of line~2 (the number of $\bC$ we draw) is $\smash{(1/\eps)^{O(k^{1/3})}}$, the lemma follows.
  \end{proof}

\begin{proof}[Proof of \Cref{finalquantumclaim}]
The proof is similar to the proof of \Cref{commonlemma1}.

We draw the sets $\bA_1,\ldots,\bA_{k^{1/3}}\sim \calP_f^*$ one by one
  and always write $\bC$ to denote 
  the union of $\bA_1,\ldots,\bA_i$ sampled  so far after $i$ rounds.
Initially (before drawing $\bA_1$) we set $\bC=\emptyset$.

Initially we examine the probability of $\bR\sim \calP_f^*$ 
  with $\smash{|\bR\setminus \bC|=|\bR|\ge k^{2/3}}$.
If this probability is already less than $\smash{\eps^2/20}$, then we are done 
  (this is because as we draw more $\bA_i$'s, 
  their union $\bC$ can only grow and this probability can only decrease).
Assuming this probability is at least $\eps^2/20$, 
  then we pay a probability of $\eps^2/20$ to draw such a sample
  $\bA_1$ with $|\bA_1|\ge k^{2/3}$, and $\bC$ grows in size by at least $k^{2/3}$ after the first round.

For the second round,  we again consider the probability of $\bR\sim \calP_f^*$
  satisfying $\smash{|\bR\setminus \bC|\ge k^{2/3}}$.
Again if this probability is at most $\smash{\eps^2/20}$, then we are done.
Otherwise we pay a probability of $\smash{\eps^2/20}$ to draw a sample $\bA_2$
  such that $|\bA_2\setminus \bC|\ge k^{2/3}$ and again, $\bC$ grows in size by at least $k^{2/3}$.

We repeat the above argument. This can repeat no more than $k^{1/3}$ times since $U^*$ is of size $k$ 
  and thus $\bC\subseteq U^*$ can only grow by $k^{2/3}$ for at most $k^{1/3}$ rounds.
So within $k^{1/3}$ rounds, there must be a round after which
 the probability of $\bR\sim \calP_f^*$ 
  with $\smash{|\bR\setminus \bC|=|\bR|\ge k^{2/3}}$ is at most $\eps^2/20$.
The total probability we pay is $(\eps^2/20)^{k^{1/3}}$.
\end{proof}


\section{A $2^{\wt{O}(k^{1/3}))}$ Query Classical Tolerant Junta Tester} 
\label{sec:classical}

In this section we prove our main tolerant junta testing result, \Cref{thm:classical}.
Given \Cref{sec:quantum}, the main challenge is that in the classical setting, we no longer have access to draws from the spectral sample. 
The main challenge, as alluded in \Cref{sec:overview-classical}, is that we need to handle the case where there is more than $\eps^2$ mass above level $k^{2/3}$. To do this, we will use the notion of normalized influences and an extension of the  machinery developed in \cite{ITW21}.
 Recall that a draw from the level-$\ell$ normalized influences corresponds to first drawing a set $\bS$ from the spectral sample of $f$ conditioned on $|\bS| \geq \ell$, and then sampling a uniform random set $\bT \subseteq \bS$ of size $\ell$.

As sketched in \Cref{sec:overview-classical}, we will use these normalized influences to design and analyze our key technical tool.  This tool is a filtering procedure which, roughly speaking, takes as input a set $C \subseteq [k']$ of coordinates and outputs a collection $\calC$ of $2^{\wt{O}(k^{1/3})}$ subsets of $C$ such that one of the sets $C' \in \calC$ has half as many  ``irrelevant'' variables as $C$, but still contains all of the ``relevant'' variables that are likely to appear above level $k^{2/3}$.
Using this, we can start with $C = [k']$ and use our filtering procedure to get an initial collection of subsets $\calC$. We then run the filtering procedure on each element of $\calC$ (for each element of $\calC$, this results in a collection of subsets), and take the union of those collections to create a new collection $\calC'$. We repeat this (running the filtering procedure on each element of $\calC'$, etc.) $O(\log(k'))$ times; after doing this, we end up with a list of $2^{\wt{O}(k^{1/3})}$ elements (subsets of $[k']$) where with high probability one of them contains only relevant variables. We can then apply $\shnoise$ to all variables outside of this subset, just as we did in the quantum tester, and use our local estimators to determine the junta correlation. 

Before embarking on the proof, we first note that we will assume throughout that $\eps \geq 2^{-k^{0.001}}$. Note that this is essentially without loss of generality, since for smaller values of $\eps$ the $\polylog(1/\eps)$ factor in the exponent of \Cref{thm:classical} is bigger than $k$ (assuming a sufficiently large exponent). It then follows that the tolerant tester of \cite{DMN19}, which makes $2^k \cdot \poly(k,1/\eps)$ queries, has the claimed query complexity. 
We will also assume throughout that $k$ is at least some sufficiently large absolute constant (since otherwise we can use the \cite{DMN19} tester).

\subsection{Approximating Normalized Influences}

The filtering procedure that we describe and analyze in \Cref{sec:filtering} crucially uses the normalized influences of a suitable function $h$, which is based on the function $f$ that is being tolerantly tested.  
In \cite{ITW21}, Iyer et al.~gave an algorithm that  gives a modestly accurate estimate of the normalized influences; roughly speaking, their algorithm computes an approximation that is accurate to within a constant multiplicative error. 
This error in the \cite{ITW21} estimate would cause some complications in our analysis. 
To bypass these complications, in this subsection we give a procedure that computes an \emph{unbiased} estimator for the value of $\NormInf_U[f]$, for a given input set $U \subseteq [k']$, to within any desired accuracy $\epsilon$.
We remark that while the \cite{ITW21} approach to estimating $\NormInf_U[f]$ is based on random restrictions, in contrast, our procedure takes a quite different approach and is based on random sampling. We carefully design a particular random variable, show that it has expected value exactly $\NormInf_U[f]$ (see \Cref{lem:ninf-exp-good}), and give an algorithm (\Cref{alg:estimate-ninf}) which samples from that random variable and thus achieves an arbitrarily accurate approximation of $\NormInf_U[f]$.

\bigskip

\begin{algorithm}[H]
\addtolength\linewidth{-2em}

\vspace{0.5em}

\textbf{Input:} A function $f: \bits^{k'} \rightarrow [-B,B]$, a positive real number $B$,  a set $U \subseteq [k']$, an accuracy parameter $\eps$, and a failure probability parameter $\delta$  \\[0.25em]
\textbf{Output:} An estimate for $\NormInf_U[f]$

\vspace{0.5em}

\estninf:

\vspace{0.5em}

\begin{enumerate}
	\item $S \gets \emptyset$ 
	\item $g_S \gets f$
	\item For $u \in U$:
		\begin{enumerate}
			\item $g_{S \cup \{u\}} (x) = g_S(x) - \frac{1}{2} \left( g_S(x) + g(x_S^{\oplus u}) \right)$
			\item $S \gets S \cup \{u\}$
		\end{enumerate}
	\item $g(x) = g_U(x) \cdot \chi_U(x)$
	\item $M \gets 1000 (U!)^4 B^4 \cdot \frac{\log(1/\delta)}{\eps^2} $
	\item For $t \in [1, M]$:
	\begin{enumerate}
		\item Sample $\br_1, \br_2, \dots, \br_{|U|} \in [0,1]$ uniformly at random
		\item Sort the $\br_i$'s in decreasing order to get $\by_1 > \by_2 > \dots > \by_{|U|}$
		\item Set $\bgamma_t$ to be an estimate of $\E_{\bx} \left[ \left(T_{\sqrt{\by_{|U|}}} g(\bx) \right)^2 \right]$ evaluated via \Cref{lem:est-l2-bounded-function} with accuracy $\frac{\eps}{4(|U|!)^2}$ and failure probability $\frac{\eps}{100B^2(|U|!)^2}$.
	\end{enumerate}
	\item Output $(|U|!)^2 \cdot \frac{1}{M} \sum_{t = 1}^M \bgamma_t$
\end{enumerate}
\caption{An algorithm to approximate $\NormInf_U[f]$.}
\label{alg:estimate-ninf}
\end{algorithm}

\bigskip

The main lemma of this section is the following:

\begin{lemma}
\label{lem:estimate-ninf}
Let $B$ be a positive real number, $f : \{\pm1\}^{k'} \to [-B, B]$, $\varepsilon \in [0,1]$ and $\delta \in [0,1/4)$. Then with probability at least $1 - \delta$, we have that
\[
\left|\estninf(f,B,U,\varepsilon,\delta) - \NormInf_U[f] \right| \leq \varepsilon.
\]
Moreover, $\estninf(f,B,U,\varepsilon,\delta)$ makes at most $\poly(|U|!, B, 1/\varepsilon, \log(1/\delta))$ calls to the randomized algorithm $\calA$, where $\calA$ is a $B$-bounded randomized algorithm computing $f$.
\end{lemma}

\noindent
We start by noting a few simple properties of $g$. 

\begin{lemma}
\label{lem:properties-of-derivative}
For any $U \subseteq [k']$, the function $g$, as defined on line (4), satisfies
\[
g(x) = \sum_{S: S \supseteq U} \hat{f}(S) \chi_{S \setminus U}(x).
\]
Moreover, we have that
\[
g(x) = \chi_U(x) \cdot \sum_{S: S \subseteq U} \alpha_S f(x^{\oplus S})
\]
for some coefficients $(\alpha_S)_{S \subseteq U}$ which satisfy $\sum_{S \subseteq U} |\alpha_S| = 1$.
\end{lemma}

\begin{proof}
We will prove both claims by induction on $|U|$. To start, note that when $U = \emptyset$ we get that $g = f$, giving us the statement. We now assume that $|U| \ge 1$ and the statement is true for all sets of size $|U|-1$. Fix some $u \in U$, and note that by the inductive hypothesis, $g_{U \setminus \{u\}} \cdot \chi_{U \setminus \{u\}}$ is bounded in absolute value by $B$ and satisfies
\[
g_{U \setminus \{u\}}(x) \cdot \chi_{U \setminus \{u\}}(x) = \sum_{S: S \supseteq U \setminus \{u\}} \hat{f}(S) \chi_{S \setminus (U \setminus \{u\})}(x)
\]
Multiplying both sides by $\chi_{U \setminus \{u\}}$ then gives
\[
g_{U \setminus \{u\}}(x) = \sum_{S: S \supseteq U \setminus \{u\}} \hat{f}(S) \chi_S(x).
\]
We next note that
\begin{align*}
g_{U \setminus \{u\}}(x) + g_{U \setminus \{u\}}(x^{\oplus u}) 
&= \sum_{S: S \supseteq U \setminus \{u\}} \hat{f}(S) \chi_S(x) + \sum_{S: S \supseteq U \setminus \{u\}} \hat{f}(S) \chi_S(x^{\oplus u}) \\
&= 2 \sum_{S: S \supseteq U \setminus \{u\}, u \notin S} \hat{f}(S) \chi_S(x).
\end{align*}

This implies that
\[
g_U(x) = g_{U \setminus \{u\}}(x) - \frac{1}{2} \left( g_{U \setminus \{u\}}(x) + g_{U \setminus \{u\}}(x^{\oplus u}) \right) = \sum_{S: S \supseteq U} \hat{f}(S) \chi_S(x)
\]
which implies
\[
g(x) = g_U(x) \cdot \chi_U(x) = \sum_{S: S \supseteq U} \hat{f}(S) \chi_{S \setminus U}(x)
\]
as desired.

For the ``Moreover'' claim, we note that for any $x$ we have
\begin{align*}
g(x)\chi_U(x) = g_U(x) 
= \tfrac{1}{2} g_{U \setminus \{u\}}(x) - \tfrac{1}{2} g_{U \setminus \{u\}}(x^{\oplus u}).
\end{align*}
By the inductive hypothesis there exist coefficients $\beta_S$ for $g_{U \setminus \{u\}}$ such that $\sum_{S \subseteq U \setminus \{u\}} |\beta_S| = 1$. So we get that
\begin{align*}
g(x)\chi_U(x) 
= \frac{1}{2} \sum_{S: S \subseteq U \setminus \{u\}} \beta_S f(x^{\oplus S}) - \frac{1}{2} \sum_{S: S \subseteq U \setminus \{u\}} \beta_S f(x^{\oplus S \cup \{u\}})
\end{align*}
Since the above coefficients have absolute values that sum to $1$, we have completed the proof of the inductive step and the lemma.
\end{proof}

Now we come to the key observation underlying our estimator:
\begin{lemma}
\label{lem:ninf-exp-good}
\[
\Ex_{\by_{|U|}} \left[ \Ex_{\bx} \left[ \left( T_{\sqrt{\by_{|U|}}} g(\bx) \right)^2 \right] \right] = \frac{1}{(|U|!)^2} \NormInf_U[f].
\]
\end{lemma}

\begin{proof}
We start by noting that the vector of sorted values $(\by_1, \dots, \by_{|U|})$ is distributed uniformly over the region $A := \{(a_1, \dots, a_{|U|})  \in [0,1]^{|U|}: 1 \ge a_1 \ge a_2 \ge \cdots \ge a_{|U|} \ge 0 \}$ (see e.g.~Lemma 3 of \cite{BentleySaxe}). Thus, we have
\begin{align*}
& \Ex_{\by_{|U|}} \left[ \Ex_{\bx} \left[ \left( T_{\sqrt{\by_{|U|}}} g(\bx) \right)^2 \right] \right]\\
 &= \frac{1}{|U|!} \int_0^1 \int_0^{a_1} \int_0^{a_2} \cdots \int_0^{a_{|U|-1}} \Ex_{\bx} \left[ \left( T_{\sqrt{a_{|U|}}} g(\bx) \right)^2 \right] da_{|U|} \cdots da_1 \\
&= \frac{1}{|U|!} \int_0^1 \cdots \int_0^{a_{|U|-1}} \Ex_{\bx} \left[ \left( \sum_{S: S \supseteq U} \sqrt{a_{|U|}}^{|S| - |U|} \hat{f}(S) \chi_{S \setminus U}(\bx) \right)^2 \right] da_{|U|} \cdots da_1 \\
&= \frac{1}{|U|!} \int_0^1 \cdots \int_0^{a_{|U|-1}} \sum_{S: S \supseteq U} \left( a_{|U|}^{|S| - |U|} \hat{f}(S)^2 \right) da_{|U|} \cdots da_1 \\
&= \sum_{S: S \supseteq U} \frac{1}{|U|!} \hat{f}^2(S) \int_0^1 \cdots \int_0^{a_{|U|-1}} a_{|U|}^{|S| - |U|} da_{|U|} \cdots da_1 \\
&= \sum_{S: S \supseteq U} \frac{1}{|U|!} \hat{f}^2(S) \cdot \frac{1}{(|S| - |U| + 1)(|S| - |U| + 2) \cdots |S|} \\
&= \sum_{S: S \supseteq U} \frac{1}{(|U|!)^2} \binom{|S|}{|U|}^{-1} \hat{f}^2(S) \\
&= \frac{1}{(|U|!)^2} \NormInf_U[f]
\end{align*}

as desired.
\end{proof}

It now simply remains to argue that we indeed approximate this expectation with high accuracy, while making not too many calls to $\calA$. To do this, we first show that there are $B$-bounded randomized algorithms for $g$ and $T_{\sqrt{y_{|U|}}}g$:

\begin{lemma}
\label{lem:alg-for-derivative}
There are $B$-bounded randomized algorithms for $g$ and for $T_{\sqrt{\by_{|U|}}} g$. Moreover both algorithms call $\calA$ at most $2^{|U|}$ times, where $\calA$ is the $B$-bounded randomized algorithm for $f$.
\end{lemma}

\begin{proof}
Note that by \Cref{lem:properties-of-derivative}, we have that there are coefficients $\alpha_S$ with $\sum_S |\alpha_S| = 1$ such that
	\[
	g(x) = \chi_U(x) \cdot \sum_{S: S \subseteq U} \alpha_S f(x^{\oplus S}).
	\]	
Given an $x \in \bits^{k'}$, our randomized algorithm $\calA'$ for $g$ simply outputs
	\[
	\chi_U(x) \cdot \sum_{S: S \subseteq U} \alpha_S \calA(x^{\oplus S}).
	\]	
It is easy to see that this computes $g$ in expectation, and that it makes at most $2^{|U|}$ queries to $\calA$, as desired. 
Moreover, since $\sum_{S \subseteq U} |\alpha_S| = 1$, we see that $\calA'$ is also $B$-bounded. 

Finally, since there is an algorithm $\calA'$ for $g$, by \Cref{lem:alg-for-noise} we have that there is a $B$-bounded randomized algorithm for $T_{\sqrt{\by_{|U|}}} g$ that makes the same number of queries to $\calA$.
\end{proof}

\begin{lemma}
\label{lem:ninf-approx-good}
In any iteration $t \in [1,M]$ of the loop on line 6 of the algorithm, we have that $0 \leq \bgamma_t \le B^2$ and
\[
\left| \mathbf{E}[\bgamma_t] - \mathbf{E} \left[ \left( T_{\sqrt{\by_{|U|}}} g(\mathbf{x}) \right)^2 \right] \right| \le \frac{\varepsilon}{2(|U|!)^2}.
\]
\end{lemma}

\begin{proof}
Note that by \Cref{lem:est-l2-bounded-function} and the fact that $T_{\sqrt{\by_{|U|}}}g$ has a $B$-bounded randomized algorithm from \Cref{lem:alg-for-derivative}, we get that $0 \leq \bgamma_t \leq B^2$. 

For the claim on the expectation, fix any outcome  of $y_{|U|}$ and let $\calW$ denote the algorithm from \Cref{lem:est-l2-bounded-function} used on line 6(c) to compute $\bgamma_t$, the estimate of $\E_\bx \left[ \left( T_{\sqrt{y_{|U|}}} g(\mathbf{x}) \right)^2 \right]$. 
Recalling the failure probability from line~6(c) and the fact that the estimate, like the true value of $\E_\bx \left[ \left( T_{\sqrt{y_{|U|}}} g(\mathbf{x}) \right)^2 \right]$, lies in $[0,B^2],$ by \Cref{lem:est-l2-bounded-function} we get that
our estimate $\bgamma_t$ satisfies
	\[\Ex_{\calW} \left[ \left | \bgamma_t - \Ex_{\bx} \left[ \left(T_{\sqrt{y_{|U|}}} g(\bx) \right)^2 \right] \right| \right]\leq \frac{\eps}{4(|U|!)^2} + \left( \frac{\eps}{100 B^2 (|U|!)^2} \right) B^2 \leq \frac{\eps}{2(|U|!)^2}.
	 \]
We then note that using Jensen's inequality, we have
\[\left| \Ex_{\by_{U}, \calW} \left[ \bgamma_t \right] - \Ex_{\bx, \by_{|U|}} \left[ \left(T_{\sqrt{\by_{|U|}}} g(\bx) \right)^2 \right] \right| \leq \E_{\by_U} \left[ \Ex_{\calW} \left[ \left | \bgamma_t - \Ex_{\bx} \left[ \left(T_{\sqrt{\by_{|U|}}} g(\bx) \right)^2 \right] \right| 
\right] \right] \leq \frac{\eps}{2(|U|!)^2} \]
as desired.
\end{proof}

We can now prove \Cref{lem:estimate-ninf}:

\begin{proof}[Proof of \Cref{lem:estimate-ninf}]
 We begin by proving that the algorithm is correct. Indeed, note that by a Hoeffding bound, using the fact (\Cref{lem:ninf-approx-good}) that each $\bgamma_t$ lies in $[0,B^2]$, we have that
\[
\mathbf{Pr} \left[ \left| \frac{1}{M} \sum_{t=1}^M \bgamma_t - \mathbf{E}[\bgamma_t] \right| \ge \frac{\varepsilon}{2(|U|!)^2} \right]
\le 2 \exp\left( \frac{-(M\varepsilon)^2}{2MB^4(|U|!)^4} \right) \le \delta.
\]

On the other hand, note that when $\left| \frac{1}{M} \sum_{t=1}^M \bgamma_t - \mathbf{E}[\bgamma_t] \right| \le \frac{\varepsilon}{2(|U|!)^2}$, then
\[
\left| (|U|!)^2 \cdot \frac{1}{M} \sum_{t=1}^M \bgamma_t - (|U|!)^2 \cdot \mathbf{E}[\bgamma_t] \right| \le \frac{\varepsilon}{2}.
\]

Applying \Cref{lem:ninf-exp-good} and \Cref{lem:ninf-approx-good} then yields
\[
\left| (|U|!)^2 \cdot \frac{1}{M} \sum_{t=1}^M \bgamma_t - \NormInf_U[f] \right| \le \varepsilon,
\]
which completes the proof of correctness.

For the bound on the number of queries, note that in each iteration of the line~6 loop, the call to \Cref{lem:est-l2-bounded-function} on line 6(c) requires $\poly(1/\eps, |U|!, B)$ calls to the randomized algorithm, $\calA'$, for $T_{\sqrt{\by_{|U|}}} g$. By \Cref{lem:alg-for-derivative} each call to $\calA'$ makes $2^{|U|}$ queries to $\calA$. So all in all we make at most
	\[M \cdot 2^{|U|} \cdot \poly(1/\eps, |U|!, B) = \poly(1/\eps, |U|!, B, \log(1/\delta))\]
calls to $\calA$ throughout the algorithm.
\end{proof}

\subsection{Filtering Coordinates} \label{sec:filtering}
As mentioned above, the key tool for handling functions with mass on ``high'' levels is a filtering procedure which we call \refinecoords{}. Intuitively, this procedure takes a set of coordinates $C \subseteq [k']$  as input and outputs a collection of subsets $C' \subseteq C$. (Actually, for each subset $C'$ that it outputs, it also outputs an associated set $I$; we will explain the role of these $I$ sets later.) We will seek to show that for some $C'$ returned by the algorithm, we have that $|C' \cap \overline{R}| \leq  |C \cap \overline{R}|/2$, where $R \subseteq [k']$ is a set of size $k$. (It may be helpful to think of $R$ as corresponding to the set $U^\star$ of $k$ coordinates achieving the maximum junta correlation, but in our main lemma, \Cref{lem:good-pair}, $R$ can be any $k$-set.)

\bigskip

\begin{algorithm}[H]
\addtolength\linewidth{-2em}

\vspace{0.5em}

\textbf{Input:} A function $f: \bits^{k'} \rightarrow \mathbb{R}$, a set of coordinates $C \subseteq [k']$, and a parameter $\eps \in (0,1)$ \\[0.25em]
\textbf{Output:} A set $\calF$ of pairs $(\bC', \bI)$ with $\bC',\bI \subseteq C$

\vspace{0.5em}

\refinecoords:

\vspace{0.5em}

\begin{enumerate}
	\item Set $\calF \gets \emptyset$ and let 
	$$\beta = \max \big\{|C|, k\big\},\quad m= \frac{{10}\beta}{k^{2/3}} \log^{4}\left(\frac{k'}{\eps}\right)
	\quad\text{and}\quad \gamma = k^{1/3} \log^3\left(\frac{k'}{\eps}\right).$$
	\item Repeat $\exp\big(O(k^{1/3}\log^5(k'/\eps))\big)$ many 
	times:
	\begin{enumerate}
		\item Draw $\bi_1, \dots, \bi_m \sim C$ uniformly and independently, and set $\bI = \{\bi_1,\dots, \bi_m\}$.
		\item Set $\bC' \gets \emptyset$ and let $$\fave^{\bI}(x) := \Ex_{\by \sim \{\pm 1\}^{k'}} \Big[f(\by)\hspace{0.06cm} \big |\hspace{0.06cm} \by_{\overline{\bI}} = x_{\overline{\bI}} \Big] .$$
			\item While true do
			\begin{enumerate}
				\item Let $\ell,\kappa$ and $\Delta$ be the following parameters:
				$$
				\ell=k^{2/3}\log^3\left(\frac{k'}{\eps}\right),\quad\kappa=10\log\left(\frac{k'}{\eps}\right)\quad\text{and}\quad\Delta=10.
				$$
				\item Let $h$ be the following function: 
				  $$h:= \fave^{\bI} - \shnoise_{\ell,\kappa,\Delta}
				  ^{C \setminus \bC'} 
				  \hspace{0.05cm}\fave^\bI.$$
				\item Compute an estimate $\wh{\sigma}^2$ of $\E_{\bx}[h^2(\bx)]$ and 
				exit the loop if $\wh{\sigma}^2 \leq \eps^2/k^2$. 
				\item Sample a set ${\bT} \subseteq [k']$ among all sets $U \subseteq [k], |U| = \gamma$ with probability proportional to $\NormInf_U[h]$.
				\item Exit the loop if $\bT \cap C \subseteq \bC'$.
				\item Otherwise, 
				set $\bC' \gets \bC' \cup ({\bT} \cap C)$ and repeat.
			\end{enumerate}
		\item Add $(\bC',\bI)$ to $\calF$.
	\end{enumerate}
	\item Return $\calF$.
\end{enumerate}
\caption{An algorithm to refine coordinates.}
\label{alg:refine-coordinates}
\end{algorithm}

\bigskip

The main lemma that we will want to prove is the following:

\begin{lemma}
\label{lem:good-pair}
Let $f: \{\pm 1\}^{k'} \rightarrow [-1,1]$ be such that
$$
\bW^{\ge 100 k \log(k'/\eps)}[f]\le \left(\frac{\eps}{k'}\right)^5 
$$
and let $R\subset [k']$ be a set of size $k$.
Given access to $f$ via a 1-bounded randomized algorithm $\calA$ that computes it, and 
  a set $C \subseteq [k']$ of coordinates such that $C \cap \overline{R} \not = \emptyset$, 
  with probability at least $1 - 3/(k')^2$, \refinecoords{} outputs a collection $\calF$ of at most $2^{k^{1/3} \cdot \polylog(k'/\eps)}$ pairs $(\bC', \bI)$ with $\bC',\bI\subseteq C$ such that at least one pair satisfies 
\begin{equation}\label{eq:haha2}
\bI \subseteq \overline{R}, \quad 
|\bC' \cap \overline{R}| \leq  \frac{|C \cap \overline{R}|}{2},\quad\text{and}
\sum_{\substack{S\subseteq [k']:\\ |S \cap C \setminus \bC'| \geq 
\kappa\ell}} \wh{\fave^\bI}^2(S) \leq \frac{4\eps^2}{k^2}.
\end{equation}
\end{lemma}

The procedure \refinecoords{} is given in \Cref{alg:refine-coordinates}.
The reader may notice that details on how to compute an 
   estimate $\wh{\sigma}^2$ of $\E_{\bx}[h^2(\bx)]$ (line 2(c)(iii)) and 
   how to sample a set $\bT$ of size 
   $\gamma$ proportional to $\NormInf$ of $h$
   (line 2(c)(iv)) are missing in \Cref{alg:refine-coordinates}, 
   and these are actually the only steps in
   \refinecoords{} that are computation heavy.
This is because the proof of \Cref{lem:good-pair} is quite complex and 
  for clarity, we will start with a proof under an idealized setting in 
  \Cref{sec:idealized}.
Afterwards, we will explain details behind 2(c)(iii) and 2(c)(iv) and 
  lift the proof in the idealized setting to the real setting in \Cref{sec:full-approx-filter}.

\subsubsection{The Idealized Setting}\label{sec:idealized}

In the idealized setting, we prove \Cref{lem:good-pair} under the following assumption: 

\begin{assumption}
\label{ass:good-var}
We make the following two assumptions regarding 2(c)(iii) and 2(c)(iv):
\begin{flushleft}\begin{enumerate}
\item	Whenever line 2(c)(iii) is reached, if $\E_{\bx}[h^2(\bx)] \leq {\eps^2}/({2k^2})$ then the algorithm always exits the loop; if 
$\E_{\bx}[h^2(\bx)] \geq {2 \eps^2}/{k^2}$, then the algorithm always continues in the loop.
\item Whenever line 2(c)(iv) is reached, a set $\bT$ is drawn with probability proportional to 
  $\NormInf_U[h]$ among all sets $U\subseteq [k']$ of size $\gamma$.
An equivalent description of the distribution of $\bT$ is first drawing a set $\bS$ of size at least $\gamma$
   from the spectral sample of $h$ and then drawing a uniform subset $\bT$ of $\bS$ of size $\gamma$.
    (Of course this is only well defined if $h$ has mass at or above level $\gamma$, but
  as will become clear in the proof, this will always be the case given the first assumption.)
\end{enumerate}\end{flushleft}
\end{assumption}
(In fact, under \Cref{ass:good-var} we will prove \Cref{lem:good-pair} with a slightly stronger success probability of $1 - 1/(k')^2$; later we will use this stronger success probability to prove the actual \Cref{lem:good-pair} with its $1 - 3/(k')^2$ success probability.) 

We now gather several lemmas on the behavior of $\bT$ that will be useful for us. We begin by noting that the purpose of guessing $\bI$ is to get a ``random looking'' set of ``irrelevant'' coordinates in $C$, where we think of coordinates in $R$ as relevant and coordinates outside $R$ as irrelevant. Towards this, we show that we indeed expect to sample a set $\bI$ that is comprised of irrelevant variables with non-negligible probability.

\begin{lemma}
\label{lem:I-irr-vars}
	Let $R \subseteq [k']$ be a set of size at most $k$, and $C \subseteq [k']$ be such that $C \cap \overline{R} \not = \emptyset$. Then
		\[\Prx_{\bI} \big[\bI \cap R = \emptyset\big] \geq \exp\left(-O\left(k^{1/3} \log^{5}(k'/\eps)\right)\right). \]
\end{lemma}

\begin{proof}
Note that 
	\[
	\Prx_{\bI}[\bI \cap R = \emptyset] = \Prx_{\bI}\big[\bI \subseteq C \cap \overline{R}\big] = \left( 1 - \frac{|R \cap C|}{|C|} \right)^{m}. \]
If $|C| \geq 2k$, then we have $\beta=|C|$ and $m=(10|C|/k^{2/3})\log^4 (k'/\eps)$ and thus, 
\begin{align*}
	\Prx_{\bI}\big[\bI \subseteq C \cap \overline{R}\big] &\geq \left( 1 - \frac{k}{|C|} \right)^{m} \geq \exp \left( - \frac{2km}{|C|} \right) = \exp \left( -O\left(k^{1/3} \log^4(k'/\eps)\right) \right),
\end{align*}
where the inequality follows from $1 - x \geq e^{-2x}$ for $0 \leq x \leq 1/2$.

On the other hand, if $|C| \leq 2k$, then we have $\beta\le 2k$ and $m=O(k^{1/3}\log^4(k'/\eps))$ and thus, 
\begin{align*}
	\Prx_{\bI}\big[\bI \subseteq C \cap \overline{R}\big] &\geq \left( 1 - \frac{|C|-1}{|C|} \right)^{m} \geq \exp\left(-m \log\big(|C|\big)\right) \geq \exp \left(-O\left(k^{1/3} \log^{5}(k'/\eps)\right) \right),
\end{align*}
where the first inequality uses the fact that $C \cap \overline{R} \not = \emptyset$.
\end{proof}

Next we show that if $\bI$ were drawn from $C\cap \overline{R}$ rather than from $C$, then 
we would expect the Fourier mass of $f_{\ave}^{\bI}$ on terms with many irrelevant variables
(i.e.~variables in $C\cap \overline{R}$) to be heavily dampened. 
Let $$\delta= \frac{k^{2/3}}{\beta \log^3 (k')}$$ 
(this will be the last parameter needed in this 
  subsection).

\begin{lemma}
\label{lem:mostly-pure-terms}
	Let $f: \bits^{k'} \rightarrow [-1,1]$, $R \subseteq [k']$ be a subset of size $k$, and $C \subseteq [k']$ be such that $C \cap \overline{R} \not = \emptyset$.
	Let $\bI$ be a set of $m$ coordinates drawn independently from $C \cap \overline{R}$.
	Then with probability at least $0.8$, 
	we have
\begin{equation}\label{eq:haha1}
\sum_{\substack{S \subseteq [k']:\\ |S \cap C \cap \overline{R}|\geq \delta |C\cap \overline{R}|
	}} \wh{f_{\ave}^{\bI}}^2(S) \leq \left(\frac{\eps }{k'}\right)^5.
\end{equation}
\end{lemma}
\begin{proof}
Fix a set $S \subseteq [k']$ such that $|S \cap C \cap \overline{R}|\geq 
  \delta|C\cap \overline{R}|$. We have that
	\begin{align*}
		\Pr_{\bI} \left[ S \cap \bI = \emptyset 
		\right] &= \left( 1 - \frac{|S \cap C \cap \overline{R}|}{|C \cap \overline{R}|} \right)^m\le \frac{1}{10}\cdot \left(\frac{\eps}{k'}\right)^5,
	\end{align*}
by our choices of $\delta$ and $m$.
Recalling from \Cref{eq:corr} that 
 $f^{\bI}_{\mathrm{ave}}(x) = \sum_{S: S \subseteq \overline{\bI}} \widehat{f}(S) \chi_S(x)$,
we get that
	\begin{align*}
		\Ex_{\bI} \left[ \sum_{\substack{S \subseteq [k']: \\ |S \cap C \cap \overline{R}|\geq 
		\delta |C\cap \overline{R}|}}
		\wh{f_{\ave}^{\bI}}^2(S) \right] 
		  \leq \frac{1}{10}\cdot \left(\frac{\eps}{k'}\right)^5\sum_{\substack{S \subseteq [k']: \\ |S \cap C \cap \overline{R}|\geq
		   \delta |C\cap\overline{R}|}}
		   \wh{f}^2(S) 
		  \leq \frac{1}{10}\cdot \left(\frac{\eps}{k'}\right)^5,
	\end{align*}
	where the second inequality is by using $\E_{\bx}[f^2(\bx)]\le 1$.
By Markov's inequality, \Cref{eq:haha1} occurs with probability at least $0.9$.
\end{proof}

Given the number of times line 2 is repeated, we have from 
  \Cref{lem:I-irr-vars} and \Cref{lem:mostly-pure-terms}:

\begin{corollary} \label{cor:manyI}
With probability at least $1-1/(2(k')^2)$, at least $10\log(k')$ rounds of $\bI$ that are sampled in line~2(a)
  satisfy \Cref{eq:haha1}.
\end{corollary}

We fix such a set $I$ (satisfying \Cref{eq:haha1}) in the rest of the proof, and show 
  that for the iteration with this $I$, the pair $(\bC',I)$ that is returned satisfies
  the desired properties in \Cref{eq:haha2} with probability at least $1/2$. 
Let $\calD_h$ be the spectral sample of $h$ conditioning on sets being 
  of size at least $\gamma$ (see \Cref{lem:fourier-wt-lb} below which implies that $\calD_h$ is well defined).
Recall that an equivalent view of step~2(c)(iv) of the algorithm, for the purpose of analysis, is that a set $\bS$ is drawn from $\calD_h$ and then a uniform $\gamma$-subset $\bT$ is drawn from $\bS$. Our subsequent analysis will heavily use this view.

To understand $\bT$, 
we prove a few lemmas to show that most likely $\bS$ is \emph{good}, as defined below: 
\begin{definition}[Good Sets]
	We say a set $S \subseteq [k']$ is \emph{good} if it satisfies the following conditions:
\begin{enumerate}	
		\item[(i)]  $|S \cap R \cap C \setminus \bC'| \geq \ell/2$;
		\item[(ii)] $|S \cap C \cap \overline{R}| \leq \delta |C\cap \overline{R}|$; and
		\item[(iii)] $|S| \leq 100 k \log(k'/\eps)$.
	\end{enumerate}
In words, its size must be at most roughly $k$; its intersection with $C \setminus R$ must be a small fraction of $C \setminus R$; and the part of its intersection with $C \cap R$ that lies outside $C'$ must not be too small.
\end{definition}

\begin{lemma}
\label{lem:fourier-wt-lb}
If $\E_{\bx}[h^2(\bx)]\ge \eps^2/(2k^2)$ and $\E_{\bx}[f^2(\bx)] \leq 1$, then we have
	\[\bW^{\ge\gamma}[h] \geq \frac{\eps^2}{4 k^2}. \]	
\end{lemma}
\begin{proof}
By the properties of \shnoise{} (\Cref{lem:sharp-noise}), we have that
	\[\bW^{< \gamma} [h] \leq \bW^{\leq \ell} [h] \leq 
	(\Delta 2^{-\kappa})^2\cdot \E_{\bx}[f^2(\bx)]\le 
	\eps^2/(4k^2)\]
	by {our choices of $\Delta$ and $\kappa$}.
The lemma follows from
	${\eps^2}/({2k^2}) \leq \E_{\bx}[h^2(\bx)] = \bW^{< \gamma}[h] + \bW^{\geq \gamma}[h]$.
\end{proof}

We are now ready to show that most likely $\bS\sim \calD_h$ is good:

\begin{lemma} \label{lem:S-good}
Suppose that $f$ and $h$ satisfy 
\begin{align}
\E_{\bx}[f^2(\bx)] &\leq 1,\nonumber \\
\bW^{\ge 100k\log(k'/\eps)}[f] &\le  ({\eps}/{k'})^5 \text{~(as in the hypothesis of \Cref{lem:good-pair}),} \label{eq:snack1}\\
 \bW^{\geq \gamma}[h] &\ge \eps^2/(4k^2)
 \text{~(as given by \Cref{lem:fourier-wt-lb}),}\label{eq:snack2}
 \end{align}
 and \Cref{eq:haha1}. 
  Then we have
	\[\Prx_{\bS\sim \calD_h}\big[\bS \text{ good}\big] \ge 1 - \frac{1}{(k')^2}.\] 
\end{lemma}
\begin{proof}
	By \Cref{lem:sharp-noise}, we have that for any $S \subseteq [k']$
		\[\wh{h}^2(S) = \left( \wh{\fave^\bI}(S) - \reallywidehat{\shnoise^{C \setminus \bC'}_{\ell,\kappa,\Delta} 
		\fave^\bI}(S) \right)^2 \leq \wh{\fave^\bI}^2(S), \]
	so it follows from \Cref{eq:haha1}
	that
		\[\sum_{\substack{S \subseteq [k']:\\ |S \cap C \cap \overline{R}|\geq \delta |C \cap \overline{R}|}} 
		\wh{h}^2(S) \leq \left(\frac{\eps}{k'}\right)^5.\]	
	It then follows \Cref{eq:snack2} that
		\[\Prx_{\bS\sim \calD_h} \left [ |\bS \cap C \cap \overline{R}| \geq \delta|C\cap \overline{R}|\right]
		\leq \frac{(\eps/k')^{5}}{\bW^{\geq \gamma}[h]} \leq \frac{1}{3(k')^{2}},\]
so part (ii) of $\bS$ being good holds with high probability.
		We now consider part (i); towards this end,
	we will argue that it is unlikely that $|\bS\cap C\setminus \bC'|\le \ell$.
	Indeed, note that by \Cref{lem:sharp-noise}	
		\[\sum_{S: |S \cap C \setminus \bC'| \leq \ell} \wh{h}^2(S) \leq (\Delta 2^{-\kappa})^2
		\sum_{S: |S \cap C \setminus \bC'| \leq \ell} \wh{\fave^\bI}^2(S) 
		\le (\Delta 2^{-\kappa})^2\cdot \E_{\bx}[f^2(\bx)]
		\leq \left(\frac{\eps }{ k'  }\right)^5,\]
		{by our choices of $\Delta$ and $\kappa$.}
Similar to the earlier argument, this implies that
	$$\Pr_{\bS\sim \calD_h}\left[|\bS \cap C \setminus \bC'| \leq \ell\right] \leq 
	\frac{1}{3(k')^2},$$
	and when $|\bS \cap C \setminus \bC'| \ge \ell$ and $|\bS \cap C \cap \overline{R}| \le \delta|C\cap \overline{R}|$, we have
	$$
	|\bS\cap C \cap R \setminus \bC'|\ge \ell-\delta |C\cap \overline{R}|
	\ge \ell- \frac{k^{2/3}}{|C|\log^3(k')}\cdot |C\cap \overline{R}|\ge \frac{\ell}{2},
	$$ 
giving part (ii).
Finally, for part (iii), again by \Cref{lem:sharp-noise} we have 
\[
\sum_{\substack{S \subseteq [k']:\\ |S| \geq 100 k \log(k'/\eps)}} \wh{h}^2(S) 
\leq
\bW^{\geq 100 k \log(k'/\eps)}[f]
\leq \left({\frac {\eps}{k'}}\right)^5
\]
(by \Cref{eq:snack1}).  As in the two preceding arguments, this means that
	\[\Prx_{\bS\sim \calD_h}\left[ |\bS | \geq 100 k \log(k'/\eps)\right] \leq \frac{(\eps/k')^5}{\bW^{\ge \gamma}[h]} \le \frac{1}{3(k')^2} . \]
The lemma then follows by a union bound over (i), (ii) and (iii).
\end{proof}

The following lemma makes it easier to consider the random process of growing $\bC'$:

\begin{lemma}\label{lem:finalhaha}
Under the assumptions of \Cref{lem:S-good}, with probability at least $0.9$, every 
  $\bS\sim \calD_h$ is a good set and when the loop of step~2(c) exits, it does so in step (iii) rather than step (v).

\end{lemma}
\begin{proof}
There are at most $k'+1$ many rounds as line 2(c)(vi) always adds an element to $\bC'$ and there are at most $k'$ elements. So, by a union bound over \Cref{lem:S-good}, every $\bS$ sampled is good.
Assuming $\bS$ is good, $|\bS\cap C\setminus \bC'|\ge \ell/2$ and thus, 
  the probability of having $\bT\cap (C\setminus \bC')=\emptyset$ is at most
  $$
\left(1-\frac{\ell/2}{|S|}\right)^\gamma\le 
\left(1-\frac{\ell/2}{100k\log(k'/\eps)}\right)^\gamma\le \frac{1}{(k')^2}
  $$
and the lemma follows from a union bound.
\end{proof}

Assuming the lemma above holds,  the algorithm must exit the loop 
  because $\E_\bx[h^2(\bx)]\le 2\eps^2/k^2$ for the function $h$ during that round. 
With this and \Cref{lem:sharp-noise}, we have
\begin{equation}
\frac{2\eps^2}{k^2}\ge \E_\bx[h^2(\bx)] \ge (1-2^{-\Delta})^2\sum_{S: |S \cap C \setminus \bC'| \geq \kappa\ell} \wh{\fave^\bI}^2(S)\ge \frac{1}{2} \sum_{S: |S \cap C \setminus \bC'| \geq \kappa\ell} \wh{\fave^\bI}^2(S).  \label{eq:cheese}
\end{equation}

\def\balpha{\boldsymbol{\alpha}}
\def\bbeta{\boldsymbol{\beta}}

Assuming the condition of \Cref{lem:finalhaha},
  we finally work on how $\bC'$ evolves round by round and show that 
  with probability at least $1/2$, the final $\bC'$ satisfies 
$$
|\bC'\cap \overline{R}|\le |C\cap \overline{R}|\big/ 2.
$$
To this end we consider
  the following random process that captures how $\bC'$, starting with $\bC'=\emptyset$,
  evolves.
The process takes at least one round and 
  at most $k'$ rounds to end.
At the beginning of the $i$-th round, we receive three random numbers 
  (which may depend on what happens in previous rounds) 
  $\bt_i,\balpha_i$ and $\bbeta_i$ such that
$$
1\le \bt_i\le \balpha_i+\bbeta_i,\quad \balpha_i\ge \ell/2\quad \text{and}\quad
  \bbeta_i\le \delta|C\cap \overline{R}|\ll \ell.
$$
(They corresponds to $|\bT\cap C\setminus \bC'|$, $|\bS\cap R\cap C\setminus \bC'|$
  and $|\bS\cap \overline{R}\cap C\setminus \bC'|$, respectively.
The last two inequalities follow from the assumption that $\bS$ is good;
the first inequality follows from $\bT\subseteq \bS$ and the assumption that $\bT\cap C\setminus \bC'\ne\emptyset$.) 

Given $\bt_i,\balpha_i$ and $\bbeta_i$, we then draw $\bX_i$, which is a random variable that
  corresponds to the number of 
  blue balls we get if we draw $\bt_i$ balls from a pool of $\balpha_i$ red balls and 
  $\bbeta_i$ blue balls without replacement.
So $\bE[\bX_i]=\bt_i\cdot \bbeta_i/(\balpha_i+\bbeta_i)$.
(This corresponds to the fact that, conditioning on the value of $|\bT\cap C\setminus \bC'|$,
  $|\bT\cap \overline{R}\cap C\setminus \bC'|$ is distributed exactly the same as $\bX_i$.)
If the process ended before round $i$, then we just set $\bt_i=\bX_i=0$.

Our goal, then, is to show that with probability at least $1/2$, we have 
\begin{equation}\label{hehe100}
\sum_{i\in [k']}\bX_i\le 
\frac{|C\cap \overline{R}|}{2|C|} \cdot \sum_{i\in [k']}\bt_i.
\end{equation}
(Note that $\sum_i \bX_i$ corresponds to $|\bC'\cap \overline{R}|$ at the end,
  and $\sum_i \bt_i$ corresponds to $|\bC'|$ at the end.
Using $\bC'\subseteq C$ and $|\bC'|\le |C|$, it follows from \Cref{hehe100} that
$$
|\bC'\cap \overline{R}|\le |C\cap \overline{R}|\big/2.
$$
So it suffices to prove that \Cref{hehe100} occurs with probability at least $1/2$.)

\def\btau{\boldsymbol{\tau}}

To this end, we let $\btau_i=\E[\bX_i]$ for each $i$ so $\btau_i=0$ if the process ended before 
  the $i$-th round and $\btau_i=\bt_i\cdot \bbeta_i/(\balpha_i+\bbeta_i)$ otherwise.
The following claim is straightforward:

\begin{claim}
With probability at least $0.9$,
  we have $\bX_i\le O(\log(k'))\cdot \max(1,\btau_i)$ for every $i\in [k']$.
\end{claim}
\begin{proof}
Fix an $i\in [k']$.
If $\bt_i\ge (\balpha_i+\bbeta_i)/2$, then the bound is trivial because 
  $\btau_i\ge \bbeta_i/ 2$.
 
Otherwise, $\bX_i$ is dominated by the sum of $\bt_i$ indicator random variables,
  each taking $1$ with probability at most $2\bbeta_i/(\balpha_i+\bbeta_i)\ll 1$ given that $\bbeta_i\ll \balpha_i$.
If $\btau_i\ge 1/4$,
the claim follows from a Chernoff bound. 
Otherwise, the probability of $\bX_i$ taking some value $a\ge 10\log(k')$ is at most
$$
(\bt_i)^a\cdot \left(\frac{2\bbeta_i}{\balpha_i+\bbeta_i}\right)^a\le (1/2)^a 
$$ 
and summing over $a$ shows that the probability of $\bX_i\ge O(\log k')$ is at most $1/(k')^2$.
The claim then follows by a union bound over the $k'$ rounds. 
  \end{proof}

Assume the condition of the above claim holds.
Let $\bZ_i=1$ if $\btau_i\le 1/2$ and $\bX_i\ge 1$, and $\bZ_i=0$ otherwise.
Then we have 
$$\bX_i\le O(\log(k'))\cdot \left(\bZ_i+\btau_i\right)$$
On the other hand, $\bZ_i$ is a sequence of indicator random variables, each
  taking $1$ with probability at most $\btau_i\le 1/2$ by Markov (or $0$ when $\btau_i> 1/2$).
(Of course it is still the case that $\btau_i$ is a random number that depends on previous rounds.)
We prove the following claim:

\def\bi{\mathbf{i}}

\begin{claim} 
With probability at least $0.9$, we have 
\begin{equation}\label{eq:lastone}
\sum_{i\in [k']} \bZ_i\le O\big(\log^2 k'\big)\cdot \sum_{i\in [k']} {\btau_i}.
\end{equation}
\end{claim}
\begin{proof}
Consider the random process where in each of the $k'$ rounds, an adversary 
  picks a number $\btau_i\le 1/2$ and then $\bZ_i$ is set to be $1$ with probability
  $\btau_i$ and $0$ with probability $1-\btau_i$, and we would like to prove that
  \Cref{eq:lastone} holds with probability at least $0.9$.
We can further assume without loss of generality that every nonzero $\btau_i$ is $2^{-j}$ for 
  some $j\ge 1$; to see this we just round $\btau_i$ up to the closest $2^{-j}$ but only
  charge $2^{-j-1}$ on the RHS of \Cref{eq:lastone}.
  Also note from the definition of $\btau_i=\bt_i\bbeta_i/(\balpha_i+\bbeta_i)$ that every nonzero $\btau_i$ is at least $1/k'$.
  
Let $\bi^*$ be the random variable that denotes the smallest $\bi^*$ such that 
  $\sum_{i\le \bi^*} \btau_i\ge 0.01$,  and set $\bi^*=k'+1$ if no such $\bi^*$ exists.
Let $E$ be the event such that $\sum_{i< \bi^*} \bZ_i=0$. Then we have
$$
\Pr[E]=(1-\btau_1)\cdots (1-\btau_{\bi^*-1})\ge \exp\left(-2\sum_{i<\bi^*} \btau_i\right)
\ge e^{-0.02}\ge 0.98,
$$
where we used $1 - x \geq e^{-2x}$ for $0 \leq x \leq 1/2$. 
It then suffices to show that the probability of 
$$
\sum_{i\ge \bi^*} \bZ_i\le O\big(\log^2 k'\big)\cdot \max\left(1,\sum_{i\ge \bi^*} \btau_i\right)
$$
is at least $0.98$. Note this is the same claim we would like to prove, 
  except that we now get to add a $\max$ with $1$ on the RHS (and raising the probability from $0.9$ to $0.98$).
For convenience, we will still use the same notation and show that
\begin{equation}\label{eq:lastlastone}
\sum_{i\in [k']} \bZ_i\le O\big(\log^2 k'\big)\cdot \max\left(1,\sum_{i\in [k']} {\btau_i}\right).
\end{equation}
with probability at least $0.98$ in the rest of the proof.
  
To this end, we consider the following equivalent setting:
\begin{flushleft}\begin{enumerate}
\item For each $j\le \log(k')$, start 
  by drawing a sequence of $k'$ random bits $\bZ_{j,1},\ldots,\bZ_{j,k'}\in \{0,1\}$,
  each set to $1$ with probability $2^{-j}$ independently.
  The adversary does not get to see them.
\item In each of the $k'$ rounds, the adversary gets to pick a $j$, set $\btau$ to $2^{-j}$
  and $\bZ$ to the next unused bit in the sequence $\bZ_{j,1},\ldots,\bZ_{j,k'}.$
\end{enumerate}\end{flushleft}
The following event would imply \Cref{eq:lastlastone}: For every $j$ and every $\ell\le k'$, we have 
$$
\sum_{i'\in [\ell]} \bZ_{j,i'}\le O(\log k')\cdot \max\left(1, \frac{\ell}{2^j}\right).
$$
It now follows from a standard Chernoff bound and union bound that the above event
  occurs with probability at least $0.98$.
\end{proof}

As a result of \Cref{eq:lastone}, we have that
$$
\sum_{i\in [k']}\bX_i\le O\left(\log^3 k'\right)\cdot \sum_{i\in [k']} {\btau_i}=
O\left(\log^{{3}} k'\right)\sum_{i\in [k']} \frac{\bbeta_i}{\balpha_i+\bbeta_i}\cdot \bt_i.
$$
But for every $i$, we have by the choice of $\delta$ that
$$
\frac{\bbeta_i}{\balpha_i+\bbeta_i}\le \frac{\delta|C\cap \overline{R}|}{\ell/2}.
$$
Using $(2\delta/\ell)\cdot O(\log^3 (k'))\le 1/(2|C|)$ by our choices of $\delta$ and $\ell$,
  it follows that \Cref{hehe100} holds with probability at least $0.8$.

We briefly summarize how the results established above in this subsubsection give a proof of 
 \Cref{lem:good-pair} in the idealized setting.
 
 \begin{proofof}{\Cref{lem:good-pair} under \Cref{ass:good-var}}
 It is clear that the collection ${\cal F}$ of $(\bC',\bI)$ pairs has size at most $2^{\tilde{O}(k^{1/3})}$ and that each $\bC',\bI \subseteq C.$
 By \Cref{cor:manyI}, except with failure probability at most $1/(2(k')^2)$ there are at least $10 \log k'$ rounds in which the $\bI$ sampled in line~2(a) lies entirely in $\overline{R}$.
For each such $\bI$ we can apply \Cref{lem:finalhaha} and \Cref{eq:cheese} to get that with probability at least 0.9, both
\begin{itemize}
\item the second inequality of \Cref{eq:haha2} holds, and
\item \Cref{hehe100} holds with probability at least 0.8 (in which case we get $|\bC'\cap \overline{R}|\le |C\cap \overline{R}|\big/ 2$, the first inequality of \Cref{eq:haha2}).
\end{itemize}
So assuming that there were at least $10\log k'$ rounds in which $\bI \subseteq \overline{R}$, the probability that none of those rounds yields a $\bC'$ satisfying \Cref{eq:haha2} is at most $(1-0.9 \cdot 0.8)^{10 \log k'} < 1/(2(k')^2).$ The lemma follows by a union bound. 
 \end{proofof}


\subsubsection{Handling the Approximations}
\label{sec:full-approx-filter}

\newcommand{\iters}{2^{k^{1/3} \polylog(k'/\eps)}}

We now turn to address the fact that we only have 
access to $f$ via a randomized algorithm $\calA$. 
To start, we begin by observing that we have a randomized algorithm that computes $h$ in Line 2(c)(ii).

\begin{lemma}
\label{lem:randomized-alg-h}
In each iteration of the loop on line~2(c) in \refinecoords{}, there is a $(2^{2 \kappa \Delta}+1)$-bounded randomized algorithm $\calA''$ for computing $h$
that makes $O(\kappa \Delta)$ queries to the randomized algorithm $\calA$ for $f$.
\end{lemma}

\begin{proof}

Recall that we have that by \Cref{lem:sharp-noise}
	\[\shnoise_{\ell, \kappa, \Delta}^{C \setminus \bC'} \fave^{\bI} = \sum_{i=0}^{\kappa \Delta} \alpha_i \T_{\rho^i}^{C \setminus \bC} \fave^{\bI} \]
for $\rho := 1 - \frac{1}{2 \ell}$ and $\alpha_i$ satisfying
	\[\sum_{i=0}^{\kappa \Delta} |\alpha_i| \leq 2^{2 \kappa \Delta}. \]

We now design our $(2^{2 \kappa \Delta}+1)$-bounded randomized algorithm $\calA''$ for $h$ as follows: Given $x$, output 
	\[\calA'(x) - \sum_{i=0}^{\kappa \Delta} \alpha_i \calN_{\rho^i}^{C \setminus \bC'}(x)\]
where $\calA'$ is the $1$-bounded randomized algorithm for $\fave^{\bI}$ from \Cref{lem:alg-for-average} and $\calN_{\rho^i}^{C \setminus \bC'}$ is the $1$-bounded randomized algorithm for computing $\T_{\rho^i}^{C \setminus \bC'} \fave^{\bI}$ given by \Cref{lem:alg-for-noise}. To see that the mean is correct, note that for any $x \in \bits^{k'}$ we have
	\[\E[\calA''(x)] = \E \left[ \calA'(x) - \sum_{i=1}^{\kappa \Delta} \alpha_i \calN_{\rho^i}^{C \setminus \bC'}(x) \right] = \fave^{\bI}(x) - \sum_{i=1}^{\kappa \Delta} \alpha_i \T_{\rho^i}^{C \setminus \bC'} \fave^{\bI} (x) = h(x) \]
as desired. We further note that boundedness follows from our bound on the sum of the absolute values of $\alpha_i$. For the query bound, note that each query to $\calA'$ and $\calN_{\rho^i}^{C \setminus \bC'}$ requires one query to $\calA$, so we make at most $1 + \kappa \Delta$ queries.
\end{proof}

With \Cref{lem:randomized-alg-h} in hand, we can now describe our full implementation of \refinecoords{}. In particular, to implement 2(c)(iii) we use the algorithm from \Cref{lem:est-l2-bounded-function} with the bounded randomized algorithm $\calA''$ given by \Cref{lem:randomized-alg-h} to compute an estimate that has accuracy $\pm \frac{\eps^2}{10k^2}$ with failure probability at most $2^{-(k')^2}$. In line 2(c)(iv), for every set $U \subseteq [k']$ of size $\gamma$, we compute an estimate of $\NormInf_U[h]$, which we denote $\lambda_U$, by calling 
\begin{equation}
\label{eq:estninfcall}
\estninf \left(h, 2^{2 \kappa \Delta}+1, U, \left( \frac{\eps}{k'} \right)^{10} \left( \iters \right)^{-1} \cdot \binom{k'}{\gamma}^{-2}, 2^{-(k')^2} \cdot {k' \choose \gamma}^{-1} \right)
\end{equation}
	where the reader should think of the $\iters$ term above as corresponding to the number of times we repeat the loop on line $2$ in \refinecoords{}. We then sample $\wt{\bT} = U$ with probability proportional to $\lambda_U$. (Throughout this section, we will write $\wt{\bT}$ to denote sets drawn from the $\lambda_U$'s and $\bT$ denote the original sets drawn according to the true normalized influences.)

We begin our analysis with a very useful lemma which shows that in any iteration of the loop on Line 2(c) of \refinecoords{}, the distributions of $\wt{\bT}$ and $\bT$ are extremely close to each other: 

\begin{lemma}
\label{lem:tv-dist-small}
In any iteration of the loop on Line 2(d) of \refinecoords{}, the total variation distance between $\wt{\bT}$ and $\bT$ 
is at most $\left(\frac{\eps}{k'} \right)^5 \left( \iters \right)^{-1}$.
\end{lemma}

\begin{proof}
	Note that it suffices to prove that for any $U \in \binom{[k']}{\gamma}$ we have that
	\[\Pr[\wt{\bT} = U] \geq \Pr[\bT = U] - \left(\frac{\eps}{k'} \right)^5 \cdot \binom{k'}{\gamma}^{-1} \left( \iters \right)^{-1}\]
	and 
	\[\Pr[\wt{\bT} = U] \leq \Pr[\bT = U] + \left(\frac{\eps}{k'} \right)^5 \cdot \binom{k'}{\gamma}^{-1} \left( \iters \right)^{-1}.\]
	To show this, we condition on the compound event that (i) for all $U \in \binom{[k']}{\gamma}$ we have 
	\[ |\lambda_U - \NormInf_U[h]| \leq \left( \frac{\eps}{k'} \right)^{10} \binom{k'}{\gamma}^{-2} \left( \iters \right)^{-1}, \] 
	and (ii) $\E_{\bx}[h^2(\bx)] \geq \frac{\eps^2}{2(k')^2}$. Note that these both happen with probability at least $1 - 2 \cdot 2^{-(k')^2}$ by \Cref{lem:estimate-ninf} and \Cref{lem:est-l2-bounded-function} by our choice of parameters in the invocation of each call to $\estninf$.
	
By (i), recalling that
		\[\sum_{U \in \binom{[k']}{\gamma}} \NormInf_U[h] = W^{\geq \gamma}[h] \]
(see the comment immediately after \Cref{definition:norm-influences-general}),
	we get that
		\[\left |\sum_U \lambda_U - \bW^{\geq \gamma}[h] \right| \leq \left( \frac{\eps}{k'} \right)^{10} \binom{k'}{\gamma}^{-1} \left( \iters \right)^{-1}\]
	Since we are conditioning on (ii), we can apply \Cref{lem:fourier-wt-lb} to get that
		\[\bW^{\geq \gamma}[h] \geq \frac{\eps^2}{4(k')^2}.\]
We thus have
		\[\left |\sum_U \lambda_U - \bW^{\geq \gamma}[h]\right| \leq \left( \frac{\eps}{k'} \right)^{10} \binom{k'}{\gamma}^{-1} \left( \iters \right)^{-1} \leq \left( \frac{\eps}{k'} \right)^{7} \binom{k'}{\gamma}^{-1} \left( \iters \right)^{-1} \cdot \bW^{\geq \gamma}[h].\] 
	It then follows that conditioned on (i) and (ii), we have
	\begin{align*}
		\Pr[\wt{\bT} = U] &\leq \frac{\NormInf_U[{h}] + (\eps/k')^{10}\binom{k'}{\gamma}^{-2}\left( \iters \right)^{-1}}{\bW^{\geq \gamma}{[h]} \left(1 - \left( \frac{\eps}{k'} \right)^{7} \binom{k'}{\gamma}^{-1} \left( \iters \right)^{-1}\right)} \\
		&\leq \Pr[\bT = U] + 9 \left( \frac{\eps}{k'} \right)^{7} \binom{k'}{\gamma}^{-1} \left( \iters \right)^{-1},
	\end{align*}
where the second inequality uses ${\frac {\NormInf_U[h]}{\bW^{\geq \gamma}[h]}} = \Pr[\bT=U]$ and the lower bound on $\bW^{\geq \gamma}[h]$ from above.

	For the lower bound,
	reasoning as in the upper bound using (i) and (ii) we get that
	\begin{align*}
		\Pr[\wt{\bT} = U] &\geq \frac{\NormInf_U[f] - (\eps/k')^{10}\binom{k'}{\gamma}^{-2}\left( \iters \right)^{-1}}{\bW^{\geq \gamma} {[h]} \left(1 + \left( \frac{\eps}{k'} \right)^{7} \binom{k'}{\gamma}^{-1} \left( \iters \right)^{-1}\right)} \\
		&\geq \Pr[\bT = U] - 9 \left( \frac{\eps}{k'} \right)^{7} \binom{k'}{\gamma}^{-1} \left( \iters \right)^{-1}.
	\end{align*}
	Since our assumptions fail with probability at most $2^{-(k')^2}$ we then get that 
		\[\Pr[\wt{\bT} = U] \geq \Pr[\bT = U] - 9 \left( \frac{\eps}{k'} \right)^{7} \binom{k'}{\gamma}^{-1} \left( \iters \right)^{-1} - 2 \cdot 2^{-(k')^2} \]
	and
		\[\Pr[\wt{\bT} = U] \leq \Pr[\bT = U] + 9 \left( \frac{\eps}{k'} \right)^{7} \binom{k'}{\gamma}^{-1} \left( \iters \right)^{-1} + 2 \cdot 2^{-(k')^2}\]
	which completes the proof.
\end{proof}

With this strong bound on variation distance in hand, we can now prove \Cref{lem:good-pair} for our non-idealized \refinecoords{} procedure.

\begin{proof}[Proof of \Cref{lem:good-pair}.] 
In each fixed iteration of the loop on Line 2(d) of \refinecoords{}, we have that $d_{TV}(\wt{\bT}, \bT) \leq \left( \frac{\eps}{k'} \right)^5 \cdot \left( \iters \right)^{-1}$ by \Cref{lem:tv-dist-small}. By the coupling interpretation of total variation distance, this means that there exists a coupling between $\bT$ and $\wt{\bT}$ such that $\Pr[\bT \not = \wt{\bT}] = d_{TV}(\bT, \wt{\bT})$. 
	
	We now imagine drawing sets from $\wt{\bT}$ according to this coupling in line 2(c)(iv). Correspondingly, let $\wt{\calF}$ and $\calF$ be the sets built using draws of $\wt{\bT}$ and $\bT$, respectively. 
	Note that the estimate of $\E_{\bx}[h^2(\bx)]$ in line 2(c)(iii) has accuracy $\pm \frac{\eps^2}{10k^2}$ at every iteration 
with probability at least
$1 - 2^{-(k')^2}.$
As the loop on line 2(c) can never repeat more than $k'+1$ times (as we must add a new element to $\bC'$ each time we reach step 2(c)(vi)), by  
a union bound with high probability
	 the algorithm always exits the loop in line 2(c)(iii) when $\E_{\bx}[h^2(\bx)] \leq \frac{\eps^2}{2(k')^2}$ and never exits when $\E_{\bx}[h^2(\bx)] \geq \frac{2\eps^2}{(k')^2}$. This means that item~(1.) of \Cref{ass:good-var} is satisfied with probability at least $1 - 2^{-k'}$.
On the other hand, by \Cref{lem:tv-dist-small}, the coupling, and a union bound over all iterations of line 2(c)(iv), we have that $\wt{\bT} = \bT$ in every iteration of line 2(c)(iv) with probability at least
			\[1 - (k') \cdot \left( \iters \right) \cdot \left( \frac{\eps}{k'} \right)^5 \cdot \left( \iters \right)^{-1} \geq 1 - \left( \frac{\eps}{k'} \right)^4; \]
when this happens we have that that item~(2.) of \Cref{ass:good-var} is satisfied as well.
Our earlier analysis proved \Cref{lem:good-pair} under \Cref{ass:good-var}, so with overall probability at least 
\[1 - 2^{-k'} - \left( \frac{\eps}{k'} \right)^4 \geq 1 - \frac{2}{(k')^2}, \]
we have that the output $\tilde{\calF}$ of the non-idealized procedure is identical to the output $\calF$ of the idealized procedure.
By the idealized version of \Cref{lem:good-pair} we get the non-idealized version of \Cref{lem:good-pair} where now the success probability is at least $1 - {\frac 3 {(k')^2}}$.
\end{proof}

Finally, we turn to proving that \refinecoords{} does not make too many calls to the randomized algorithm $\calA$. 

\begin{lemma}
\label{lem:refine-coords-query-bound}
	Let $f: \{\pm 1\}^{k'} \rightarrow [-1,1]$ be a function that we have access to via a $1$-bounded randomized algorithm $\calA$, $\eps \in [0,1]$, and $C \subseteq [k']$ be a set of coordinates. Then \refinecoords($f,C$) makes at most $\poly(\iters)$ calls to $\calA$.
\end{lemma}

\begin{proof}
Calls to $\calA$ are made in lines 2(c)(iii) and 2(c)(iv). In particular, since as noted earlier we can only repeat the loop of line~2(c) at most $k'+1$ times, using the fact that we only run the loop on line $2$ of \refinecoords{} $\iters$ times, it follows that we make at most $(k'+1) \cdot \iters$ many calls to each of lines 2(c)(iii) and 2(c)(iv).

By \Cref{lem:est-l2-bounded-function}, we have that each call of line 2(c)(iii) makes $\poly(k', 2^{\kappa \Delta}, \eps)$ queries to the randomized algorithm $\calA''$ for $h$ (recall that as discussed earlier we invoke \Cref{lem:est-l2-bounded-function} in this context with accuracy parameter $\frac{\eps^2}{10k^2}$ and failure probability $2^{-(k')^2}$). In turn, each call to ${\cal A}''$ makes $O(\kappa \Delta)$ queries to $\calA$ by \Cref{lem:randomized-alg-h}. Thus, in total throughout all executions of line 2(c)(iii) of \refinecoords{} we make at most 
	\[k' \cdot \iters \cdot \poly(k', 2^{\kappa \Delta}, \eps) \cdot O(\kappa \Delta) \leq \poly( \iters) \]
queries to $\calA$.

On the other hand, each execution of line 2(c)(iv) makes ${k' \choose \gamma}$ calls to $\estninf$ as in \Cref{eq:estninfcall}. By \Cref{lem:estimate-ninf} each of those calls makes
	\[\poly(\gamma!, 2^{\kappa \Delta}, \iters) \]
calls to $\calA''$, and hence $O(\kappa \Delta)$ times as many calls to $\calA$ by \Cref{lem:randomized-alg-h}. Thus, over the entire execution of \refinecoords{}, we make at most 
	\[k' \cdot \iters \cdot \poly(\gamma!, 2^{\kappa \Delta}, \iters) \cdot O(\kappa \Delta) \leq \iters \]
many calls to $\calA$ via 2(c)(iv).
\end{proof}

\subsection{Finding a Pure Set of Coordinates}

In this section, we explain how we can iteratively apply \refinecoords{} so as to output a set $C' \subseteq C,$ where $C'$ contains no irrelevant variables but still captures almost all of the Fourier mass in $R$ above degree $\wt{\Omega}(k^{2/3})$.

\bigskip

\begin{algorithm}[H]
\addtolength\linewidth{-2em}

\vspace{0.5em}

\textbf{Input:} A function $g: \bits^{k'} \rightarrow [-1,1]$ along with a $1$-bounded randomized algorithm for $g$, a set of pairs $\calF_{{i-1}}$, an integer $i$, and a parameter $\eps \in [0,1]$ that will correspond to the desired junta correlation accuracy \\[0.25em]
\textbf{Output:} A set $\calF''$ of pairs $(C', I')$ with $\calF_{{i-1}} \subseteq \calF''$

\vspace{0.5em}

\findhighcoordsrec:

\vspace{0.5em}

\begin{enumerate}	
	\item Set ${\calF_i} \gets \emptyset$
	\item For each $(C,I)$ in $\calF_{{i-1}}$:
	\begin{enumerate}
		\item For each $(\bC', \bI')$ in $\refinecoords(g_{\ave}^I,C, \eps)$:
			\begin{enumerate}
				\item Add $(\bC', I \cup \bI')$ to ${\calF_i}$
			\end{enumerate}
	\end{enumerate} 	
	\item If $i < \log(k')$, return $\calF_{{i-1}} \cup \findhighcoordsrec(g, {\calF_i}, i+1, \eps)$
	\item Return ${\calF_{i-1} \cup \calF_i}$
\end{enumerate}
\caption{A recursive algorithm to find variables in high level Fourier coefficients.}
\label{alg:find-high-level-coordinates}
\end{algorithm}

\bigskip

\begin{algorithm}[H]
\addtolength\linewidth{-2em}

\vspace{0.5em}

\textbf{Input:} A function $f: \bits^{k'} \rightarrow [-1,1]$ along with a $1$-bounded randomized algorithm for $f$, and a parameter $\eps \in [0,1]$ that will correspond to the desired junta correlation accuracy \\[0.25em]
\textbf{Output:} A set $\calF$ of pairs $(C, I)$ with $C,I \subseteq [k']$

\vspace{0.5em}

\findhighcoords:

\vspace{0.5em}

\begin{enumerate}	
	\item $g \gets \T_{1 - \frac{1}{2k}} f$
	\item Return $\findhighcoordsrec(g, {\calF_{-1}=}\{([k'], \emptyset)\}, 0, \eps)$
\end{enumerate}
\caption{A wrapper method around $\findhighcoordsrec$.}
\label{alg:find-high-level-coordinates}
\end{algorithm}

\bigskip

The main guarantee we will need about the \findhighcoords~procedure given in \Cref{alg:find-high-level-coordinates}, is the following:

\begin{lemma}
\label{lem:found-high-coords}
Let $f: \bits^{k'} \rightarrow [-1,1]$ be a function computed by a $1$-bounded randomized algorithm $\calA$ and let $R \subseteq [k']$ be a set of size $k$. With probability at least $1 - \frac{1}{k'}$, \\
$\findhighcoords(f, \eps)$ outputs a set $\calF$ such that for some $(\bC,\bI) \in \calF$ we have

\begin{itemize}
\item	[$(i)$] $\bI \subseteq \overline{R}$;
	
\item	[$(ii)$] $\bC \subseteq R$;
	
\item	[$(iii)$] $\sum\limits_{\substack{ S \subseteq [k']: \\ |S \setminus \bC| \geq k^{2/3} \polylog(k'/\eps), \\ |S| \leq k}} \wh{\fave^{\bI}}^2(S) \leq \frac{\eps^2}{100}$.

\end{itemize}
\end{lemma}

\begin{proof}
We start with the simple observation that since $f$ is a $[-1,1]$ valued function, the attenutation of its high-degree weight by the noise operator yields
	\[\bW^{\geq 100 k \log(k'/\eps)} \left[ \T_{1 - \frac{1}{2k}} f \right] \leq \left(1 - \frac{1}{2k} \right)^{100 k \log(k'/\eps)} \E[f^2] \leq \left( \frac{\eps}{k'} \right)^{20}. \]
Hence by \Cref{eq:corr}, for any set $I \subseteq [k']$ we have
	\[\bW^{\geq 100 k \log(k'/\eps)} [g_{\text{ave}}^I] \leq \left( \frac{\eps}{k'} \right)^{20}, \]
which means that we can apply \Cref{lem:good-pair} to any function $g^I_{\ave}$. Indeed, we will repeatedly apply \Cref{lem:good-pair} to establish the following claim:

\begin{claim}
\label{claim:inductive-found-coords}
For $i \leq \log(k')$, let $\calF'_i$ denote the set $\calF_i$ that has been produced by the procedure $\findhighcoordsrec(g,C,i,\eps)$ after line $(2)$ has been completed, and let $\calF_{-1} = \{([k'], \emptyset)\}$. Then with probability at least $1 - (i + 1) \cdot \frac{10}{(k')^2}$, there exists an $\ell \leq \log(k')$ and pairs $(\bC_j, \bI_j) \in \calF'_j$, for all $-1 \leq j \leq {\min\{i,\ell\}}$,
such that $[k'] := \bC_{-2} =: \bC_{-1} \supseteq \bC_0 \supseteq \bC_1 \supseteq \bC_2 \supseteq \dots \supseteq \bC_{\min\{i,\ell\}}$ and  $\bI_{\min\{i,\ell\}} \supseteq \bI_{i-1} \supseteq \dots \supseteq \bI_{-1} := \emptyset$, which satisfy the following for all $-1 \leq j \leq \ell$:

\begin{itemize}
 \item [$(a)$] $\bI_j \subseteq \overline{R}$

\item [$(b)$] $|\bC_j \cap \overline{R}| \leq \begin{cases} \frac{|\overline{R}|}{2^{j+1}} & j < \ell 
\\ 0 & j =\ell 
\end{cases}$

\item [$(c)$] $\sum\limits_{\substack{S: |S \cap \bC_{j-1} \setminus \bC_j| \geq k^{2/3} \polylog(k'/\eps)}} \wh{g_{\ave}^{\bI_j}}^2(S) \leq \frac{{4}\eps^2}{k^2}.$

\end{itemize}

\end{claim}

Before proving the claim, we first show how it implies \Cref{lem:found-high-coords}. Indeed, we'll apply the claim with $i = \log(k')$ and show that $(\bC_\ell, \bI_\ell)$ will satisfy conditions $(i)$, $(ii)$, and $(iii)$ of the lemma. Let $[k'] := \bC_{-2} =: \bC_{-1} \supseteq \bC_0 \supseteq \bC_1 \supseteq \bC_2 \supseteq \dots \supseteq \bC_{\min\{i,\ell\}}$ and $\bI_{\min\{i,\ell\}} \supseteq \bI_{i-1} \supseteq \dots \supseteq \bI_{-1} := \emptyset$ be the sequence of sets produced by \Cref{claim:inductive-found-coords} for $i = \log(k')$. To start, note that properties $(a)$ and $(b)$ immediately gives us that $(\bC_\ell, \bI_{\ell})$ satisfies $(i)$ and $(ii)$.

It remains to show that property $(iii)$ holds. To start, note that for every set $S \subseteq [k']$ of size at most $k$, we have that for any set $I \subseteq [k']$
	\[\wh{g_{\ave}^I}^2(S) \geq \left(1 - \frac{1}{2k} \right)^{2|S|} \wh{f_{\ave}^I}^2(S) \geq \frac{1}{4} \wh{f_{\ave}^I}^2(S). \]
Thus, we can conclude that that property $(c)$ implies that for all $-1 \leq j \leq \ell$ we have that

\begin{equation} \label{eq:poke}
\sum\limits_{\substack{S: |S \cap \bC_{j-1} \setminus \bC_j| \geq k^{2/3} \polylog(k'/\eps) \\ |S| \leq k}} \wh{f_{\ave}^{\bI_j}}^2(S) \leq 4 \sum\limits_{\substack{S: |S \cap \bC_{j-1} \setminus \bC_j| \geq k^{2/3} \polylog(k'/\eps) \\ |S| \leq k}} \wh{g_{\ave}^{\bI_j}}^2(S) \leq \frac{{16}\eps^2}{k^2},
\end{equation}

We now note that if $|S \setminus \bC'_{\ell}| \geq k^{2/3} {(\ell+2)}\polylog(k'/\eps)$, then there must  exists a $j \leq \ell$ such that $|S \cap \bC_{j-1} \setminus \bC_j| \geq 
k^{2/3} \polylog(k'/\eps)$, since 
	\[\bigcup_{j=0}^{\ell} \bS \cap \bC'_{j-1} \setminus \bC'_j = \bS \setminus \bC'_{\ell}.\]
Using this and $\bI_\ell \supseteq \bI_{i-1} \supseteq \cdots$ for the first inequality below, we get that
\begin{align*}
\sum\limits_{\substack{ S \subseteq [k']: \\ |S \setminus \bC| \geq k^{2/3} (\ell+2) \polylog(k'/\eps), \\ |S| \leq k}} \wh{\fave^{\bI_{\ell}}}^2(S) &\leq \sum_{j=0}^{\ell} \sum_{\substack{S \subseteq [k']: \\ |S \cap \bC_{j-1} \setminus \bC_j| \geq k^{2/3} \polylog(k/\eps) \\ |S| \leq k}} \wh{\fave^{\bI_{j}}}^2{(S)}\\
	&\leq \sum_{j=0}^{\log(k')} \sum_{\substack{S \subseteq [k']: \\ |S \cap \bC_{j-1} \setminus \bC_j| \geq k^{2/3} \polylog(k'/\eps) \\ |S| \leq k}} \wh{\fave^{\bI_{j}}}^2{(S)} \\
	&\leq \log(k') \cdot \frac{{16}\eps^2}{k^2} \tag{by \Cref{eq:poke}}\\
	&\leq \frac{\eps^2}{100}.
\end{align*}

Since $\ell \leq \log(k')$ this gives us property $(iii)$ as desired.

We now turn to prove \Cref{claim:inductive-found-coords}.

\begin{proof}[Proof of \Cref{claim:inductive-found-coords}]
We prove the statement by induction on $i$. For the base case, note that when $i = -1$, taking $\ell=0$ and $(\bC_{-1},\bI_{-1})=([k'], \emptyset)$ satisfies the required conditions. We now assume the statement holds for $i$ and wish to show that it is true for $\calF_{i+1}$. By the inductive hypothesis, let $\ell \leq \log(k')$ and $\bC_{-1} \supseteq \bC_0 \supseteq \bC_1 \supseteq \bC_2 \supseteq \dots \supseteq \bC_{\min\{i,\ell\}}$ be as promised by the lemma with corresponding sets $\bI_{\min\{i,\ell\}} \supseteq \cdots  \supseteq \bI_{-1}$. 

We claim that if, in the inductive hypothesis, we can take $\ell \leq i$, then we are done. To see this, observe that in the execution of the loop on line $(2)$ in the $(i+1)$st recursive call corresponding to $(\bC_i, \bI_i)$, we will add a pair $(\bC, \bI)$ to $\calF_{i+1}$ which by the guarantee of \refinecoords{} (\Cref{lem:good-pair}) satisfies $\bC \subseteq \bC_i$ and $\bI \subseteq \bC_i$. Indeed, in this iteration we always have that the set $\bC'$ computed by \refinecoords{} on line 2(a) is a subset of $\bC_i$ by the output guarantee of \refinecoords{}. Moreover, since the pair we are adding is of the form $(\bC', \bI_{i} \cup \bI')$, we have that $\bI = \bI_{i} \cup \bI' \supseteq \bI_i$.
So taking $\bC_{i+1}=\bC$ and $\bI_{i+1}=\bI$, we have the new sequence of inclusions $[k'] := \bC_{-2} =: \bC_{-1} \supseteq \bC_0 \supseteq \bC_1 \supseteq \bC_2 \supseteq \dots \supseteq \bC_{\min\{i+1,\ell\}}$ and $\bI_{\min\{i+1,\ell\}} \supseteq \bI_{i-1} \supseteq \dots \supseteq \bI_{-1} := \emptyset$ as required, and we inherit that (a), (b) and (c) hold for all $-1 \leq j \leq \ell$ from the inductive hypothesis.

So we assume that in the inductive hypothesis we cannot take $\ell \leq i$, i.e.~we must have $\ell > i$. This implies that $|\bC_i \cap \overline{R}| \geq 1$ (because if $\bC_i \cap \overline{R} = \emptyset$, we could have taken $\ell=i$).
 Now again consider the execution of the loop on line $(2)$ of the $(i+1)$st recursive call corresponding to $(\bC_i, \bI_i)$. By \Cref{lem:good-pair}, it follows that $\refinecoords(g_{\text{avg}}^{\bI_{i}}, \bC_i)$ outputs a pair $(\bC', \bI')$ such that 
$\bI' \subseteq \overline{R}$, $\bC' \subseteq {\bC_i}$, 
\[| \bC' \cap \overline{R} | \le \frac{|\bC_i \cap \overline{R}|}{2}, \]
and
\[\sum_{S: |S \cap \bC_i \setminus \bC'| \geq k^{2/3} \polylog(k'/\eps)} \wh{g_{\text{avg}}^{\bI_{i} \cup {\bI'}}}^2(S) \leq \frac{{4}\eps^2}{(k')^2}. \]
It thus would suffice for us to take $\bC_{i+1} = \bC'$ and $\bI_{i+1} = \bI_i \cup \bI'$ if we can choose $\ell = i +2$. If we cannot choose $\ell = i+2$, then we must have that $i + 1 = \log(k')$. In this case, it follows that
	\[| \bC_{i+1} \cap \overline{R} | \le \frac{|\bC_i \cap \overline{R}|}{2} \leq \frac{|\overline{R}|}{2^{i+1}} \leq \frac{|\overline{R}|}{2k'} < 1.
	\]
So $| \bC_{i+1} \cap \overline{R} | = 0$ and we can take $\ell = i+1$. We have the desired inclusions of the $\bC_j$'s and the $\bI_j$'s; the arguments above establish that (a), (b) and (c) all hold for $j=\ell=i+1$; and we inherit (a), (b), (c) holding for smaller values of $j$ from the inductive case.  So we have achieved the inductive claim for $\calF_{i+1}$.

Finally, note that by a union bound, both the inductive step and our call to \refinecoords{} corresponding to the pair $(\bC_i, \bI_i)$ succeed with probability $1 - (i+1) \frac{10}{(k')^2}$, completing the inductive step and the proof of \Cref{claim:inductive-found-coords}.\qedhere

\end{proof}

This concludes the proof of \Cref{lem:found-high-coords}.
\end{proof}

We will also need some simple bounds on the size of the family of pairs that we output and the number of oracle calls made to the randomized algorithm computing $f$. 

\begin{lemma}
\label{lem:find-high-coords-family-bound}
Let $f: \bits^{k'} \rightarrow [-1,1]$ and $\eps \in [0,1]$, then $\findhighcoords(f,\eps)$ outputs a family $\calF$ of size at most $2^{k^{1/3} \polylog(k'/\eps)}$.
\end{lemma}

\begin{proof}
Let $\calF_i'$ denote the set $\calF'$ computed in the loop on line $2$ of the $i$th recursive call of \findhighcoordsrec{}. Since each call to \refinecoords{} in line 2(a) of \findhighcoordsrec{} adds $2^{k^{1/3} \polylog(k'/\eps)}$ pairs to $\calF$ by \Cref{lem:good-pair}, it follows that
	\[|\calF'_{i+1}| \leq 2^{k^{1/3} \polylog(k'/\eps)} \cdot |\calF'_i|. \] 
Since we have
	\[|\calF_0| \leq 2^{k^{1/3} \polylog(k'/\eps)}, \]
it follows that $\findhighcoords(f,\eps)$ outputs at most 
	\[|\{[k'], \emptyset\}| + \sum_{i=0}^{\log(k')} |\calF_i| \leq 1 + \sum_{i=0}^{\log(k')} \left( 2^{k^{1/3} \polylog(k'/\eps)} \right)^{i+1} \leq 2^{k^{1/3} \polylog(k'/\eps)}
	\]
pairs as desired. 
\end{proof}

\begin{lemma}
\label{lem:find-high-coord-query-bound}
\findhighcoords${(f,\eps)}$ makes at most {$2^{k^{1/3} \polylog(k'/\eps)}$} queries to $\calA$, the $1$-bounded randomized algorithm for $f$.
\end{lemma}

\begin{proof}
Recall that for any $I \subseteq [k']$ by \Cref{lem:alg-for-average} and \Cref{lem:alg-for-noise} it follows that there is a randomized $1$-bounded algorithm for $g_{\text{ave}}^I$ that only needs to query $\calA$ once.

Let $\calF_i'$ denote the set $\calF'$ computed in the loop on line $2$ of the $i$th recursive call of \findhighcoordsrec{}. Note that in the $(i+1)$st iteration, all queries to $f$ will come from calling \refinecoords{} in line 2(a) of \findhighcoordsrec{}. It follows from \Cref{lem:refine-coords-query-bound} that the $(i+1)$st recursive call makes 
	\[|\calF_i| \cdot 2^{k^{1/3} \polylog(k'/\eps)} \]
calls to $\calA$, where $|\calF_{-1}| = 1$. Since $\calF_i$ is a subset of the output of \\
$\findhighcoords(f, \eps)$ it then follows that $|\calF_i| \leq 2^{k^{1/3} \polylog(k'/\eps)}$. Combining this with the fact that there are at most $\log(k')+1$ rounds then yields that we make at most $2^{k^{1/3} \polylog(k'/\eps)}$ calls to $\calA$ as desired.
\end{proof}


\subsection{Putting It All Together}

\subsubsection{Parameters for the Main Algorithm}
As before, we will need several parameters for our main algorithm. Throughout the section, we will write these parameters in terms of $k'$, which will correspond to the number of coordinate oracles given by \Cref{thm:coordinate-oracles-exist}. 

We now set $\ell$ to correspond to the set size limit in the summation limits of $(iii)$ from \Cref{lem:found-high-coords}. In particular, this gives us that $\ell = k^{2/3} \polylog(k'/\eps)$.

We will also need a variety of parameters for our local estimators. In particular, we will set these to all be identical to their counterparts in \Cref{sec:local-est-junta-corr}. That is, we choose $\kappa, \Delta$ as
			\[
				\kappa=10\log\left(\frac{k'}{\eps}\right)\quad\text{and}\quad \Delta=10r \log\left(\frac{k'}{\eps}\right).
			\]
respectively where
	\[L = \kappa \ell, \quad \tau = \frac{\eps}{10}, \quad \text{ and } \quad r = \Theta(\sqrt{L \log L \log(1/\tau)})\]
and set
$$
N:={k'}^{O(r)}\cdot 2^{O(\kappa\Delta)}\cdot \frac{k'^2}{\eps^2}\le \exp\left(k^{1/3} \cdot \polylog\left(\frac{k'}{\eps}\right)\right).
$$

Because these parameters are set identically (as functions of $k'$, $\eps$, and $\ell$) as in \Cref{sec:local-est-junta-corr}, we have that {\Cref{lem:juntaestlemma} still holds for any function $g: \bits^{k'} \rightarrow \bits$ with a $1$-bounded randomized algorithm computing it with our new choice of $\ell$.

\subsubsection{The Testing Algorithm}
With the above setup in place, we can finally state our classical testing algorithm.

\bigskip

\begin{algorithm}[H]
\addtolength\linewidth{-2em}

\vspace{0.5em}

\textbf{Input:} A boolean function $f \isafunc$, an integer $k$, and a parameter $\eps \in [0,1]$ \\[0.25em]
\textbf{Output:} A estimate for $\corr(f, \calJ_k)$
\vspace{0.5em}

\tester:

\vspace{0.5em}

\begin{enumerate}	
	\item Set $\gamma = 0$
	\item \label{line:compute-oracles} Compute coordinate oracles $\{\calO_1, \calO_2, \dots, \calO_{k'} \}$ by invoking  \Cref{thm:coordinate-oracles-exist} on $f$ with accuracy $\eps$ and failure probability $\frac{1}{k}$. Let $k' = \poly(k,\eps^{-1})$ denote the number of oracles.
	\item \label{line:oracles-alg} Let $g: \bits^{k'} \rightarrow [-1,1]$ be the function given by $g(y) = \E_{\bx} \left [f(\bx) \bigg | \calO_1(\bx) = y_1, \calO_2(\bx) = y_2, \dots, \calO_{k'}(\bx) = y_{k'} \right ]$, which has a a $1$-bounded randomized algorithm, which we call $\calA_g$, by \Cref{thm:alg-for-coordinate-oracle-avg}.
	\item \label{line:find-high-coords} $\calF \gets \findhighcoords(g, \eps^2)$
	\item \label{line:tester-family-loop} For $(\bC,\bI) \in \calF$:
	\begin{enumerate}
		\item $\gamma' = 0$, $A = \emptyset$
		\item Define $\left (g_{\ave}^\bI \right)^{\overline{\bC}} \gets \shnoise_{\ell, \kappa, \Delta}^{\overline{\bC}}(g_{\ave}^\bI)$.
		\item Sample $\bx^{(1)}, \bx^{(2)}, \dots \bx^{(N)} \sim \{\pm 1\}^{k'}$ uniformly and independently at random.
		\item For each $\bx^{(i)}$ draw a sample bundle $\calB^{(i)} \sim \calD_{\bx^{(i)}, \bC}$.
		\item \label{line:compute-local-estimators-classic} For each set $U \subseteq [k']$ of size $k$ such that $U \cap \bI = \emptyset$ and $U \supseteq \bC$, run the algorithm of \Cref{lem:juntaestlemma} using $\bx^{(1)}, \dots, \bx^{(N)}, \calB^{(1)}, \dots, \calB^{(N)}$ to get an $(\eps/2)$-estimation $\boldEst_{\bC,U}$ of $\corr \left( \left( g_{\ave}^\bI \right)^{\overline{\bC}}, \calJ_U \right)$; if $\boldEst_{\bC,U} > \gamma'$, update $\gamma' = \boldEst_{\bC,U}$ and $A \gets U$.
		\item \label{line:check-true-junta-corr} Let $\boldEst_A$ be an estimate of $\corr(g,\calJ_A)$ to accuracy $\pm \eps/20$ with failure probability $2^{-{\Omega((k')^2)}}$ 
		and update $\gamma \gets \max(\gamma, \boldEst_A)$. 
	\end{enumerate}
	\item Return $\gamma$
\end{enumerate}
\caption{The overall classical tolerant junta tester}
\label{alg:classical}
\end{algorithm}

\bigskip

We give some high level intuition.
Let $U^\star$ 
denote the set of $k$ coordinates that maximize the junta correlation. At a high level the algorithm begins by building coordinate oracles so as to reduce the problem to a question about $g$, which is a real-valued function over $\bits^{k'}$ where $k' = \poly(k, 1/\eps)$, rather than $f$, which is a Boolean-valued function on $n$ coordinates. Next, we call $\findhighcoords(g,{\eps^2})$ with the goal of finding a pair $(\bC^\star,\bI^\star)$ such that $\bC^\star$ contains all coordinates in $U^\star$
that appear at level $\ell$ and above and $\bI^\star \subseteq \overline{U^\star}$. 
For this particular pair $(\bC^\star, \bI^\star)$, we will show that for every set $U$ that we produce an estimate for in line 5(e), we have that with high probability $\boldEst_{\bC^\star, U}$  will be within $\eps$ of $\corr(g, \calJ_U)$. As a result we will correctly update $\gamma$ on line 5(f) with an accurate estimate of the true junta correlation of $g$, which by the properties of the coordinate oracles will roughly equal the junta correlation of $f$. On the other hand, we will not have such guarantees about our estimates of the junta correlation arising from pairs $(\bC, \bI) \in \calF$ not equal to $(\bC^\star, \bI^\star)$. It is precisely for this reason that we will require line 5(f) to check the correlation between $g$ and the best junta over the set of variables $A$ identified in line 5(e), so as to ensure that we do not spuriously update $\gamma$ with a false estimate that is higher than the true junta correlation.

\subsubsection{Efficiency}
In this section, we will prove the following lemma:

\begin{lemma}
\label{lem:classical-tester-query-bound}
For $f: \bits^n \rightarrow \bits$ and $k$ a positive integer,  $\tester(f,k,\eps)$ make at most $2^{k^{1/3} \polylog(k/\eps)}$ queries to $f$ in expectation.
\end{lemma}

Before proving \Cref{lem:classical-tester-query-bound}, we first show that we can efficiently implement line 5(f).

\begin{lemma}
\label{lem:query-bound-estimate-junta-corr}
The junta correlation $\corr(g, \calJ_A)$ on line \ref{line:check-true-junta-corr} can be {$\pm \eps/20$}-approximated, with failure probability $2^{-\Omega((k')^2)}$, using $\poly(k/\eps)$ calls to the randomized algorithm {$\calA_{g}$} for $g$.	
\end{lemma}

\begin{proof}
First, recall from \Cref{eq:corfJC} that
\begin{equation}
\label{eq:corrhappy}
\corr(g, \calJ_A) = \Ex_{\bx} \left[ \left| g_{\ave}^{\overline{A}}(\bx) \right | \right ],
\end{equation}
and that by \Cref{eq:corr} the function $g^{\overline{A}}_{\ave}$ is $1$-bounded.  Hence by \Cref{lem:alg-for-average}, there is a $1$-bounded randomized algorithm {$\calA'_g$} for $g_{\ave}^{\overline{A}}$ {that makes a single call to $\calA_g$}. 

Our approach to approximating $\corr(g, \calJ_A)$ is as follows:  we (1) first sample $\by^1, \dots, \by^{(k')^2/\eps^2} \sim \bits^{k'}$; then (2) compute, for each $i$, a value $\bz_i(y_i)$ that is a $\pm \eps/100$-accurate estimate (with failure probability $2^{-(k')^2}$) of 
${|g_{\ave}^{\overline{A}}(\by^i)|}$ 
using \Cref{claim:estimation-of-bounded-function};
 and (3) finally output
	\[\bZ := \frac{\eps^2}{(k')^2} \sum_{i =1}^{(k')^2/\eps^2} \bz_i \]
as our approximation. Let $\calA''_{g}$ denote the (randomized) algorithm used to estimate $\bz_i(\by^i)$ in step~(2) above.

Let us argue correctness for this approximation.  Since for every possible outcome $y_i$ of $\by_i$ the random variable $\bz_i(y_i)$ is bounded in $[-1,1],$ by \Cref{claim:estimation-of-bounded-function} we have that
\[
\left | \Ex_{\calA''_g}[\bz_i(y_i)] - \left| g_{\ave}^{\overline{A}}(y^i) \right | \right| \leq \frac{\eps}{100} + 2 \cdot 2^{-(k')^2} \leq \frac{\eps}{50},
\]
from which it follows that
\begin{equation}
\label{eq:pudding}
\left | \Ex_{\by^i,\calA''_g}[\bz_i(\by_i)] - \overbrace{\Ex_{\by^i}\left[\left| g_{\ave}^{\overline{A}}(\by^i) \right | \right]}^{=\corr(g,\calJ_A) \text{~by (\ref{eq:corrhappy})}} \right| \leq\Ex_{\by^i}\left[\left | \Ex_{\calA''_g}[\bz_i(\by_i)] - \left| g_{\ave}^{\overline{A}}(\by^i) \right | \right| \right] \leq  \frac{\eps}{50}.
\end{equation}
 
We apply a Hoeffding bound to $\bZ$ (a sum of independent $1$-bounded random $\bz_i(\by^i)$'s) to infer that with failure probability at most $2^{-\Omega((k')^2)}$, we have
\begin{equation} \label{eq:treacle}
\left|
\bZ - \Ex_{\by_i,\calA''_g}[\bz_i(\by^i)] \right| \leq {\frac \eps {50}}.
\end{equation}

By the triangle inequality, \Cref{eq:pudding} and \Cref{eq:treacle} give that $|\bZ - \corr(g,\calJ_A)| \leq \eps/20$ except with failure probability $2^{-\Omega((k')^2)}$, as desired.

It remains to bound the number of queries this procedure makes. This is straightforward: each invocation of \Cref{claim:estimation-of-bounded-function} requires $\poly(k',{1/\eps})$ calls to $\calA_g$ and we make $(k')^2/\eps^2$ such invocations. Thus we make $\poly(k',1/\eps) = \poly(k/\eps)$ calls in total, as claimed.
\end{proof}

With \Cref{lem:query-bound-estimate-junta-corr} in hand, we now prove \Cref{lem:classical-tester-query-bound}. 

\begin{proof}[Proof of \Cref{lem:classical-tester-query-bound}]
	We begin by noting that by \Cref{thm:coordinate-oracles-exist}, \tester~makes at most $\poly(k, \eps^{-1})$ queries in line \ref{line:compute-oracles}.
	
	We now bound the number of queries made to $g$ in the rest of the algorithm.
	By \Cref{lem:find-high-coord-query-bound}, we have that the call to $\findhighcoords(g,\eps^2)$ in line \ref{line:find-high-coords}, makes at most $2^{k^{1/3} \polylog(k'/\eps)}$ queries to $g$. 

	Now fix a particular iteration of the loop on line \ref{line:tester-family-loop}, i.e.~a particular $(\bC,\bI)$ pair. Note that using the local estimators to compute the junta correlation for all of the sets $U$ in line \ref{line:compute-local-estimators-classic} 
	can be implemented using $2^{k^{1/3} \polylog(k'/\eps)}$ queries to $g_{\ave}^I$ by \Cref{lem:juntaestlemma}. (Similar to before, this crucially uses the fact that we can reuse the sample bundles for all $k$-subsets $U$ of $[k']$ such that $U \cap \bI = \emptyset$ and $U \supseteq \bC$.) In turn this corresponds to $2^{k^{1/3} \polylog(k'/\eps)}$ queries to $g$ by applying \Cref{lem:alg-for-average}. Finally, estimating $\corr(g, \calJ_A)$ on line 
\ref{line:check-true-junta-corr} 
	makes at most $\poly(k/\eps)$ queries by \Cref{lem:query-bound-estimate-junta-corr}.  Since we repeat the loop on line 
	\ref{line:tester-family-loop} 
	at most $2^{k^{1/3} \polylog(k'/\eps)}$ times by \Cref{lem:find-high-coords-family-bound}, we conclude that we make at most $2^{k^{1/3} \polylog(k'/\eps)}$ queries to $g$ throughout this loop. 
	
\Cref{lem:classical-tester-query-bound} now follows because each query to $g$ makes at most $\poly(k')$ queries to $f$ in expectation, by \Cref{thm:alg-for-coordinate-oracle-avg}. In particular, applying linearity of expectation gives that we make at most 
	\[2^{k^{1/3} \polylog(k'/\eps)} = 2^{k^{1/3} \polylog(k/\eps)}\]
	 queries to $f$ in expectation, as desired, where the equality used the fact that $k' = \poly(k,\eps^{-1})$.
\end{proof}

\subsubsection{Correctness}
We now turn to prove that $\gamma$ will indeed approximate $\corr(f, \calJ_k)$ up to $O(\eps)$ error. We will prove this in two parts. For the first part, we show that $\gamma$ cannot significantly overestimate the junta correlation:

\begin{lemma}
\label{lem:classical-tester-corr-small}
With probability $1 - o(1)$,
	\[\gamma \leq \corr(f, \calJ_k) + \eps.\]
\end{lemma}

\begin{proof}
	We condition on the event that the coordinate oracles that we compute in line \ref{line:compute-oracles} outputs a set of oracles satisfying \Cref{thm:coordinate-oracles-exist}. We further condition on the event that all of our estimates on line \ref{line:check-true-junta-corr} have error at most $\eps/20$. By \Cref{lem:query-bound-estimate-junta-corr} and a union bound over all elements of $\calF$, these events both happen with probability at least 
		\[1 - \frac{1}{k} - |\calF| 2^{-{\Omega((k')^2})} = 1 - o(1),\]
	where we also used the fact that $|\calF| \leq 2^{k^{1/3} \polylog(k'/\eps)}$, which follows from \Cref{lem:find-high-coords-family-bound}.
	
	Now let $\calS \subseteq [n]$ denote the set of coordinates corresponding to our coordinate oracles. We now note that since $g = \fave^{\overline{\calS}}$, we get that in line \ref{line:check-true-junta-corr} we always have
		\[\corr(g,\calJ_A) = \corr(\fave^{\overline{\calS}},\calJ_A) = \corr(f,\calJ_A) \leq \corr(f, \calJ_k) \]
	where the second equality used the fact that $A \subseteq \calS$. It thus follows that we always have 
		\[\boldEst_A \leq \corr(f, \calJ_k) + \frac{\eps}{20} \]
	whenever we run line \ref{line:check-true-junta-corr}, and consequently $\gamma \leq \corr(f, \calJ_k) + \frac{\eps}{20}$ as desired.
\end{proof}

It remains to bound $\gamma$ in the other direction:

\begin{lemma}
\label{lem:classical-tester-corr-big}
With probability $1 - o(1)$, we have that
	\[\gamma \geq \corr(f, \calJ_k) - 4\eps.\]
\end{lemma}

\begin{proof}
We will assume throughout the proof that we successfully created the coordinate oracles on line \ref{line:compute-oracles} of \tester{} (note that this happens with $1-o(1)$ probability), and we let $\calS \subseteq [n]$ denote the corresponding set of coordinates. 

We now choose $R = U^\star \subseteq [k']$ to be a set of size $k$ that achieves $\corr(g, \calJ_R) = \corr(g, \calJ_k)$. Applying \Cref{lem:found-high-coords}, we get that with $1-o(1)$ probability there exists a pair $(\bC^\star, \bI^\star) \in \calF$ such that $\bI^\star \subseteq \overline{R}$, $\bC^\star \subseteq R$, and 
\begin{equation}
\label{eq:hehaw}
\sum\limits_{\substack{ S \subseteq [k']: \\ |S \setminus \bC^\star| \geq \ell, \\ |S| \leq k}} \wh{g_{\ave}^{\bI_\star}}^2(S) \leq \frac{\eps^2}{100} 	.
\end{equation}

Our aim is to show that in the iteration of the loop on line \ref{line:tester-family-loop} corresponding to $(\bC^\star, \bI^\star)$, we will update $\gamma$ to be large. Towards this end, we assume that the conclusion of \Cref{lem:juntaestlemma} holds for every set $U$ that we compute in this iteration of line \ref{line:tester-family-loop}, that is
\begin{equation}
\label{eq:closey-close}
\left | \boldEst_{\bC^\star,U} - \corr\left( \left( g_{\ave}^{\bI_\star}\right)^{\overline{\bC^\star}}, \calJ_{{U}} \right) \right| \leq \frac{\eps}{2}.
\end{equation}
Note that this also holds with $1-o(1)$ probability, as we output a bad estimate for each individual $U$ with probability at most $2^{-(k'/\eps)^2}$ by \Cref{lem:juntaestlemma}, and there are at most $2^{k'}$ sets $U$. 

We will require the following claim:

\begin{claim}
\label{claim:sharp-noise-no-change-corr}
For any $k$-subset $U \subseteq [k']$ satisfying $U \supseteq \bC^\star$ and $U \cap \bI^\star = \emptyset$, we have that
	\[\left| \corr\left( \left( g_{\ave}^{\bI_\star}\right)^{\overline{\bC^\star}}, \calJ_{U} \right) - \corr\left(g , \calJ_{U} \right) \right| \leq \frac{\eps}{2}.\]
\end{claim}

Before proving the claim, we first use it to finish the proof of \Cref{lem:classical-tester-corr-big}. Together with \Cref{eq:closey-close}, we get that for all $k$-subsets $U$ of $[k']$ satisfying $U \supseteq \bC^\star$ and $U \cap \bI^\star = \emptyset$, we have
\begin{equation}
\label{eq:allU}
\left | \boldEst_{\bC^\star,U} - \corr\left(g , \calJ_{U} \right) \right| \leq \eps;
\end{equation}
specializing to $U=R$, we get that
\[\left | \boldEst_{\bC^\star,R} - \corr\left(g , \calJ_{{R}} \right) \right| \leq \eps.\]
So it follows that after the completion of line \ref{line:compute-local-estimators-classic}, we have
	\[\gamma' \geq \boldEst_{\bC^\star, R} \geq \corr(g, \calJ_R) - \eps = \corr(g, \calJ_k) - \eps. \]
If $A$ is the set used on line \ref{line:check-true-junta-corr}, it follows from \Cref{eq:allU} and the above inequality that
	\[\corr(g, \calJ_A) + \eps \geq \boldEst_{\bC^\star,A} = \gamma' \geq \corr(g, \calJ_k) - \eps. \]
So we conclude that 
	\[\corr(g, \calJ_A) \geq \corr(g, \calJ_k) - 2 \eps. \]
Thus assuming that the error in the estimate on line \ref{line:check-true-junta-corr} is indeed at most $\eps/20$ (which happens with overwhelming probability), we have that
	\[\gamma \geq \boldEst_A \geq \corr(g, \calJ_k) - 3 \eps. \]

Finally, since we successfully created the coordinate oracles, we have that
	\[
	\corr(g, \calJ_k) = \max_{S \subseteq \calS: |S| \leq k} \corr(f_{\ave}^{\overline{\calS}}, \calJ_S) = \max_{S \subseteq \calS: |S| \leq k} \corr(f, \calJ_S) \geq \max_{S \subseteq [n]: |S| \leq k} \corr(f, \calJ_S) - \eps\]
where the final step used property $(2)$ from \Cref{thm:coordinate-oracles-exist}. 
Thus, we get that 
	\[\gamma \geq \corr(f, \calJ_k) - 4 \eps \]
with probability $1-o(1)$, as desired. 
	
So it only remains to prove the claim:

\begin{proof}[Proof of \Cref{claim:sharp-noise-no-change-corr}]
The proof is almost identical to the proof of \Cref{lem:every-set-accurate}.
Fix any $U$ as specified in \Cref{claim:sharp-noise-no-change-corr}.
We run the argument of \Cref{lem:every-set-accurate}'s proof starting at \Cref{eq:bokchoi}, now with $\corr(g,\calJ_U)$ in place of $\corr(f,\calJ_U)$ and with $\corr\left( \left( g_{\ave}^{\bI_\star}\right)^{\overline{\bC^\star}}, \calJ_{{U}}\right)$ in place of $\corr\left(f^{\overline{C}},\calJ_U\right)$. That argument gives us that
\begin{equation} \label{eq:gum2}
\left| \corr\left( \left( g_{\ave}^{\bI_\star}\right)^{\overline{\bC^\star}}, \calJ_{{U}} \right) - \corr\left(g , \calJ_{{U}} \right) \right|
\leq
\sqrt{\sum_{S \subseteq U} \left( \wh{\left( g_{\ave}^{\bI_\star} \right)^{\overline{\bC^\star}}}(S) - \wh{g}(S) \right)^2},
\end{equation}
which is exactly analogous to \Cref{eq:gum} in the earlier argument.
As in the earlier argument, we now divide the sum inside the square root into two parts:
\begin{flushleft}\begin{enumerate}
\item $S \subseteq U$ with $|S\setminus \bC^\star|=|S\cap \overline{\bC^\star}|\le \ell$:
For this part, exactly as in the earlier argument, we have from \Cref{lem:sharp-noise} that
$$
\sum_{S \subseteq U: |S \setminus \bC^\star| \leq \ell} \left( \wh{\left( g_{\ave}^{\bI_\star} \right)^{\overline{\bC^\star}}}(S) - \wh{g}(S) \right)^2 \le \sum_{S \subseteq U: |U \setminus \bC^\star| \leq \ell} \widehat{g}(S)^2\cdot (\Delta 2^{-\kappa})^2
\le  {(\Delta 2^{-\kappa})^2 \le \frac{\eps^2}{20}}
$$
by our choices of $\kappa$ and $\Delta$.
\item $S \subseteq U$ with $|S\setminus \bC^\star|>\ell$:
For this part we have from \Cref{lem:sharp-noise} and \Cref{eq:hehaw} that 
$$
\sum_{S \subseteq U: |S \setminus \bC^\star| > \ell}  \left( \wh{\left( g_{\ave}^{\bI_\star} \right)^{\overline{\bC^\star}}}(S) - \wh{g}(S) \right)^2 \le \sum_{S \subseteq U: |S \setminus \bC^\star| > \ell}  \wh{g}^2(S) = \sum_{S \subseteq U: |S \setminus \bC^\star| > \ell}  \wh{g_{\ave}^{\bI_\star}}^2(S)
\le \frac{\eps^2}{100}
$$
where the equality step crucially used that $U \cap \bI^\star = \emptyset$ and the final inequality used that we only sum over sets with $|S| \leq |U| \leq k$ (in this inequality \Cref{eq:assump} plays a role analogous to \Cref{eq:nevertoolarge} in the earlier argument).
\end{enumerate}\end{flushleft}

Combining the bounds on these two parts, we get that
\[\left| \corr\left( \left( g_{\ave}^{\bI_\star}\right)^{\overline{\bC^\star}}, \calJ_{k} \right) - \corr\left(g , \calJ_{k} \right) \right| \leq \frac{\eps}{2}\]
as desired, and \Cref{claim:sharp-noise-no-change-corr} is proved.
\end{proof}
This concludes the proof of \Cref{lem:classical-tester-corr-big}.
\end{proof}

\subsubsection{Proof of \Cref{thm:classical}}
We can finally reap the benefits of our labor and prove \Cref{thm:classical}.

\begin{proof}[Proof of \Cref{thm:classical}]
	We run $\tester(f, k, \eps/2)$ to compute a value $\gamma$. If the number of queries it makes exceeds  $k \times ($ the bound on the expected number of calls to $f$ given in \Cref{lem:classical-tester-query-bound}$)$,  then we abort the execution of the algorithm and set $\gamma = 0$. We then output $\frac{1}{2} (1 - \gamma)$ as our estimate of $\dist(f, \calJ_k)$.
	
	By \Cref{lem:classical-tester-query-bound}, we always make at most $2^{k^{1/3} \polylog(k/\eps)}$ queries. By Markov's inequality, the probability that we abort is at most $1/k$. Assuming we do not abort, by \Cref{lem:classical-tester-corr-small} and \Cref{lem:classical-tester-corr-big}, we have that
\begin{equation}
\label{eq:accurate}\left| \gamma - \corr(f, \calJ_k) \right| \leq 2\eps
\end{equation}
with probability $1 - o(1)$. So assuming that we do not abort and that \Cref{eq:accurate} holds (which happens with overall probability $1- o(1)$), we correctly estimate the junta distance $\dist(f,\calJ_k)$  of $f$ to within additive error $\pm \eps$, where we used the fact that
	\[\dist(f, \calJ_k) = \frac{1}{2} (1 - \corr(f, \calJ_k)). \]
This completes the proof of the theorem.
\end{proof}

\section*{Acknowledgements}
X.C. is supported by NSF grants IIS-1838154, CCF-2106429, and CCF-2107187. S.P. is supported by NSF grants CCF-2106429, CCF-2107187, CCF-2218677, ONR grant ONR-13533312, and a NSF Graduate Student Fellowship. R.A.S. is supported in part by NSF awards CCF-2106429 and CCF-2211238.
\begin{flushleft}
\bibliographystyle{alpha}
\bibliography{allrefs}
\end{flushleft}

\appendix

\end{document}